\documentclass[twocolumn,10pt]{article}

\usepackage{arxiv}

\usepackage[utf8]{inputenc} 
\usepackage[T1]{fontenc}    

\usepackage{url}            

\usepackage{amsfonts}       
\usepackage{nicefrac}       
\usepackage{microtype}      
\usepackage{graphicx}
\usepackage[authoryear,sort]{natbib}
\setcitestyle{authoryear,comma}  

\usepackage{amssymb}
\usepackage{amsmath}
\usepackage{amsthm}
\usepackage{tabularx}
\usepackage{booktabs} 
\usepackage{float}
\usepackage{makecell}
\usepackage{rotating}
\usepackage{xcolor} 
\usepackage[normalem]{ulem}
\usepackage{siunitx}
\usepackage{mathrsfs}  
\usepackage{mathtools}
\usepackage{graphicx}
\usepackage{array}
\usepackage[linktoc=page,
colorlinks=true,	
linkcolor=blue,
citecolor=blue,
urlcolor=blue]{hyperref}       

\usepackage{doi}

\usepackage{newtxtext}
\usepackage{newtxmath}
\usepackage{upgreek}
\usepackage{enumitem}
\usepackage{silence}
\theoremstyle{definition}
\newtheorem{lemma}{Lemma}
\newtheorem{proposition}{Proposition}

\newtheorem{remark}{Remark}
\newtheorem{corollary}{Corollary}

\newcommand{\Set}[1]{\mathcal{#1}} 
\newcommand{\Setnum}[1]{\mathbb{#1}} 
\newcommand{\Group}[1]{\mathcal{#1}} 

\newcommand{\ve}[1]{\boldsymbol{#1}} 
\newcommand{\te}[1]{\mathbf {#1}} 
\newcommand{\V}[1]{\uline{\boldsymbol{#1}}} 
\newcommand{\M}[1]{\uuline{\boldsymbol{#1}}} 
\newcommand{\nt}{n_\textrm{steps}} 
\newcommand{\nn}{n_{\textrm{nodes}}} 
\newcommand{\nele}{n_\textrm{el}} 
\newcommand{\niter}{n_\textrm{iter}} 
\newcommand{\nqp}{n_{\textrm{qp}}} 
\newcommand{\nf}{n_{\textrm{force}}} 
\newcommand{\tvar}[1]{{}^{s}\!{#1}}
\newcommand{\tsvar}[1]{{}^{s}\hspace{-1pt}{#1}} 
\newcommand{\tivar}[1]{{}^{s,i}\!{#1}}

\newcommand{\ts}[2]{\mathcal{L}^{#1}_{#2}} 
\newcommand{\tss}[1]{\mathcal{L}^{#1}_{\textrm{sym}}} 
\newcommand{\tsa}[1]{\mathcal{L}^{#1}_{\textrm{skw}}} 
\newcommand{\hms}[1]{\mathbb{R}^{#1}} 
\newcommand{\tr}{\operatorname{tr}} 
\newcommand{\sym}{\operatorname{sym}} 
\newcommand{\cl}{\operatorname{cl}} 
\newcommand{\diffp}[2]{\frac{\partial #1}{\partial #2}} 
\newcommand{\diff}[2]{\frac{\textrm{d} #1}{\textrm{d} #2}} 
\newcommand{\point}{\text{ .}}
\newcommand{\comma}{\text{ ,}}
\newcommand{\diss}{\varPhi^*}

\newcommand{\hfeeq}{\varPsi^{\textrm{eq}}}
\newcommand{\hfeneq}{\varPsi^{\textrm{neq}}}
\newcommand{\hfeeqpann}{\varPsi^{\textrm{eq,PANN}}}
\newcommand{\hfeeqnn}{\varPsi^{\textrm{eq,NN}}}
\newcommand{\hfeneqpann}{\varPsi^{\textrm{neq,PANN}}}
\newcommand{\hfeneqnn}{\varPsi^{\textrm{neq,NN}}}
\newcommand{\disspann}{\varPhi^{*\textrm{,PANN}}}
\newcommand{\dissnn}{\varPhi^{*\textrm{,NN}}}

\newcommand{\Free}{\mathcal{FREE}} 
\newcommand{\Force}{\mathcal{FORCE}} 

\title{Calibration of neural viscoelastic models via full-field data}

\author{
	Brain M. Riemer\\
	Institute of Solid Mechanics\\
	TU Dresden,
	01062 Dresden, Germany \\
	\And
	Markus K\"{a}stner\\
	Institute of Solid Mechanics\\
	TU Dresden, 
	01062 Dresden, Germany \\
    \And
    Karl A. Kalina$^{*}$\\
	Institute of Solid Mechanics\\
	TU Dresden,
	01062 Dresden, Germany \\
}

\renewcommand{\shorttitle}{Calibration of neural viscoelastic models via full-field data}

\hypersetup{
	pdftitle={Preprint\_RiemerEtAl\_2026b},
	pdfsubject={},
	pdfauthor={RiemerEtAl},
	pdfkeywords={}
}

\theoremstyle{definition}

\usepackage[font=small]{caption} 

\begin{document}
\twocolumn[
\begin{@twocolumnfalse}
    \maketitle
    \begin{abstract}
        We propose an unsupervised learning framework for calibrating a physics-augmented neural network (PANN) for small-strain viscoelasticity via full-field data. It only requires quantities that are directly accessible in real experiments for training, namely global reaction forces and surface displacements. The underlying PANN is embedded in the generalized standard materials theory, in which two scalar-valued potentials render the constitutive model thermodynamically consistent by construction, while invariant-based representations of the free energy and the dual dissipation potential additionally ensure material symmetry. Considering a thin specimen under the plane stress assumption, we formulate a constrained optimization problem based on the equilibrium gap method in combination with quasi-Newton optimizers and automatic differentiation. Thereby, the unknown out-of-plane strain follows from the plane stress condition and the evolution of the internal variables is captured by an implicit time integration scheme. The resulting system of nonlinear equations is solved via a local Newton iteration at quadrature point and time step. To drastically reduce the computational cost of training, the backward adjoint method is employed to compute the gradient of the target loss, instead of backpropagating through all Newton iteration steps. 
        The proposed framework is demonstrated for synthetic data, including noisy displacements and forces, showing excellent agreement across a wide range of deformation rates and load paths.
    \end{abstract}
    \keywords{viscoelasticity \and generalized standard materials \and physics-augmented neural networks \and full-field data \and equilibrium gap method}	
    \vspace{0.5cm}
\end{@twocolumnfalse}
]
{
\renewcommand{\thefootnote}{}
\footnotetext{$^{*}$Corresponding author, email: \texttt{karl.kalina@tu-dresden.de}.}
}
\section{Introduction}
\label{sec:intro}
Over the past century, extensive research has established the physical and mathematical principles that constitutive models should obey \citep{silhavy_mechanics_1997,Haupt2000,holzapfel_nonlinear_2000}, giving rise to a wide range of so-called \emph{classical constitutive models}. However, when applied to materials that exhibit highly nonlinear behavior, such models may be too restrictive to accurately fit the data and may need to be modified for each new set of experimental data. To address these shortcomings, \emph{machine learning} approaches---particularly \emph{neural networks (NNs)}---have emerged as powerful tools for constitutive modeling \citep{Dornheim2023,Fuhg2024b}, offering the flexibility to capture complex material responses while automating the modeling process.

In the following, we first give an overview of neural constitutive modeling for 
both elasticity and inelasticity. Building on this, we then review the most common 
strategies for calibrating such data-driven material models from full-field data. 
\subsection{Literature review on neural constitutive modeling}
\label{ssec:nn_review}
In their seminal work from the early 1990s, \cite{Ghaboussi1991} 
were the first to employ NNs as constitutive models.  Inspired by this work, a number of follow-up studies were conducted. However, after these initial steps, the
field was not actively pursued for quite some time afterward. With the recent rise of machine learning and gains in computational efficiency, however, a variety of data-driven techniques have rapidly gained momentum in mechanics, as reviewed in 
\citep{Bock2019,Dornheim2023,Fuhg2024b}.

A significant step forward in neural constitutive modeling---and in scientific 
machine learning more broadly---is the integration of fundamental physical 
concepts, variously termed \emph{physics-informed} \citep{Raissi2019}, 
\emph{mechanics-informed} \citep{Asad2023}, \emph{physics-augmented} 
\citep{linden_neural_2023,Klein2024}, \emph{physics-based} \citep{Aldakheel2025}, 
\emph{physics-constrained} \citep{Kalina2023}, or \emph{thermodynamics-based} 
\citep{Masi2021}. Physical knowledge can be embedded in two ways \citep{rosenkranz_comparative_2023}: \emph{strongly}, 
through network architectures tailored to the problem 
\citep{Kalina2022a,Linka2021}, or \emph{weakly}, through problem-specific loss 
functions used during training \citep{Masi2021,Geiger2025}. As demonstrated 
in \citep{Masi2021,linden_neural_2023,Fuhg2023}, such models substantially enhance 
extrapolation capability. In the following, we provide a brief overview of the 
current state of the art in neural constitutive modeling.

Numerous works model elasticity with NNs, where the most common approach is to use 
architectures with the \emph{hyperelastic potential} as output and 
\emph{invariants} as inputs, e.g., \citep{Linka2021,Klein2021,Fuhg2022b,linden_neural_2023,Tac2024a,Bahmani2024,Benady2024,Peirlinck2024}. 
Direct calibration from stress-strain tuples is enabled via \emph{Sobolev 
training} proposed by \citet{Czarnecki2017} and  adopted to mechanics by \citet{Vlassis2020}, in which the loss function involves 
the gradient of the energy with respect to the deformation. Many works further 
employ \emph{polyconvex} NNs 
\citep{Klein2021,Tac2022a,Chen2022,Tac2024a,Bahmani2024,Jadoon2025a,Dammass2025b,Vijayakumaran2025,Geuken2025a} among many others, which enhance extrapolation \citep{linden_neural_2023,Kalina2024} and guarantee 
\emph{rank-one convexity} and thus \emph{ellipticity} 
\citep{Schroder2010}. This is most commonly realized via the \emph{fully 
input-convex neural networks (FICNNs)} of \citet{amos_input_2017}.

Besides neural elasticity models, a large number of works on \emph{inelasticity} have also emerged in recent years. A key requirement for any inelastic material model is compatibility with the \emph{Clausius-Duhem inequality}, which demands non-negative dissipation rates 
$D\ge 0$. 
A first group of data-driven models \emph{enforces} this through 
tailored loss terms. An influential example is the 
\emph{thermodynamics-based artificial neural network (TANN)} by \citet{Masi2021}, which realizes neural inelasticity. 
Despite its innovative character, the approach has two drawbacks: \emph{thermodynamic consistency is only weakly enforced}, and knowledge of the internal variables is required for 
training, which was solved later via integration for constitutive equations \citep{Masi2024}. A related approach for inelasticity is presented in \citep{He2022}: in 
contrast to the TANN, the internal state variables capturing the path-dependency 
are inferred from the hidden state of a recurrent neural network 
(RNN). 

Another class of \emph{neural inelasticity} models follows the 
\emph{structure of classical elasto-plastic formulations} more closely, e.g., 
\citep{Settgast2020,Malik2021,Vlassis2021,Vlassis2021b}, but does not guarantee 
$D\ge 0$ by construction. This limitation is overcome in 
\citep{Meyer2023a,Fuhg2023}, where \emph{thermodynamically consistent neural 
elasto-plasticity models} for small deformations are obtained, e.g., by modeling 
the dissipation or plastic potential via FICNNs \citep{Fuhg2023}. In both cases, 
training requires no explicitly prescribed internal variables; instead, the 
\emph{return-mapping problem} is solved within each optimization step. Extensions 
to finite deformations \citep{Boes2024,Jadoon2025} and, most recently, to 
finite-strain viscoelastic-plastic behavior with combined nonlinear kinematic and 
isotropic hardening \citep{vanderVelden2026} build on these ideas.

Beyond elasto-plasticity, several groups have also carried out extensive research 
on \emph{neural viscoelasticity models} grounded in a rigorous physical framework. 
An important contribution is that of \citet{Huang2022}, 
embedded in the \emph{generalized standard material (GSM)} framework\footnote{The GSM framework itself was introduced by 
\citet{Halphen1975} and traces back to pioneering works 
such as those of \citet{Coleman1963,Coleman1967} and \citet{Biot1965}, among others.}: 
thermodynamic consistency is ensured through a \emph{dissipation potential} that 
is convex with respect to the internal variables' rates, normalized, and 
stationary at zero rates---or, equivalently, through a \emph{dual dissipation 
potential} with analogous properties that depends on the thermodynamically 
conjugated forces. 
Meanwhile, several data-driven approaches building up on a small-strain GSM framework exist, 
e.g., \citep{rosenkranz_viscoelasticty_2024,Flaschel2025,Shi2026}. In contrast to 
\citep{Huang2022}, these require no prescribed internal variables during 
training---only their number must be specified. As discussed and benchmarked by 
\citet{rosenkranz_viscoelasticty_2024}, two main strategies exist: solving a \emph{constrained 
optimization problem} subject to the time-discretized evolution equation(s), which 
entails time integration over the entire sequence, or employing an 
\emph{auxiliary network}---either FNN \citep{Asad2023} or RNN \citep{rosenkranz_viscoelasticty_2024}---to 
\emph{provide the internal variable(s) during training}. Notably, in the latter 
case the auxiliary network is used only during training; afterwards, the evolution 
equation is solved in prediction mode.

Several works also address the more general case of \emph{finite-strain viscoelasticity}. 
A transversely isotropic model based on the multiplicative split of the 
deformation gradient and using \emph{neural ordinary differential equations 
(NODEs)} is proposed in \citep{Tac2023}, while a related GSM-based model assuming 
incompressibility and the multiplicative decomposition represents the thermodynamic 
potentials via FICNNs \citep{kalina_physics-augmented_2026}. Isotropic models with a multiplicative 
split and a co-rotational formulation are presented in 
\citep{Holthusen2024,Holthusen2024a,Holthusen2026}. Here, \citet{Holthusen2026} 
introduce a dual dissipation potential that ensures thermodynamic consistency 
under a less restrictive convexity requirement---the potential needs only be 
convex, stationary, and normalized in a modified invariant set rather than with 
respect to the thermodynamic forces themselves. An extension to transversely 
isotropic viscoelasticity follows in \citep{Holthusen2026a}. A further GSM-based 
finite-strain model, which does not assume the multiplicative split, is given in 
\citep{Asad2023,Asad2026}. 

Other thermodynamically consistent models avoid the two-potential framework 
altogether, including formulations based on deformation rates 
\citep{Upadhyay2026}, the NN-based model \citep{Abdolazizi2026}\footnote{This paper is a thermodynamically consistent follow-up to the work \citep{Abdolazizi2023a} that builds on the \emph{generalized Prony series}.} grounded in the 
framework of \citet{Liu2021b}, and deep rheological elements modeling the viscosity 
via NNs \citep{Califano2026}. 
Finally, a physics-augmented NN framework for \emph{strain-induced 
crystallization}, founded on a modified GSM framework describing the evolution of 
crystallinity, is presented in \citep{Friedrichs2026}.
\subsection{Literature review on full-field calibration techniques}
\label{sect:full-field_review}
Data for calibrating neural constitutive models can be obtained in various ways. 
Depending on their origin and spatial resolution, the sources fall into three 
categories: \emph{lower-scale simulations}, \emph{conventional experiments}, and 
\emph{full-field experiments}. We focus on the latter in the following.

Beyond global quantities such as \emph{reaction forces}, full-field setups provide 
spatially resolved displacement fields via \emph{full-field measurement 
techniques}. Unlike conventional experiments, the specimens are deliberately 
designed to produce \emph{inhomogeneous local fields}, thereby shifting 
constitutive modeling from a limited-data to a large-data regime 
\citep{Pierron2023,Fuhg2024b}. However, stress-strain pairs cannot be determined 
directly in this setting: while inhomogeneous displacement fields are accessible 
through camera systems with \emph{digital image correlation (DIC)} in two 
dimensions \citep{Pierron2020} or \emph{in-situ computed tomography (CT)} combined 
with \emph{digital volume correlation (DVC)} in three dimensions 
\citep{Lenoir2007}, the corresponding stress fields are not. Identifying the 
model parameters therefore requires solving an \emph{inverse problem}.

Numerous methods for inverse parameter identification have been developed over the 
past decades; comprehensive overviews are given in 
\citep{Avril2008,Roux2020,Romer2024,Chen2025}. In 
\citep{Mahnken1996,Sakaridis2024}, for instance, the parameters of 
classical models are identified by iteratively running finite element simulations 
of the experiment and minimizing the discrepancy between measured and simulated 
responses---an approach known as \emph{finite element model updating (FEMU)} 
\citep{Avril2008,Romer2024}. Further methods tailored specifically to inverse 
problems include the \emph{virtual fields method (VFM)} 
\citep{Grediac1989,pierron_virtual_2012}, the \emph{equilibrium gap method (EGM)} 
\citep{Claire2004}\footnote{\label{foot:EGM_VFM}The EGM can be interpreted as a 
special case of the VFM in which the virtual fields are built from locally 
supported finite element test functions, cf. \citep{Avril2008,Romer2024}.}, and the 
\emph{modified constitutive relation error (mCRE)} method.

\begin{remark}
    \label{rem:ROI_data}
	For inverse methods that obtain displacement fields from forward 
	simulations, such as FEMU or the mCRE, full-field displacement 
	measurements are not strictly required: the loss function can be formulated 
	purely in terms of global quantities such as reaction forces. When available, however, local displacement data is often 
	included to improve identification accuracy. In contrast, the EGM and VFM rely 
	directly on full-field displacement data over the entire domain to compute 
	deformation measures and, subsequently, stresses, so that no forward 
	simulation of the boundary value problem is needed. These approaches are 
	therefore typically far more efficient, but may be more sensitive to 
	measurement noise.
\end{remark}

Recently, inverse identification techniques have been transferred to a 
data-driven context, and new approaches have emerged. The \emph{EUCLID framework (efficient unsupervised constitutive law identification \& discovery)} by \citet{Flaschel2021,Flaschel2023}, for instance, extends the EGM/VFM by simultaneously determining parameters and suitable constitutive models from a catalog via sparse regression, with experimental validation on 2D DIC data given in \citep{Abbasi2026a}. NN-EUCLID \citep{Thakolkaran2022} replaces the model 
catalog with a neural constitutive model, while CANN-EUCLID \citep{Alheit2026} 
calibrates a constitutive artificial neural network (CANN) following 
\citet{Linka2023}. Related strategies discover neural yield surfaces 
\citep{Moon2026} and train Kolmogorov-Arnold networks (KANs) \citep{Thakolkaran2025}. 
Applications include 3D displacement fields from bulge inflation \citep{Meng2025} 
and 3D DVC data from a printed structure \citep{Bourdyot2026}. Along the same 
lines, \citet{Ferreira2026} calibrate various neural elasto-plastic models, 
including RNN-based ones, via EGM, whereas \citet{Lourenco2024} integrate 
RNN models into the VFM. An approach combining the statistical finite element method (stat-FEM) with EUCLID is introduced in \citep{knauf_narouie_unsupervised_2026}. Thereby, the
integration of the stat-FEM reduces sensitivity to noise and sparse displacement data.
Further approaches address anisotropic elasticity 
\citep{Li2026} and heterogeneous samples \citep{Shi2025a,Tac2026,Chaurasiya2026}.
The \emph{NN-mCRE} approach 
extends the mCRE concept, replacing the classical constitutive model with a physics-augmented neural network (PANN) 
\citep{Benady2024,Benady2024a}. PANNs are also combined with FEMU \citep{Wu2025}, 
as are spline-based constitutive models \citep{Wiesheier2024,Wiesheier2026} and 
CANNs for incompressible hyperelasticity \citep{Knipper2026}. In the 
same manner, neural elasto-plastic \citep{Gavris2025} and neural cohesive zone 
models \citep{Gavris2026} are calibrated.
\subsection{Objectives and contributions of this work}
\label{ssec:objectives}
As discussed in the literature overview given above, numerous \emph{neural 
viscoelasticity models} exist that combine modern machine learning methods with a 
sound physical basis. Among these, NN models that employ invariants and are 
embedded into the GSM framework appear very promising, as they enforce material 
symmetry and thermodynamic consistency by construction. However, irrespective of 
the specific architecture, the vast majority of these approaches are calibrated in 
a \emph{supervised} manner, i.e., they rely on stress-strain tuples that need to be derived via suitable assumptions from standard experiments, e.g., uniaxial tension tests with 
dogbone specimens.

There are only a few studies that calibrate highly flexible yet physically consistent data-driven models such as
\emph{PANNs} for \emph{viscoelasticity} in a fully
\emph{unsupervised} manner using \emph{full-field data}, thereby enabling access
to an extensive database. Such a model should not only maintain thermodynamic
consistency, objectivity, and material symmetry in its design, but also
capture the evolution of the \emph{internal variables} using a suitable \emph{time integration scheme}
without specifying them a priori.

We therefore present such a framework for the small-strain setting in this article. To this end, we adopt the 
PANN for isotropic small-strain viscoelasticity introduced 
in~\citep{rosenkranz_viscoelasticty_2024}, which builds on the GSM framework and employs two scalar-valued potentials, rendering the model 
\emph{thermodynamically consistent} by construction. Considering a thin specimen under a 
\emph{plane stress assumption}, we formulate a constrained optimization problem based on 
the EGM in combination with quasi-Newton optimizers and automatic 
differentiation. The evolution of the internal variables is captured via an \emph{implicit time 
	integration} scheme. Together with the plane stress condition $\sigma_{33}=0$, 
from which the unknown out-of-plane strain $\varepsilon_{33}$ is determined, this 
yields a system of nonlinear equations that is solved simultaneously by a local 
Newton iteration at each time step and quadrature point.
To drastically reduce the computational cost of training, the \emph{\emph{backward adjoint method}} is employed to compute the gradient of the target loss.

The organization of the remaining paper is as follows: In Sect.~\ref{sec:fundamentals}, the continuum mechanical basics and the underlying GSM framework for viscoelasticity are presented. After this, a PANN framework for the description of the potentials is introduced in Sect.~\ref{sec:pann}. This is followed by Sect.~\ref{sec:equilibrium_gap_method} introducing the EGM that is used for unsupervised training of the PANN.
The developed approach is exemplarily applied to several examples in Sect.~\ref{sec:numerical_examples}. After a discussion of the results, the paper is closed by concluding remarks and an outlook to necessary future work in Sect.~\ref{sec:conclusion}. 

In this work we use the notation specified in App.~\ref{app:notation}.
\section{Fundamentals}
\label{sec:fundamentals}
In this section, we first recall the continuum-mechanical relations for the small-strain setting. On this basis, we present a viscoelasticity framework within the class of GSMs.

\subsection{Continuum mechanics}
\label{ssec:continuum_mechanics}
Throughout this work, we restrict ourselves to \emph{small strains} and \emph{infinitesimal rigid-body rotations} of deformable bodies with the domain $\Set{B}\subset\Setnum{R}^{3}$ denoted as reference configuration with closure $\cl \Set{B}:=\Set{B} \cup \partial \Set{B}$, where $\partial\Set{B}$ is the corresponding boundary. We consider the  motion of the body at times $t\in \Set{T}:=[t_0,t_1]$, $t_0,t_1\in\Setnum{R}$.
Within this setting, the displacement field $\ve{u}: \cl\Set{B} \times \Set{T} \to \ts{1}{}, \; (\ve x,t)\mapsto \ve u(\ve x,t)$ yields the linearized strain tensor
\begin{equation}
    \ve{\upvarepsilon}(\ve x,t) :=  \sym\left(\nabla\ve{u}(\ve x,t)\right) \in\tss{2} \point
    \label{eq:linear_strain}
\end{equation}
The work-conjugate stress measure is the symmetric Cauchy stress, which is a field in the body, i.e., $\ve{\upsigma}: \cl\Set{B} \times \Set{T} \to \tss{2}, \; (\ve x,t)\mapsto \ve \upsigma(\ve x,t)$. Furthermore, $\ve{t}(\ve x, t, \ve n) = \ve{\upsigma}(\ve x,t)\cdot\ve{n}$ denotes the traction vector with the outward unit normal $\ve{n}\in\mathcal N:=\{\ve n \in\ts{1}{} \,|\, |\ve n| = 1\}$. Assuming that both volume and inertial loads are negligible, the balance of linear momentum reduces to
\begin{equation}
    \nabla\cdot\ve{\upsigma}^\top (\ve{x},t)= \ve{0}
    \qquad \forall\, \ve{x}\in\Set{B}, t\in \Set{T} \point
    \label{eq:balance_linear_momentum}
\end{equation}
The boundary $\partial \Set{B}$ of the deformable body is partitioned into a Dirichlet part $\partial\Set{B}^{\textrm{Dir}}$, on which the displacement $\ve{\hat{u}}\in\ts{1}{}$ is prescribed, and a Neumann part $\partial\Set{B}^{\textrm{Neu}}$ with prescribed traction $\ve{\hat{t}}\in\ts{1}{}$. The two parts are complementary and disjoint, i.e.,
$\partial\Set{B} = \partial\Set{B}^{\textrm{Dir}} \cup \partial\Set{B}^{\textrm{Neu}}$ and $\partial\Set{B}^{\textrm{Dir}} \cap \partial\Set{B}^{\textrm{Neu}} = \emptyset$. For a comprehensive treatment of continuum mechanics, we refer the reader to the textbooks of \citet{silhavy_mechanics_1997, Haupt2000, holzapfel_nonlinear_2000, altenbach_kontinuumsmechanik_2018}.
\subsection{Small-strain viscoelasticity}
\label{ssec:small_strain_viscoelasticity}
Within the scope of this work, we solely consider viscoelastic material behavior at \emph{small strains} with \emph{one internal variable} $\te{q}: \cl\Set{B} \times \Set{T} \to \tss{2}, \; (\ve x,t) \mapsto \te q(\ve x,t)$ of strain type. 
For simplicity, the functional dependencies on $\ve x$ and $t$ are omitted in the following.
\subsubsection{Generalized standard materials}
\label{ssec:gsms}
A well-established approach to formulate a constitutive model for viscoelasticity utilizes the theory of GSMs by \citet{Halphen1975} that traces back to pioneering works such as those of \citet{Coleman1963,Coleman1967} and \citet{Biot1965} and many others.

To define the model, we start to introduce the \emph{Helmholtz free energy (HFE)} 
\begin{equation}
    \varPsi: \tss{2}\times\tss{2} \to \Setnum{R}_{\geq0}, \; (\ve{\upvarepsilon},\te{q}) \mapsto \varPsi(\ve{\upvarepsilon},\te{q})
    \label{eq:hfe}
\end{equation}
from which the Cauchy stress and the \emph{driving force} $\te{A}$ of stress type, which is conjugated to $\te q$, follow as
\begin{align}
    \ve{\upsigma} := \diffp{\varPsi}{\ve{\upvarepsilon}} 
    \quad \text{and} \quad 
    \te{A} := -\diffp{\varPsi}{\te{q}} \in \tss{2} \comma
    \label{eq:def_derivatives_psi}
\end{align}
respectively. For $\ve{\upvarepsilon}=\te 0$ and $\te q=\te 0$, the normalization conditions $\partial \varPsi/\partial \ve{\upvarepsilon}|_{(\te 0, \te 0)} = \te 0$ and $\partial \varPsi/\partial \te q|_{(\te 0, \te 0)} = \te 0$ should hold.

In addition to the HFE, we introduce the \emph{dual dissipation potential (DDP)}\footnotemark
\footnotetext{In general the DDP can be introduced as a function of $\te{A},\ve{\upvarepsilon}$ and $\te{q}$, see \citep{rosenkranz_viscoelasticty_2024}. Throughout this work, we neglect the explicit dependence on the internal variables, as this is also a common assumption.
Please also note that, by using a Legendre-Fenchel transformation, a GSM formulation based on the dissipation potential can be obtained, cf. \citep{rosenkranz_comparative_2023}.}
\begin{equation}
    \diss: \tss{2}\times\tss{2} \to \Setnum{R}_{\geq0}, \; (\te{A}, \ve{\upvarepsilon}) \mapsto \diss(\te{A}, \ve{\upvarepsilon}) \point
    \label{eq:ddp}
\end{equation}
To imply \emph{thermodynamic consistency}, i.e., compatibility with the Clausius-Duhem inequality, by the tensor-valued evolution equations
\begin{align}
    \dot{\te{q}} = \diffp{\diss}{\te{A}} \label{eq:def_q_dot}
\end{align}
derived from $\varPsi$ and $\diss$, the following three well-known properties are imposed:
\begin{subequations}
    \begin{align}
        \diss \textrm{ is convex w.r.t. } \te{A} , \label{eq:ddp_convex} \\
        \diss(\te{0},\ve{\upvarepsilon}) = 0 \quad\forall\,\ve{\upvarepsilon}\in\tss{2} \textrm{, and} \label{eq:diss_is_zero} \\
        \left.\diffp{\diss}{\te{A}}\right|_{\te{A}=\te{0},\ve{\upvarepsilon}} = \te{0} \quad\forall\,\ve{\upvarepsilon} \in\tss{2} \point\label{eq:ddp_minimum}
    \end{align}%
\end{subequations}
With that it is straight forward to show that non-negative dissipation rates $D\ge 0$ are ensured by construction, cf. \citep[p. 569]{ottosen_mechanics_2005}.
\subsubsection{Material symmetry}
Beyond \emph{thermodynamic consistency}, the constitutive model is required to comply with the \emph{principle of material symmetry}, thereby reflecting the structure of
the material at a smaller length scale~\citep{neumann_vorlesungen_1885, riemer_construction_2026}. To this end, both the HFE and the DDP are formulated in terms of \emph{scalar-valued
invariants}
\begin{equation}
    I_\alpha: \: \tss{2}\times\dots\times\tss{2} \to \Setnum{R},
    (\te{s}_1\dots\te{s}_n) \mapsto I_\alpha(\te{s}_1\dots\te{s}_n) \comma
    \label{eq:def_invariant}
\end{equation}
$\alpha\in\{1,\dots,m\}$, $m\in\Setnum{Z}_{\geq3}$, which depend in turn on the symmetric 2nd order tensors $\te{s}_1,\dots,\te{s}_n\in\tss{2}$,
$n\in\Setnum{Z}_{\geq1}$. Each argument $\te{s}_\beta$, $\beta\in\{1,\dots,n\}$, may be identified with either $\ve{\upvarepsilon}$, $\te{q}$, or $\te{A}$. By construction, every invariant $I_\alpha$ in Eq.~\eqref{eq:def_invariant} satisfies the \emph{invariance condition}
\begin{equation}
    I_\alpha(\te{s}_1,\dots,\te{s}_n) = I_\alpha(\te{Q}\cdot\te{s}_1\cdot\te{Q}^\top,\dots, \te{Q}\cdot\te{s}_n\cdot\te{Q}^\top ) \:\forall\:\te{Q}\in\Group{G}
    \label{eq:invariant_condition}
\end{equation} 
where $\Group{G}\subseteq\Group{O}(3)$ denotes the symmetry group. As the present study is concerned exclusively with \emph{isotropic materials}, $\Group{G}=\Group{O}(3)$. The corresponding invariants are then referred to as \emph{isotropic invariants}.\footnote{In the anisotropic case, $\Group{G}\subset\Group{O}(3)$. Nonetheless through the introduction of \emph{structural tensors} in conjunction with the principle of \emph{isotropic extension}, results from \emph{isotropic invariants} can be used, see \citet{lokhin_nonlinear_1963, zheng_tensors_1993, xiao_isotropic_1996, apel_approaches_2004, riemer_construction_2026}.} 
In order to distinguish the invariants entering the HFE from those entering the DDP, each is endowed with the superscript $(\bullet)^{\varPsi}$ or $(\bullet)^{\diss}$, respectively. The two potentials are thus reintroduced as
\begin{subequations}
    \begin{align}
        \varPsi&:\tss{2}\times\tss{2} \to \Setnum{R}_{\geq0}, \; (\ve{\upvarepsilon},\te{q}) \mapsto \varPsi(\mathcal{I}^{\varPsi}(\ve{\upvarepsilon},\te{q}))
        \label{eq:hfe_invar} \\
        \diss&: 
        \tss{2}\times\tss{2} \to \Setnum{R}_{\geq0}, \; (\te{A}, \ve{\upvarepsilon}) \mapsto \diss(\mathcal{I}^{\diss}(\te{A}, \ve{\upvarepsilon})
        ) \, ,
        \label{eq:ddp_invar}
    \end{align}%
\end{subequations}
where the respective invariants are collected by the tuples $\mathcal{I}^{\varPsi} := (I_1^{\varPsi},\dots,I_{n_{\varPsi}}^{\varPsi})\in \Setnum{R}^{n_\varPsi}$ and $\mathcal{I}^{\diss} := (I_1^{\diss},\dots,I_{n_{\diss}}^{\diss})\in \Setnum{R}^{n_{\diss}}$, $n_{\varPsi}, n_{\diss}\in\Setnum{Z}_{\geq3}$.
\section{Physics-augmented neural network model}
\label{sec:pann}
Having introduced the GSM framework underlying the model, we now specify the potentials, introduced in Eqs.~\eqref{eq:hfe_invar} and \eqref{eq:ddp_invar}, by means of PANNs. To this end, three NNs are employed. Most of the modeling assumptions are adopted from \citet{rosenkranz_viscoelasticty_2024}; there, however, a formulation based on the dissipation potential is used. We emphasize that the PANN presented in the following is of \emph{Maxwell type} and therefore does not represent the most general form of a viscoelastic model---not even within the framework already partially restricted in Sect.~\ref{ssec:small_strain_viscoelasticity}, for instance through the use of a single internal variable.

\subsection{Helmholtz free energy}
The HFE is assumed to be convex w.r.t. $\ve{\upvarepsilon}$ and to be
additively decomposed into an \emph{equilibrium part} $\hfeeqpann$ and a \emph{non-equilibrium part} $\hfeeqpann$ with
\begin{align}
    \hfeeqpann&: (\ve{\upvarepsilon}) \mapsto 
    \hfeeqpann(\mathcal{I}^{\hfeeq}(\ve{\upvarepsilon})) \text{ and} \\ \hfeneqpann&: (\ve{\upvarepsilon},\te{q}) \mapsto \hfeneqpann(\mathcal{I}^{\hfeneq}(\underbrace{\ve{\upvarepsilon}-\te{q}}_{=:\te p})) \comma
\end{align}
where $\te{p}$ denotes the elastic strain.
We thus introduce two PANNs, one for each part of the HFE.
In the equations above, the tuples containing \emph{isotropic invariants} for the \emph{equilibrium} and \emph{non-equilibrium parts} are defined by 
\begin{align}
    \mathcal{I}^{\hfeeq}(\ve \upvarepsilon) &:= (I_1^{\hfeeq}(\ve \upvarepsilon),I_2^{\hfeeq}(\ve \upvarepsilon), I_3^{\hfeeq}(\ve \upvarepsilon)) \text{ and} \\
     \mathcal{I}^{\hfeneq}(\te p) &:= (I_1^{\hfeneq}(\te p),I_2^{\hfeneq}(\te p), I_3^{\hfeneq}(\te p)) \comma
 \end{align}
with the \emph{isotropic invariants} of a symmetric 2nd order tensor $\te s$---in our case $\ve \upvarepsilon$ or $\te p$---given by 
\begin{align}
    I_1^\bullet(\te s)&:= \tr(\te s) \; , \; 
    I_2^\bullet(\te s):= \tr((\mathbb{I}^{\textrm{dev}}:\te s)^2) \; \textrm{and} \label{eq:iso_invs}\\
    I_3^\bullet(\te s)&:= \tr(\te s^4) \; . \nonumber
\end{align}
As shown in App.~\ref{app:functional_basis_trace}, the invariants given in Eq.~\eqref{eq:iso_invs} form a \emph{functional basis} for the group $\Group{O}(3)$ if $\te s\in \tss{2} \cap \ts{2}{\tr\neq 0}$ is a symmetric and non-deviatoric tensor with $\ts{2}{\tr\neq 0}:=\{\te{r}\in\ts{2}{}| \tr(\te{r})\neq0\}$. Furthermore, the invariants are convex in $\te s$.\footnote{Convexity w.r.t. $\te{s}\in\tss{2}$ for invariants $\tr(\te{s})$, $\tr(\te{s}^2)$ and $I_3=\tr(\te{s}^4)$ is shown in \citep{rosenkranz_viscoelasticty_2024}. Since $\mathbb{I}^{\textrm{dev}}:\te{s}$ is linear in $\te{s}$, $I_2 = \tr((\mathbb{I}^{\textrm{dev}}:\te{s})^2)$ inherits the convexity property of $\tr(\te{s}^2)$~\citep[Sect. 3.2.2]{boyd_convex_2004}.}

The PANN for the \emph{equilibrium part} is given by
\begin{align}
    \hfeeqpann(\mathcal{I}^{\hfeeq}(\ve{\upvarepsilon}))
    &= \hfeeqnn(\mathcal{I}^{\hfeeq}(\ve{\upvarepsilon})) \label{eq:hfeeqpann_contributions} \\
    &- \hfeeqnn(\mathcal{I}^{\hfeeq}(\te{0})) \nonumber\\
    &- \left.\diffp{\hfeeqnn}{I_1^{\hfeeq}}\right|_{I_1^{\hfeeq}(\te{0})} 
    \!\!I_1^{\hfeeq}(\ve{\upvarepsilon}) \point \nonumber
\end{align}
The value $\hfeeqnn(\mathcal{I}^{\hfeeq}(\ve{\upvarepsilon}))$ is the scalar-valued output of a component-wise non-decreasing FICNN \citep{amos_input_2017,Klein2021,kalina_physics-augmented_2026} 
\begin{equation}
    \hfeeqnn: \ve \upvarepsilon \mapsto \hfeeqnn(\mathcal{I}^{\hfeeq}(\ve \upvarepsilon)) 
    \label{eq:hfeeqnn}
\end{equation}
with the introduced strain-type invariants as input.
 Following \citet{linden_neural_2023}, two correction terms are subtracted from the network output $\hfeeqnn(\mathcal{I}^{\hfeeq}(\ve{\upvarepsilon}))$ to ensure that both energy and stress contribution from the \emph{equilibrium part} vanish in the undeformed state $\ve{\upvarepsilon}=\te{0}$.
The entire PANN for the \emph{equilibrium part} is visualized in Fig.~\ref{fig:psi_eq_PANN}.
\begin{figure*}[ht]
  \centering
  \includegraphics[clip]{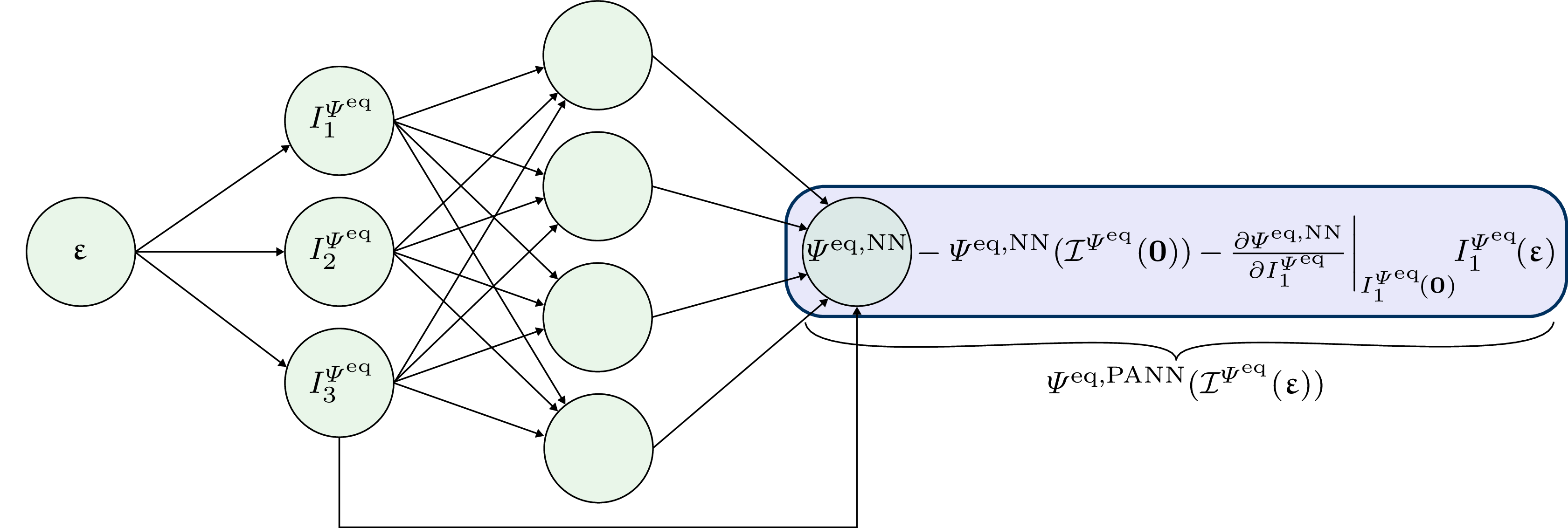}
  \caption{PANN model $\hfeeqpann$ for the \emph{equilibrium part} of the HFE. Here, for example, the NN has one hidden layer with four neurons. Additionally, we use an architecture with skip connections. Two correction terms $\hfeeqnn(\mathcal{I}^{\hfeeq}(\te{0}))$ and $\left.\diffp{\hfeeqnn}{I_1^{\hfeeq}}\right|_{I_1^{\hfeeq}(\te{0})} 
    \!\!I_1^{\hfeeq}(\ve{\upvarepsilon})$ are subtracted from the network output $\hfeeqnn(\mathcal{I}^{\hfeeq}(\ve{\upvarepsilon}))$, thereby guaranteeing a vanishing energy and stress contribution in the undeformed state $\ve{\upvarepsilon}=\te{0}$.}
  \label{fig:psi_eq_PANN}
\end{figure*}

\begin{remark}
\label{rem:convexity}
    Since the invariants are convex in $\ve \upvarepsilon$, $\hfeeqnn$ is convex and component-wise non-decreasing, and $I_1^{\hfeeq}$ is linear in $\ve \upvarepsilon$, the chosen potential $\hfeeqpann$ given in Eq.~\eqref{eq:hfeeqpann_contributions} is convex in $\ve \upvarepsilon$. 
\end{remark}

Analogous to the \emph{equilibrium part}, we define the PANN for the \emph{non-equilibrium part} by
\begin{align}
    \hfeneqpann(\mathcal{I}^{\hfeneq}(\te p))
    &= \hfeneqnn(\mathcal{I}^{\hfeneq}(\te{p})) \label{eq:hfeneqpann_contributions} \\
    &- \hfeneqnn(\mathcal{I}^{\hfeneq}(\te{0})) \nonumber\\
    &- \left.\diffp{\hfeneqnn}{I_1^{\hfeneq}}\right|_{\te{0}} 
    I_1^{\hfeneq}(\te p) \comma \nonumber
\end{align}
with a second component-wise non-decreasing FICNN
\begin{equation}
    \hfeneqnn: \te p \mapsto \hfeneqnn(\mathcal{I}^{\hfeneq}(\te p)) \point
    \label{eq:hfeneqnn}
\end{equation}
Since only the \emph{elastic strain} $\te p = \ve{\upvarepsilon} - \te{q}$ enters the invariants, the constitutive model is of \emph{Maxwell type}. As $\te p$ is linear in $\ve \upvarepsilon$, convexity of $\hfeneqpann$ in the strain follows from the argumentation given in Remark~\ref{rem:convexity}.
\subsection{Dual dissipation potential}
For the DDP, an architecture similar to that of the two parts of the HFE is employed. To this end, we introduce an \emph{isotropic invariant} set for $(\te A, \ve \upvarepsilon)$ that consists of the invariants 
\begin{equation}
    \begin{aligned}
        I_1^{\diss}\!(\te{A})\!&:= \tr(\te{A}) ,\quad
        I_2^{\diss}\!(\te{A})\!:= \tr((\mathbb{I}^{\textrm{dev}}:\te{A})^2) ,\\
        I_3^{\diss}\!(\te{A})\!&:= \tr(\te{A}^4), \quad
        I_4^{\diss}\!(\ve{\upvarepsilon})\!:= \tr(\ve{\upvarepsilon}) , \\
        I_5^{\diss}\!(\ve{\upvarepsilon})\!&:= \tr((\mathbb{I}^{\textrm{dev}}:\ve{\upvarepsilon})^2) \quad\textrm{and}\quad
        I_6^{\diss}\!(\ve{\upvarepsilon})\!:= \tr(\ve{\upvarepsilon}^4) \point
    \end{aligned}
    \label{eq:disspann}
\end{equation}
We collect the invariants given in Eq.~\eqref{eq:disspann} in the tuple $\mathcal{I}^{\diss}(\te A, \ve \upvarepsilon):= \left(I_1^{\diss}(\te A),\ldots, I_6^{\diss}(\ve{\upvarepsilon})\right)$. Note that the invariants in $\mathcal{I}^{\diss}$ are convex in $\te A$ and $\ve \upvarepsilon$.
\begin{remark}
    We want to point out that the invariants in Eq.~\eqref{eq:disspann} that are used to describe the DDP do not form a \emph{functional basis}. This is the case because of Corollary~\ref{corollary:not_functional_basis_deviatoric} and missing \emph{mixed invariants}, i.e., invariants that depend on both $\te{A}$ and $\ve{\upvarepsilon}$, cf. \citep{smith_isotropic_1971, boehler_irreducible_1977}.
\end{remark}

The \emph{PANN for the DDP} is defined as
\begin{align}
    \disspann(\mathcal{I}^{\diss}(\te{A}, \ve{\upvarepsilon}))
    &= \dissnn(\mathcal{I}^{\diss}(\te{A}, \ve{\upvarepsilon})) \label{eq:disspann_contributions} \\
    &- \dissnn(\mathcal{I}^{\diss}(\te{0}, \ve{\upvarepsilon})) \nonumber\\
    &- \left.\diffp{\dissnn}{I_1^{\diss}}\right|_{I_1^{\diss}(\te{0},\ve{\upvarepsilon})} 
    I_1^{\diss}(\te{A}, \ve{\upvarepsilon}) \comma \nonumber
\end{align}
where the introduced NN is again a component-wise non-decreasing FICNN with the invariants collected in $\mathcal{I}^{\diss}$ acting as inputs. The NN is given by
\begin{equation}
    \dissnn: (\te A, \ve \upvarepsilon) \mapsto  \dissnn(\mathcal{I}^{\diss}(\te A, \ve \upvarepsilon)) \point
    \label{eq:dissnn}
\end{equation}
Convexity of $\disspann$ in $\te A$ and $\ve \upvarepsilon$ follows by applying the argumentation given in Remark~\ref{rem:convexity}.
\begin{remark}
    It should be noted that the convexity of the DDP in the strain $\ve \upvarepsilon$ is not strictly required in the GSM framework.    
    It would therefore also be possible to use a partially input-convex neural network (PICNN), cf. \citep{amos_input_2017,rosenkranz_viscoelasticty_2024,Friedrichs2026}. 
    Such an approach is more flexible, but the number of trainable variables in the network increases. Since the selected PANN model, which is based exclusively on FICNNs, has proven to be sufficiently flexible for the examples considered, we will not discuss PICNNs further here.
\end{remark}
\begin{remark}
    \label{rem:parameters}
    Due the required convexity and component-wise non-decreasingness of the NN potentials $\hfeeqnn$, $\hfeneqpann$, and $\dissnn$, the weights of the three individual FICNNs have to be non-negative. A detailed discussion about this can be found in \citep{amos_input_2017,Klein2021, rosenkranz_viscoelasticty_2024}. To keep the notation clear, all trainable parameters, i.e., all weights and biases of all three FICNNs, are summarized in $\V{\theta}\in\mathcal{P}_\theta$,  where the set $\mathcal{P}_\theta\subset\Setnum{R}^{n_\theta}$, $n_\theta\in\Setnum{Z}_{>0}$. The set $\mathcal{P}_\theta$ contains the weight constraints.
    For all FICNNs we use the convex and non-decreasing softplus activation function.
\end{remark}
\section{Equilibrium gap method}
\label{sec:equilibrium_gap_method}
As discussed in the Sect.~\ref{sect:full-field_review}, numerous approaches are available for the calibration of constitutive models via full-field data. In the present work, we adopt a special case of the VFM, the EGM \citep{Claire2004}, to identify the material model's parameter in an unsupervised manner. The parameters, collected in $\V \theta$, are in our case weights and biases of the FICNNs introduced in Sect.~\ref{sec:pann}. The methodology is therefore similar to the NN-EUCLID framework by \citet{Thakolkaran2022}. In our case, however, we consider a neural viscoelastic model instead of a hyperelastic one.

Data for the EGM can be collected via experiments with full-field displacement measurement \citep{Abbasi2026a,jailin_experimental_2024}, where the displacements are typically obtained via optical techniques such as DIC.
In the EGM, the displacement is required not just at a few discrete locations but at a sufficiently large number of points distributed across the region of interest (ROI), cf. Remark~\ref{rem:ROI_data}. 
Since in real setups DIC only allows to observe displacements on the sample surface---and in many setups only on one side---we will restrict ourselves to thin samples in the following. In this case, the ROI is therefore a 2D domain denoted by $\varOmega\subset\Setnum R^2$ with boundary $\partial \varOmega$, i.e., the closure is $\cl \varOmega := \varOmega \cup \partial \varOmega$.

We are working on the basis of the following assumptions and assume that the data listed below is known after the raw data has been preprocessed:
\begin{enumerate}[label=(\Roman*)]
    \item The specimen is thin compared with its in-plane characteristic length, so that a \emph{plane stress state} is assumed throughout the entire domain.
    \item A set of nodes with reference coordinates $(x_1^\lambda,x_2^\lambda)\in\varOmega \cup \partial \varOmega$, $\lambda\in\{1,\ldots,\nn\}$, and a \emph{triangulation} thereof into  $e\in\{1,\ldots,\nele\}$ elements is given.
    \item \emph{In-plane displacements} $\tsvar{u}_1^\lambda$, $\tsvar{u}_2^\lambda \in \Setnum{R}$  at the  nodes are known for all time steps $s\in\{1,\ldots,\nt\}$.
    \item The \emph{in-plane displacement field} is approximated as continuous and piecewise linear over the \emph{triangulation}, i.e., constant strain triangle (CST) elements are used.
    \item The \emph{global reaction forces} $\tvar{F}^\kappa$ normal to respective boundaries $\partial \varOmega_\kappa$ with applied loading are known, $\kappa\in\{1,\dots,\nf\}$ with $\nf\in\Setnum{Z}_{\geq1}$ denoting the number of reaction forces.
\end{enumerate}
The core idea of the EGM is to enforce the discretized weak form of the balance of linear momentum via loss terms that \emph{minimize the deviation from equilibrium}. More strictly speaking, within the discretized (finite element) setting, this means enforcing equilibrium of the \emph{assembled nodal forces} only approximately, i.e., as a soft constraint via the loss term, rather than exactly.

\subsection{Two-dimensional finite element discretization}
\label{ssec:fem}
The optimization problem formulated for the EGM requires some parts of the finite element method (FEM). To clearly state that the quantities that are introduced in this section are \emph{two-dimensional}, we switch to index notation and use the lower case Latin letters $a,b,c,d\in\{1,2\}$ as subscripts.

Let $\cl \varOmega^{\textrm{CST}}:=\{(\xi_1,\xi_2)\in\Setnum{R}^2 | 0\leq\xi_1\leq1,\:0\leq\xi_2\leq1-\xi_1\}$ denote the closure (domain and boundary) of a CST element and
\begin{equation}
	N^\alpha: \cl \varOmega^{\textrm{CST}} \to \Setnum{R}, \quad (\xi_1,\xi_2) \mapsto N^\alpha(\xi_1,\xi_2)
	\label{eq:shape_function_map}
\end{equation}
the corresponding shape functions. The reference \emph{in-plane} coordinates of the nodes with local numbers $\alpha\in\{1,2,3\}$ of an element $e\in\{1,2,\dots,\nele\}$ are given by $x^{e\alpha}_b \in \Setnum{R}$. Hence, the coordinates of the elemental Jacobian are defined by
\begin{equation}
	J_{ab}^e:=\sum_{\alpha=1}^3 \diffp{N^\alpha(\xi_1,\xi_2)}{\xi_a}  x^{e\alpha}_b = \text{const} \point
	\label{eq:jacobian_element}
\end{equation}
Due to the use of CST triangles, $J_{ab}^e$ is element-wise constant. The determinant of the Jacobian's matrix representation $\M J^e\in\Setnum{R}^{2\times 2}$ is given by $J^e := \det \M J^e$.
The coordinates of the so-called B-matrix of element $e$ and the node with local node number $\alpha$ are defined in index notation by
\begin{equation}
	B^{e \alpha}_{abc} := \sym
	\left( 
	\diffp{N^\alpha(\xi_1,\xi_2)}{\xi_d} (J^{e}_{da})^{-1} \delta_{bc} 
	\right) = \text{const}\comma
	\label{eq:b_matrix}
\end{equation}
where $\sym(\bullet)$ denotes the symmetrization w.r.t. the indices $a$ and $b$. Note that this quantity is also element-wise constant for the special choice of CST elements. 
\begin{remark}
	\label{rem:evaluation_at_quad_point}
	The quantities $J_{ab}^e$ and $B^{e \alpha}_{abc}$ are evaluated at quadrature or Gauss points. To note this, an index $g\in\{1,2,\dots,\nqp\}$ could be added, i.e., $J_{ab}^{eg}$ and $B^{e g \alpha}_{abc}$. However, as we only consider \emph{CST elements}, the number $\nqp\in\Setnum{Z}_{>0}$ of Gauss points is set to $\nqp=1$ and the quadrature point to $(\xi_1,\xi_2)=(\frac{1}{3}, \frac{1}{3})$. Hence, rendering the additional index $g$ not necessary. We rather keep notation simple and impose that the index $e$ marks not only the element, but also the evaluation at the single quadrature point of that element. Exceptions to that are local node values that have another index $\alpha$, such as the elemental nodal displacements $\tsvar{u}^{e\alpha}_a$ at a time step $s\in\{1,2,\dots,\nt\}$.
\end{remark}

To obtain the \emph{in-plane displacement fields} $\tsvar{u}_1(x_1,x_2)$ and $\tsvar{u}_2(x_1,x_2)$ at a time step $s$, we apply assumptions (III) and (IV). Thus, the \emph{in-plane displacements} in each CST element $e$ are approximated by 
\begin{align}
	\tsvar{u}_a^e(\xi_1,\xi_2):=\sum_{\alpha=1}^3 \tsvar{u}_a^{e\alpha} N^\alpha(\xi_1,\xi_2) \point
\end{align}
Hence, the \emph{in-plane strains} of an element are given by
\begin{equation}
	\tvar{\varepsilon}^{e}_{ab} := \sum_{\alpha=1}^3 \tsvar{u}^{e\alpha}_c B^{e \alpha}_{abc} = \text{const} \point
	\label{eq:epsilon_fem2d}
\end{equation}
The \emph{local nodal forces} of node $\alpha$ in element $e$ follow from the integral
\begin{align}
	\tvar{f}^{e\alpha}_a :\!\!&=
	\int\limits_{\varOmega^\text{CST}}
	{\sigma}_{bc}(\tvar{\ve \upvarepsilon}(\xi_1,\xi_2), \tsvar{\te{q}}(\xi_1,\xi_2)) B^{e \alpha}_{cba} J^e h \, \mathrm{d}A \nonumber\\
	&=
	\tvar{\sigma}^{e}_{bc} B^{e \alpha}_{cba} J^e \frac{h}{2} \comma
	\label{eq:local_nodal_forces}
\end{align}
by numerical integration, where $h\in\Setnum{R}_{>0}$ is the constant reference thickness. To compute the \emph{local nodal forces} $\tvar{f}^{e\alpha}_a$, the \emph{in-plane stresses} $\tvar{\sigma}^{e}_{ab} = \tvar{\sigma}_{ab}(\tvar{\ve \upvarepsilon}^e, \tsvar{\te{q}}^e)$ in each element are required, where $\tvar{\ve \upvarepsilon}^e$ and $\tsvar{\te{q}}^e$ are strain tensor and internal variable at time step $s$ in element $e$ and $\tvar{\sigma}_{ab}=\tvar{\ve \upsigma} :(\ve e_a \ve e_b)$. Since we only consider CST elements, strain and also internal variable are spatially constant in each element for a fixed time. Thus, the numerical integration according to Eq.~\eqref{eq:local_nodal_forces} is exact.
We discuss how $\tvar{\ve \upvarepsilon}^e$ and $\tsvar{\te{q}}^e$ are determined in the considered \emph{plane stress} setup to derive the \emph{in-plane stresses} $\tvar{\sigma}^{e}_{ab}$
in the next Section~\ref{ssec:plane_stress}. 

Finally, we get the \emph{assembled nodal forces} at a node with global number $\lambda\in\{1,\ldots,\nn\}$ from the assembly process
\begin{equation}
    \tvar{f}^\lambda_a := \sum_{e=1}^{\nele} \sum_{\alpha=1}^{3} \tvar{f}^{e\alpha}_a A^{\lambda e \alpha} \point
    \label{eq:global_nodal_forces}
\end{equation}
The quantity $A^{\lambda e \alpha}\in\{0,1\}$ in Eq.~\eqref{eq:global_nodal_forces} takes either the value 1 or 0 depending on the connectivity given by the mesh. 
\subsection{Stress determination for a single time step}
\label{ssec:plane_stress}
In order to calculate the \emph{local nodal forces} $\tvar{f}^{e\alpha}_a$ at a time step $s$, the \emph{in-plane stresses} $\tvar{\sigma}_{ab}^{e}$ must be determined by means of a material model, such as the PANN introduced in Sect.~\ref{sec:pann}.
According to Eq.~\eqref{eq:def_derivatives_psi}, the calculation of stress $\tvar{\ve \upsigma}^{e}$ is possible under the condition that both \emph{three-dimensional strain} $\tvar{\ve{\upvarepsilon}}^{e}$ and internal variable $\tsvar{\te{q}}^e$ are known for a time step $s$. 
The latter can be determined via \emph{implicit time integration} by finding the root of the residual
\begin{equation}
	{}^{s,i}\te{R}^e :=\tsvar{\te{q}}^e - {}^{s-1}\hspace{-1pt}\te{q}^e -
	\tvar{\Delta t} \left. \diffp{\diss}{\te{A}}
	\right|_{\te{A}=\tivar{\te{A}}^e, \ve{\upvarepsilon}=\tvar{\ve{\upvarepsilon}}^e} \comma
	\label{eq:residual_internal_variable}
\end{equation}
where $\tvar{\Delta t} := {}^{s}t - {}^{s-1}t$ denotes the time step size between two consecutive times ${}^{s}t$ and ${}^{s-1}t$. The coordinates of the residual are given by ${}^{s,i}\!R^e_{kl}$, $k,l\in\{1,2,3\}$.
As noted in the Eq.~\eqref{eq:residual_internal_variable} above, the \emph{three dimensional strain} $\tvar{\ve{\upvarepsilon}}^{e}$ is also required in this case. However, Eq.~\eqref{eq:epsilon_fem2d}, and the \emph{in-plane nodal displacements} $\tsvar{u}^{e\alpha}_a$, give rise to only the \emph{in-plane strains} $\tvar{\varepsilon}^{e}_{ab}$. 
By using the \emph{plane stress assumption} $(\textrm{I})$, we can determine the missing out-of-plane strain component $\tvar{\varepsilon}^{e}_{33}$ and hence the full strain tensor $\tvar{\ve{\upvarepsilon}}^{e}$. 
In a real experimental setup, this quantity could be determined by stereo measurement using a second DIC system, cf.~\citep{linden_dual-stage_2025}. On the other hand, $\tvar{\varepsilon}^{e}_{33}$ can also be calculated from the known constraint for a \emph{plane stress state}. Thus, we introduce a second residual
\begin{equation}
	\tivar{R}^e := \tivar{\sigma}^{e}_{33} = \left.\left(\diffp{\varPsi}{\ve{\upvarepsilon}}\right)_{33}\right|_{\ve{\upvarepsilon}=\tivar{\ve{\upvarepsilon}^{e}}, \te{q}=\tivar{\te{q}}^{e}} \point
	\label{eq:residual_stress33}
\end{equation}
Accordingly, the two residuals in Eqs.~\eqref{eq:residual_internal_variable}~and~\eqref{eq:residual_stress33} are solved in a \emph{monolithic} manner by means of a Newton iteration at each quadrature point.\footnote{\label{foot:least_squares}In the combined Newton iteration of the two residuals in  Eqs.~\eqref{eq:residual_internal_variable}~and~\eqref{eq:residual_stress33}, a \emph{least squares method} is used to calculate the incremental update. Owing to the non-symmetric coefficient matrix, a standard linear-system solver was less robust compared to the \emph{least squares method} during calibration.}

At the first times step $s=1$, it holds ${}^{1}\!\ve{\upvarepsilon}^{e} = {}^{1}\te{q}^{e} := \te{0}$. 
To start the iteration ($i=1$) for a subsequent time step, we set ${}^{s,1}\!{\varepsilon}^{e}_{33} := {}^{s-1}\!{\varepsilon}^{e}_{33}$ and $\tsvar{\te{q}}^{e} := {}^{s-1}\hspace{-1pt}\te{q}^{e}$. 
The solution is regarded as converged if the residual in Eq.~\eqref{eq:residual_internal_variable} fulfills $\max_{ekl}|{}^{s,i}\!R_{kl}^{e}| < \epsilon^{\te{q}}_{\textrm{tol}}\in\Setnum{R}_{>0}$
and the residual in Eq.~\eqref{eq:residual_stress33} $\max_{e}|\tivar{R}^{e}| < \epsilon^\sigma_{\textrm{tol}}\in\Setnum{R}_{>0}$. We then set $\tvar{\varepsilon}^{e}_{33} := \tivar{\varepsilon}^{e}_{33}$, $\tsvar{\te{q}}^{e} := {}^{s,i}\te{q}^{e}$, ${}^{s}\te{R}^e:={}^{s,i}\te{R}^e$ and $\tvar{R}^e:=\tivar{R}^e$.
\subsection{Optimization problem}
\label{ssec:optimization problem}
With the necessary concepts from the finite element discretization being introduced, we can start to formulate an optimization problem for inversely identifying the parameters $\V \theta$. Thus, the main task is to formulate a suitable criterion---a loss function---whose minimization leads to a well-calibrated material model, in our case a viscoelastic PANN. Parallel to the description given in the following paragraphs, Fig.~\ref{fig:egm} summarizes the individual steps involved in the calculation of the loss.

First, the Jacobian matrices $\M J^e$ and the determinants $J^e$, the coordinates $B^{e \alpha}_{abc}$ of all B-matrices and from that the \emph{in-plane strains} $\tvar{\varepsilon}^{e}_{ab}$ are calculated from the given full-field data. This step needs to be performed only once before the optimization begins.
The next steps are performed in each iteration of the optimizer.
For the current parametrization $\V{\theta}\in\Set{P}_\theta$ of the PANN, we compute the \emph{three-dimensional stress} 
\begin{align}
   \tvar{\ve{\upsigma}}^{e}(\V \theta) = {\ve{\upsigma}}(\tvar{\varepsilon}^e_{11},\tvar{\varepsilon}^e_{22}, \tvar{\varepsilon}^e_{12}, \tvar{\varepsilon}^e_{33}(\V \theta), \tsvar{\te q}^e(\V \theta), \V \theta) \label{eq:depend_sigma}
\end{align}
in each element and time step as explained in Sect.~\ref{ssec:plane_stress}. 
With that, \emph{local nodal forces} $\tvar{f}^{e\alpha}_a(\V \theta)$ and, consequently, \emph{assembled nodal forces} $\tvar{f}^\lambda_a(\V \theta)$ are calculated using Eqs.~\eqref{eq:local_nodal_forces}~and~\eqref{eq:global_nodal_forces}. 
\paragraph{Sorting assembled nodal forces}
To formulate the loss function, we must distinguish between two types of nodes:
\begin{itemize}
    \item those located in the domain or on boundaries with homogeneous Neumann boundary conditions, and
    \item those located on boundaries with imposed reaction forces.
\end{itemize}
To do so, we identify $\nf\in\Setnum{Z}_{\ge 1}$ boundaries $\partial\varOmega_{\kappa}^\text{disc}\subset \partial\varOmega^{\textrm{disc}}$ of the discretized specimen $\cl\varOmega^{\textrm{disc}} = \varOmega^{\textrm{disc}} \cup \partial\varOmega^{\textrm{disc}}$ at which \emph{global reaction forces} are applied. 
Thereby, $\kappa\in\{1,\dots,\nf\}$ is the index of the boundary at which a \emph{global reaction force} ${}^s \!\ve{F}^\kappa\in\ts{1}{}$ with the normal component $\tvar{F}^\kappa := \tvar{F}^\kappa_a  n^\kappa_a$ is known $(\textrm{V})$, see Fig.~\ref{fig:egm}. 
The remaining part $\cl\varOmega^{\textrm{disc}} \setminus \bigcup_{\kappa=1}^{\nf}\partial\varOmega_{\kappa}^\text{disc}$ of the discretized geometry is thus not associated with \emph{reaction forces} resulting from applied loading.

To identify the nodes laying on the force boundaries and those laying in the domain or on free boundaries, we introduce the sets 
\begin{align}
    \Force^\kappa&:=\left\{ \lambda \, | \, (x_1^\lambda, x_2^\lambda) \in \partial \varOmega_\kappa^\text{disc} \right\} \text{ and } \\
    \Free&:=\left\{ \lambda \, | \, (x_1^\lambda, x_2^\lambda) \in \cl\varOmega^{\textrm{disc}} \setminus \cup_{\kappa=1}^{\nf}\partial\varOmega_{\kappa}^\text{disc} \right\} \comma
\end{align}
respectively.

\paragraph{Loss function}
\begin{figure*}[ht]
  \centering
  \includegraphics[clip]{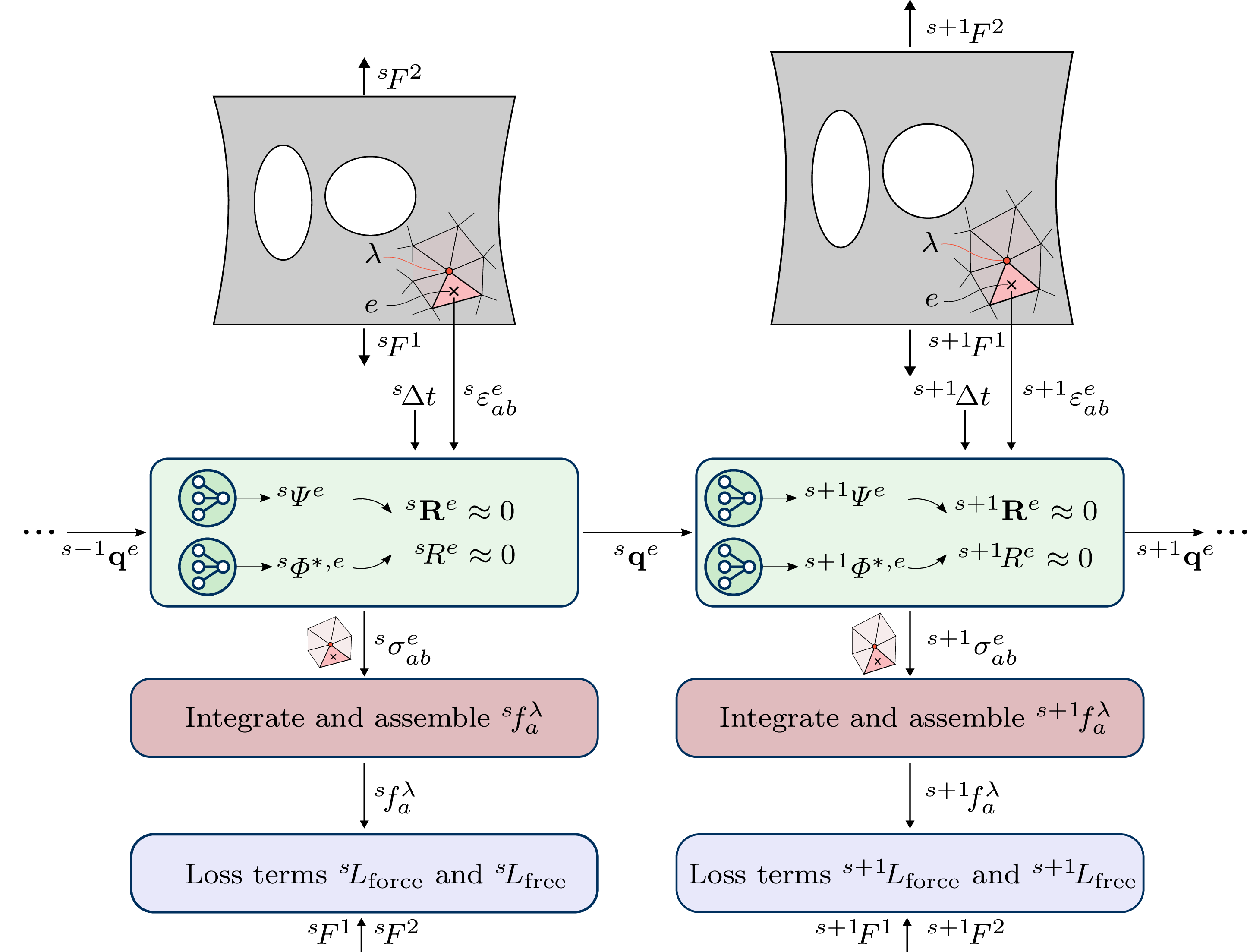}
  \vspace{0.5cm}
  \caption{\emph{Forward pass} and principle calculation of the loss contribution for two adjacent time steps. For a time step $s$, \emph{global reaction forces} $\tvar{F}^\kappa$ and \emph{in-plane strains} $\tvar{\varepsilon}_{ab}^e$ at the quadrature points are provided by the training data, while the internal variable ${}^{s-1}\hspace{-1pt}\te{q}^e$ at the quadrature point is given from the last time step. The nonlinear equations ${}^{s}\te{R}^e = \te{0}$ (Eq.~\eqref{eq:residual_internal_variable}) and $\tvar{R}^e = 0$ (Eq.~\eqref{eq:residual_stress33}), arising from the \emph{implicit time integration} and the \emph{plane-stress condition}, are solved in a Newton iteration, thereby 
		determining the internal variable $\tsvar{\te{q}}^e$ and the out-of-plane strain $\tvar{\varepsilon}_{33}^e$ of the current time step. After doing this for every time step, the \emph{in-plane stresses} $\tvar{\sigma}^e_{ab}$, \emph{assembled nodal forces} $\tvar{f}^\lambda_a$, and finally the two loss contributions are computed in a vectorized manner over all time steps. Functional dependencies are omitted for clarity in this figure.}
  \label{fig:egm}
\end{figure*}

This finally enables us to formulate two key principles of the EGM:
\begin{enumerate}[label=(\alph*)]
    \item The sum of \emph{assembled nodal forces} $\tvar{f}^\lambda_a(\V \theta) n^\kappa_a$ with $\lambda \in\Force^\kappa$ must equal the global reaction force $\tvar{F}^\kappa$, and
    \item \emph{assembled nodal forces} $\tvar{f}^\lambda(\V \theta)$ with $\lambda\in\Free$ must vanish.
\end{enumerate}
Both principles (a) and (b) follow directly from the \emph{discretized weak form of the balance of linear momentum} that yields equilibrium of the forces. 
We fulfill them in a weak sense by determining an optimal parametrization $\V \theta\in\Set{P}_\theta$. To this end, we introduce the two associated loss terms 
\begin{subequations}
    \begin{align}
        L_\text{force}(\V \theta) &:=  
        \sum_{\kappa=1}^{\nf}
        \sum_{s=1}^{\nt}
        \frac{1}{
        n_\kappa
        }
        \left(\tvar{F}^\kappa - \!\!\!\!
        \displaystyle \sum_{\lambda \in \Force^\kappa} \!\!\!\! \tvar{f}^{\lambda}_a(\V \theta)  n^\kappa_a
        \right)^2
        \comma
        \label{eq:loss_forces} \\
        L_\text{free}(\V \theta) &:= \frac{1}{2 \nt n_{\text{free}}}
        \sum_{s=1}^{\nt}
        \sum_{\lambda\in\Free}
        \tvar{f}^\lambda_a(\V \theta) \tvar{f}^\lambda_a(\V \theta) \point \label{eq:loss_domain}
    \end{align}
    \label{eq:loss_parts}%
\end{subequations}
Therein, the normalization factor for the force losses is defined as $n_\kappa := \nt \max_{s}\left| \tvar{F}^{\kappa} \right|^2$ and $n_{\text{free}}=|\Free|$. With the two introduced loss terms we arrive at the total loss given by $L_{\textrm{EGM}}(\V \theta) := L_\text{force}(\V \theta) + L_\text{free}(\V \theta)$. Based on this, we formulate the constrained optimization problem 
\begin{align}
    \V{\hat{\theta}} = \underset{\V{\theta} \in\Set{P}_{\theta}}{\textrm{argmin }} L_{\textrm{EGM}}(\V \theta)
    \;\textrm{subject to}\;
    {}^{s}\te{R}^{e} = \te{0} \textrm{ and }
    \tvar{R}^{e} = 0 
    \label{eq:optimization_problem}
\end{align}
for all $s\in\{1,\ldots,\nt\}$, $e\in\{1,\ldots,\nele\}$. The optimization problem given in Eq.~\eqref{eq:optimization_problem} is constrained by the non-negativity bounds of the weights in the PANN and the two residuals resulting from the time-discretized evolution equations and the \emph{plane stress condition}. 

In order to obtain the optimal parametrization $\V{\hat{\theta}}\in\Set{P}_{\theta}$, and hence the calibrated material model, we use the quasi-Newton optimizer SLSQP (sequential least squares programming). The parameter constraints $\Set{P}_{\theta}$ are treated as bounds by the optimizer. The constraints imposed by the evolution equations and the \emph{plane stress condition} are incorporated by determining, at each epoch of the optimizer, the out-of-plane strains $\tvar{\varepsilon}_{33}^e(\V \theta)$ and the internal variables ${}^{s}\te{q}^e(\V \theta)$ at all quadrature points and for all time steps. These constraints are then taken into account when calculating the gradient $\mathrm d L_{\textrm{EGM}}(\V \theta)/ \mathrm{d}\V{\theta}$ via automatic differentiation and the \emph{backward adjoint method}~\citep{Seidl2022,Kumar2025,Akerson2025}. Details on the calculation of this gradient are provided in Appendix~\ref{app:backward_adjoint}.
\begin{remark}
    Although the choice of activation functions and the imposed parameter constraints of the PANNs ensure physical admissibility of the learned potentials, the resulting constitutive equations remain highly nonlinear. As a consequence, numerical difficulties may arise during the solution of the associated evolution equations during training.
    Thus, should convergence not be reached within the prescribed iteration limit when solving the optimization problem~\eqref{eq:optimization_problem}, the iterates ${}^{s,\niter}\te{q}^e$ and ${}^{s,\niter}\varepsilon^{e}_{33}$ are nonetheless accepted as converged, i.e., ${}^{s}\te{q}^e := {}^{s,\niter}\te{q}^e$ and $\tvar{\varepsilon}_{33}^e := {}^{s,\niter}\!\varepsilon_{33}^e$. 
    In this sense, we also set the two residuals to ${}^{s}\te{R}^{e}:= {}^{s,\niter}\te{R}^{e}$ and ${}^{s}\!R^{e}:= {}^{s,\niter}\!R^{e}$.
    This is a deliberate design choice that permits the evolution of the internal variable to proceed even when the parametrization of the potentials described by the PANNs are furnished in a numerically unfavorable form. 
    For instance, for material parameters that preclude convergence. As described in \citep[Remark~10]{kalina_physics-augmented_2026}, such issues can arise in ODE-constrained optimization problems in the context of parameterizing PANNs.
\end{remark}
\paragraph{Regularization for non-ideal data}
The proposed loss contributions in Eq.~\eqref{eq:loss_parts} work very well if the input data is ideal, i.e., obtained from a simulation. In real experiments, however, forces and displacements are subject to measurement noise.
Non-ideal displacement data poses a considerable difficulty for the EGM, since the strains, or more generally the first spatial derivative, must be computed from the displacements, which considerably amplifies the measurement error. This amplified error then propagates through the computation of stresses and nodal forces, rendering the objective pursued by the loss term $L_\text{free}(\V \theta)$ questionable, as local equilibrium cannot be achieved by a physically reasonable constitutive model.\footnote{In theory, there might exist a material model with a huge number of parameters that achieves almost perfect equilibrium of nodal forces, even with noisy input data. This, however, is not a desirable outcome, since it gives rise to overfitting, i.e., learning the noise rather than the material behavior.}
The loss contribution $L_\text{force}(\V \theta)$, on the other hand, does not enforce local equilibrium, but equilibrium with the applied \emph{global reaction forces}. 

Due the summation, errors in the individual \emph{assembled nodal forces} $\tvar{f}^\lambda_a$ belonging to free nodes, i.e., $\lambda\in\Free$, are expected to average out in the loss given in Eq.~\eqref{eq:loss_domain}. Nevertheless, we apply a method introduced by \citet{jailin_noise-bias_2026} to formulate a modified loss contribution $L_\text{free,mod}(\V \theta)$ and thereby make the proposed optimization problem more robust against measurement noise. 
Please note that there is no change to the loss contribution $L_\text{force}(\V \theta)$.
The key idea is to apply a \emph{Gaussian kernel matrix} $\M G\in\Setnum{R}_{\ge 0}^{\nn \times \nn}$ to the \emph{assembled nodal forces}, whose entries $G^{\lambda\mu}\in\Setnum{R}_{\geq0}$ depend on the distance between the initial node positions $x^\lambda_a$. The modified \emph{assembled nodal forces} are thereby obtained as
\begin{align}
    \tvar{f}^\mu_{a,\textrm{mod}}(\V \theta) :\!\!&=
    \sum_{\lambda=1}^{\nn} \tvar{f}^\lambda_a(\V \theta) G^{\lambda\mu} \\
    &=
    \sum_{\lambda=1}^{\nn}
    \tvar{f}^\lambda_a(\V \theta)
    \frac{\exp\left( -\dfrac{(x^\mu_b - x^\lambda_b)(x^\mu_b - x^\lambda_b)} {\delta^2} \right)}
    { \displaystyle\sum\limits_{\rho=1}^{\nn} \exp\left( -\dfrac{(x^\mu_c - x^\rho_c)(x^\mu_c - x^\rho_c)} {\delta^2} \right)}
     \point \nonumber
    \label{eq:modified_forces}
\end{align}
The factor $\delta=\alpha_{\textrm{infl}}\Delta l$ is computed from the average initial distance $\Delta l\in\Setnum{R}_{>0}$ between the nodes of all elements and an influence factor $\alpha_{\textrm{infl}}\in\Setnum{R}_{>0}$.
After that, the modified \emph{assembled nodal forces} belonging to the free nodes with indices $\mu\in\Free$ are used to obtain the modified loss term
\begin{equation}
    L_\text{free,mod}(\V \theta) := \frac{1}{2 \nt n_{\textrm{free}}}
    \sum_{t=1}^{\nt} \!
    \sum_{\mu\in\Free} \!\!\!\!
    \tvar{f}^\mu_{a,\textrm{mod}}(\V \theta) \tvar{f}^\mu_{a,\textrm{mod}}(\V \theta) \point
    \label{eq:loss_domain_modified}
\end{equation}
\section{Numerical examples}
\label{sec:numerical_examples}
To illustrate the performance of the proposed framework for training viscoelastic PANNs in an unsupervised manner via full-field data, we consider several numerical examples, including ideal and noisy data sets, in the following.

In all numerical studies, each of the three element-wise non-decreasing FICNNs with skip connections consisted of a single hidden layer with two neurons and softplus activation functions. To ensure the desired properties (element-wise non-decreasingness and convexity), non-negativity constraints were imposed on the weights in the hidden layer and the skip connections~\citep{Klein2021,amos_input_2017}.
Because the resulting number $n_\theta$ of trainable parameters, $\V{\theta}\in\Set{P}_{\theta}$, was fairly small, we chose the SLSQP optimizer, which performs better than the ADAM optimizer---widely used in machine learning---for small to moderate network sizes, cf.~\citep[Appendix~G]{Kalina2024}.

To improve the efficacy of the training process, calibration was split into two steps: a reduced dataset comprising only the first 50 time steps was used to calibrate the PANN initially, after which the complete dataset was used for training.
In both training steps, the convergence criterion, i.e., the function tolerance for SLSQP, was set to the extremely small value of $10^{-32}$, thus effectively limiting the optimization by the number of function evaluations instead. As a result, calibration was terminated almost exclusively by reaching the maximum number of function evaluations, rather than by satisfying the convergence criterion. For the first training step, using 50 time steps, the number of function evaluations was set to 2000, while for the subsequent step, using all time steps, it was set to 1000.

The two convergence criteria in the combined Newton iteration, employed to compute $\tsvar{\te{q}}^e$ and $\tvar{\varepsilon}_{33}^e$, were defined by $\epsilon^{\te{q}}_{\textrm{tol}}=10^{-6}$ and $\epsilon^\sigma_{\textrm{tol}}=10^{-12}$. The maximum number of allowed Newton iterations was set to $\niter=20$.

For all considered examples, training was conducted for ten different initializations. Subsequently, the model with the least training loss was selected and validated in each case. The implementation of the PANN model and the unsupervised calibration workflow was realized using Python, TensorFlow, and SciPy.\footnote{Once the final version of the article is published, the code will be made publicly available.}
All computations were performed on an HPC cluster at TU Dresden; hardware details are given later.
\begin{remark}
    It is known from the literature that training inelastic neural models via ODE-constrained optimization, as done in this work, can fail whenever an unfavorable set of the trainable parameters causes quantities such as stresses to drift toward infinity, resulting in NaN values, cf.~\citep[p.~20]{Holthusen2026a} or \citep[Remark~10]{kalina_physics-augmented_2026}. The SLSQP optimizer used here is comparatively robust against such scenarios. If the optimizer receives the value NaN, it can switch to a stable region of the parameters $\V \theta \in \Set{P}_{\theta}$.
     Only if NaN values occur in the very first training epoch does SLSQP terminate immediately. One way to prevent NaNs in this initial epoch is to choose the trainable parameters such that the initial relaxation times take on meaningful values~\citep[Remark~9]{kalina_physics-augmented_2026}. For the cases considered here, however, it was sufficient to restrict the initialization of the weights to the interval $[0,0.1]$.
\end{remark}
\subsection{Generation of data}
\label{ssec:generation_of_data}
To calibrate the PANN in an unsupervised manner using the EGM, synthetic training data, namely \emph{in-plane displacements} and \emph{global reaction forces} normal to the boundaries, are generated via finite element simulations. These simulations were performed by using an in-house finite element code.
\subsubsection{Ground truth viscoelastic model}
\label{sssec:gt_model}
\begin{figure*}[p]
  \centering
  \vspace{-2pt}
  \includegraphics[clip]{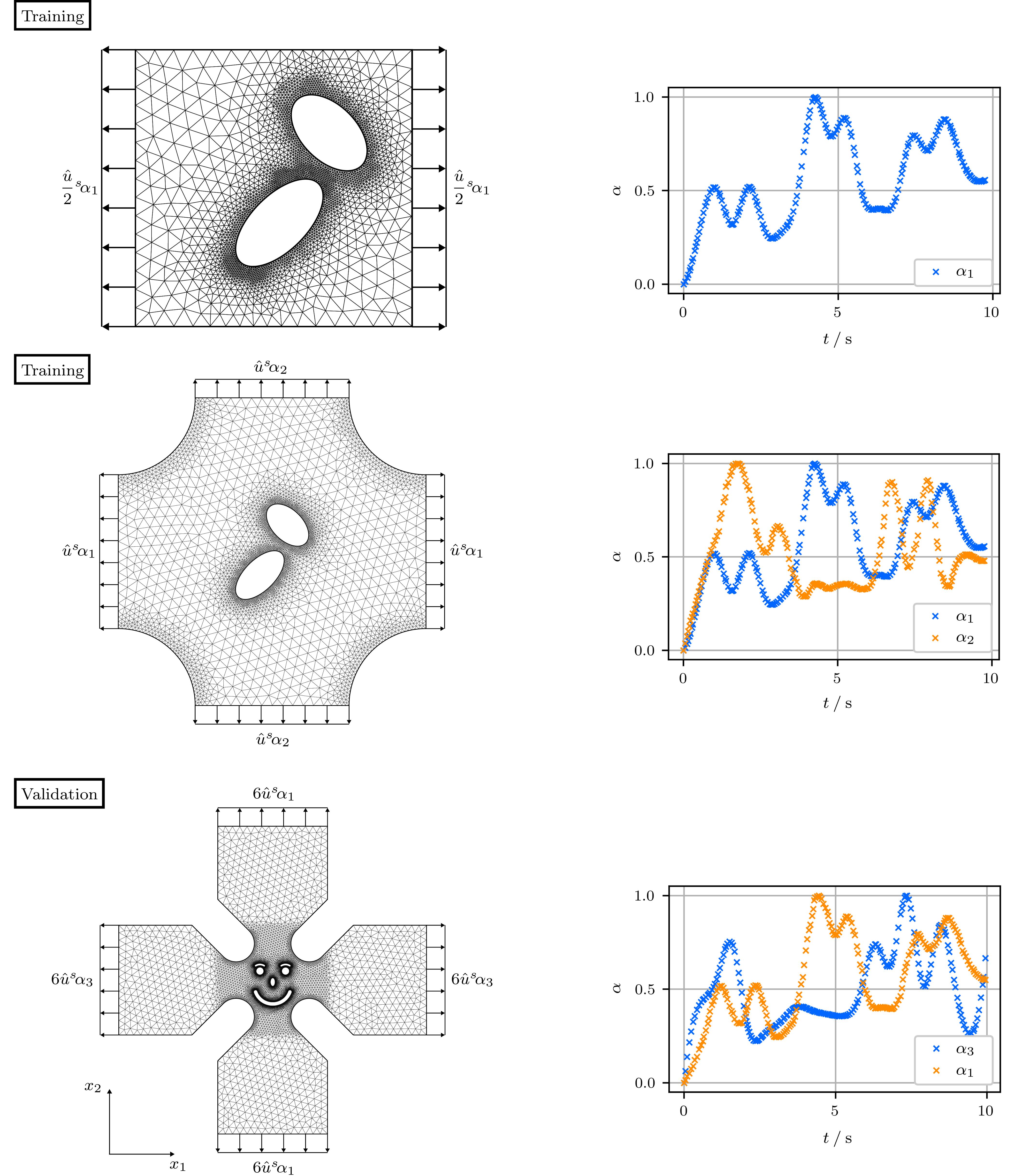}
  \caption{\textbf{Training:} Uniaxial specimen with outer dimensions $\SI{100}{\milli\meter}\times\SI{100}{\milli\meter}$, discretized into 5063 CST elements,
  and biaxial specimen with outer dimensions $\SI{200}{\milli\meter}\times\SI{200}{\milli\meter}$, discretized into 6253 CST elements, from which training data was generated using the linear viscoelastic reference model given in Eqs.~\eqref{eq:gt_1}--\eqref{eq:gt_5}.
  The prescribed displacement $\hat{u}=\SI{1}{\milli\meter}$ on a respective boundary was multiplied by a scaling factor $\tvar{\alpha}_1$ or $\tvar{\alpha}_2$, $s\in\{1,2,\dots,199\}$, as shown on the right of the respective specimen. 
  \textbf{Validation:} The biaxial specimen at the bottom, with outer dimensions $\SI{841.4}{\milli\meter}\times\SI{841.4}{\milli\meter}$, is discretized into 17277 quadratic triangular elements and used for validation. In this case, the scaling factors $\tvar{\alpha}_1, \tvar{\alpha}_3$, $s\in\{1,2,\dots,194\}$ are also shown on the right.}
  \label{fig:synthetic_specimens}
\end{figure*}

The ground truth (GT) material model employed is a classical linear viscoelastic Maxwell model with one internal variable, given by
\begin{subequations}
    \begin{align}
        \varPsi^{\textrm{GT}}(\ve{\upvarepsilon},\te{q}) \!&:= \!\frac{1}{2}\ve{\upvarepsilon}\!:\!\mathbb{C}^{\textrm{eq}}\!:\!\ve{\upvarepsilon} \!
        +\! \frac{1}{2}(\ve{\upvarepsilon} \!-\! \te{q})\!:\!\mathbb{C}^{\textrm{neq}}\!:\!(\ve{\upvarepsilon} \!-\! \te{q}) \label{eq:gt_1}\\
        \mathbb{C}^{\textrm{eq}} &:= 3K\mathbb{I}^{\textrm{sph}} + 2G\mathbb{I}^{\textrm{dev}} \\
        \mathbb{C}^{\textrm{neq}} &:= 3K_0\mathbb{I}^{\textrm{sph}} + 2G_0\mathbb{I}^{\textrm{dev}} \\
        \varPhi^{*\textrm{GT}}(\te{A}) &:= \frac{1}{2}\te{A}:\mathbb{V}^{-1}:\te{A}\\
        \mathbb{V}^{-1} &:= \frac{1}{3\eta^{\textrm{sph}}}\mathbb{I}^{\textrm{sph}} + \frac{1}{2\eta^{\textrm{dev}}}\mathbb{I}^{\textrm{dev}} \label{eq:gt_5}
    \end{align}
\end{subequations}
and material parameters $K = \SI{5}{\mega\pascal}$, $G = \SI{3}{\mega\pascal}$, $K_0 = \SI{10}{\mega\pascal}$, $G_0 = \SI{7}{\mega\pascal}$, $\eta^{\textrm{sph}} = \SI{4}{\mega\pascal\second}$ and $\eta^{\textrm{dev}} = \SI{2}{\mega\pascal\second}$.
\subsubsection{Training data}
\label{ssec:descr_training_data}
The \emph{uniaxial and biaxial specimens} used to generate the training data are shown in Fig.~\ref{fig:synthetic_specimens}; in each specimen, two ellipses were cut out to produce a wide range of strain and stress states~\citep{linden_dual-stage_2025}.
Each specimen has an initial thickness $h=\SI{1}{\milli\meter}$.
For the uniaxial specimen, displacements in the $x_1$-direction were applied at the left and right edges in the form of \emph{random walks}, so as to obtain as many different strain rates as possible during the synthetic experiment~\citep{Asad2023,rosenkranz_viscoelasticty_2024}. These \emph{random walks} were obtained by means of a scaling factor $\tvar{\alpha}_1 \in [0,1]$, $s\in\{1,\ldots,\nt\}$, that was varied in each time step. The temporal course was obtained using cubic splines, following the approach of \citet{rosenkranz_viscoelasticty_2024}; see Fig.~\ref{fig:synthetic_specimens}.
Clamping in a testing machine was then simulated by additionally locking the displacement in the $x_2$-direction at the same edges.
The same principle applied to the biaxial specimen, where clamping was simulated for each of the four arms: the displacement perpendicular to a given edge followed a \emph{random walk}, as in the uniaxial case, while the displacement parallel to it was locked. For the left arm, for instance, this meant a \emph{random walk} in the $x_1$-direction combined with a locked displacement in the $x_2$-direction.
\subsubsection{Validation data}
\label{ssec:generation_validation}
For validation, the synthetic biaxial experiment shown in Fig.~\ref{fig:synthetic_specimens}, which was not part of the training data, was considered. In addition, quadratic triangular elements were used to discretize the validation specimen. Again, the reference thickness is given by $h=\SI{1}{\milli\meter}$. The trained PANN was embedded in a finite element code to solve the corresponding boundary value problem; the resulting stresses are compared against those obtained with the linear viscoelastic reference model in the following sections.
\subsection{Training with ideal data}
\label{ssec:ideal_data}
\begin{figure*}[ht]
  \centering
  \includegraphics[clip, trim={0cm 0.25cm 0cm 0cm}]{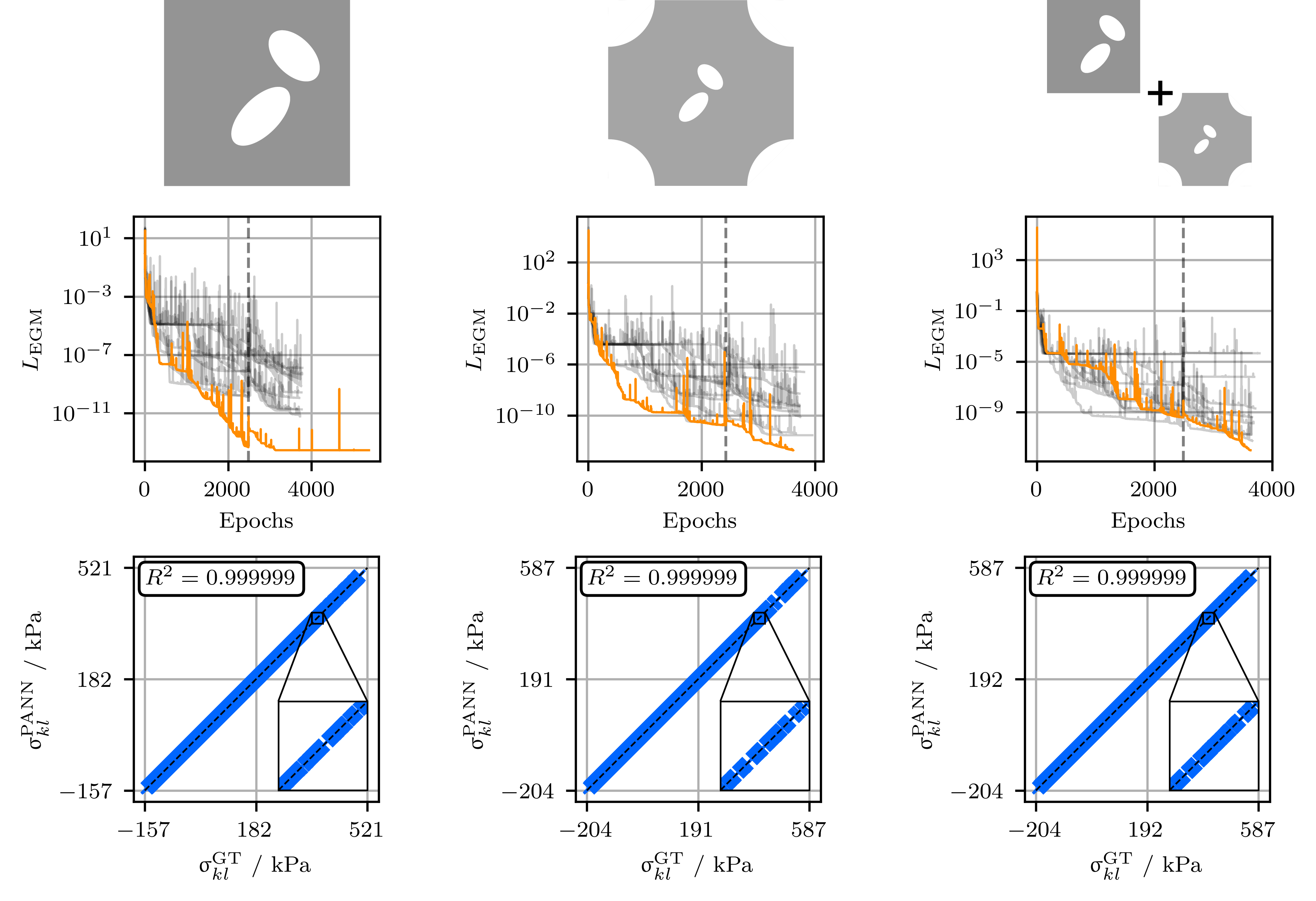}
  \caption{Loss courses during calibration of the PANN, for all ten initializations, using ideal training data generated from uniaxial and biaxial training specimens (see Fig.~\ref{fig:synthetic_specimens}), respectively. The loss with the smallest final value per training data set is visualized in orange. In each column, the prediction of the calibrated PANN with the smallest associated loss for $\sigma_{kl}$ is plotted against the reference values calculated with the GT material model. Additionally, the commonly used $R^2$-value is depicted. For details on the computation of the $R^2$-value please refer to \citet{dammas_when_2025}.}
  \label{fig:training_results_r2_loss}
\end{figure*}

To review the performance of the optimization framework in conjunction with the PANN, ideal training data (without added artificial noise) were used first. Such data, obtained directly from the synthetic experiments, are not subject to the uncertainties present in real experiments; the displacements and global reaction forces computed with the finite element code using the reference model were therefore provided directly for training. The effect of noisy data is addressed in Sect.~\ref{ssec:noisy_data}. Three datasets were considered for training:
\begin{itemize}
    \item data from the uniaxial experiment in Fig.~\ref{fig:synthetic_specimens},
    \item data from the biaxial experiment in Fig.~\ref{fig:synthetic_specimens}, and
    \item combined data from both the uniaxial and biaxial experiments in Fig.~\ref{fig:synthetic_specimens}.
\end{itemize}
Hardware details and approximate computation times per training epoch are listed in Tab.~\ref{tab:hardware_details}.
\begin{table}[ht]
	\centering
	\caption{Number of cores, approximate RAM usage, and approximate computation times per epoch to calibrate the PANN for all of the three training data sets. The latter is specified for the approximate computation times per epoch for the pre-training with only 50 time steps, and the subsequent training with all 199 time steps. All trainings have been done on a machine with Intel Xeon Platinum 8470.}
    \vspace{6pt}
	\label{tab:hardware_details}
	\renewcommand{\arraystretch}{1.}
    \begin{small}
    	\begin{tabularx}{0.48\textwidth}{c|ccc}
        \toprule
            Training data set & Uniaxial & Biaxial & Uniaxial + Biaxial \\
            \midrule
            Number of cores & 40 & 40 & 50 \\[4pt]
            \shortstack[c]{Approx. maximum \\ RAM usage} & $\SI{15}{\giga\byte}$ & $\SI{16}{\giga\byte}$ & $\SI{24}{\giga\byte}$ \\[4pt]
            \shortstack[c]{Approx. $\Delta t_{\textrm{epoch}}$ \\ using 50 time steps}  & $\SI{5}{s}$ &  $\SI{5}{s}$ &  $\SI{8}{s}$ \\[4pt]
            \shortstack[c]{Approx. $\Delta t_{\textrm{epoch}}$ \\ using 199 time steps} & $\SI{16}{s}$ &  $\SI{17}{s}$ &  $\SI{27}{s}$ \\
    		\bottomrule
    	\end{tabularx}
    \end{small}
\end{table}
\begin{remark}
    We elaborate briefly on some technicalities regarding the results in Tab.~\ref{tab:hardware_details} in order to put them into context. First, note that the main computation time per training epoch occurs in the forward pass, due to the Newton iteration. Because of the \emph{adjoint method}, the \emph{backward pass} is comparatively less demanding. In addition, different number of Newton iterations also change the required time per training epoch. A coarser mesh would also drastically reduce the computation time. Furthermore, the values reported in Tab.~\ref{tab:hardware_details} for maximum RAM usage and computation time per training epoch vary between different initializations of the PANN, even for the same training data. These values are therefore marked as approximate. Finally, even for the same initialization of the PANN and identical training data, different training results can be obtained: while the loss is identical during the first training epoch, it gradually drifts apart in subsequent epochs. Since all computations are performed on a machine using 40 to 50 cores, this behavior can be attributed to the non-associativity of floating-point operations combined with the non-deterministic order of parallel execution~\citep{collange2015reproducibility}.
\end{remark}
\subsubsection{Training results}
\label{sssec:training_results}
\begin{figure}[ht]
  \centering
  \includegraphics[clip, trim={0.2cm 0.0cm 0cm 0cm}]{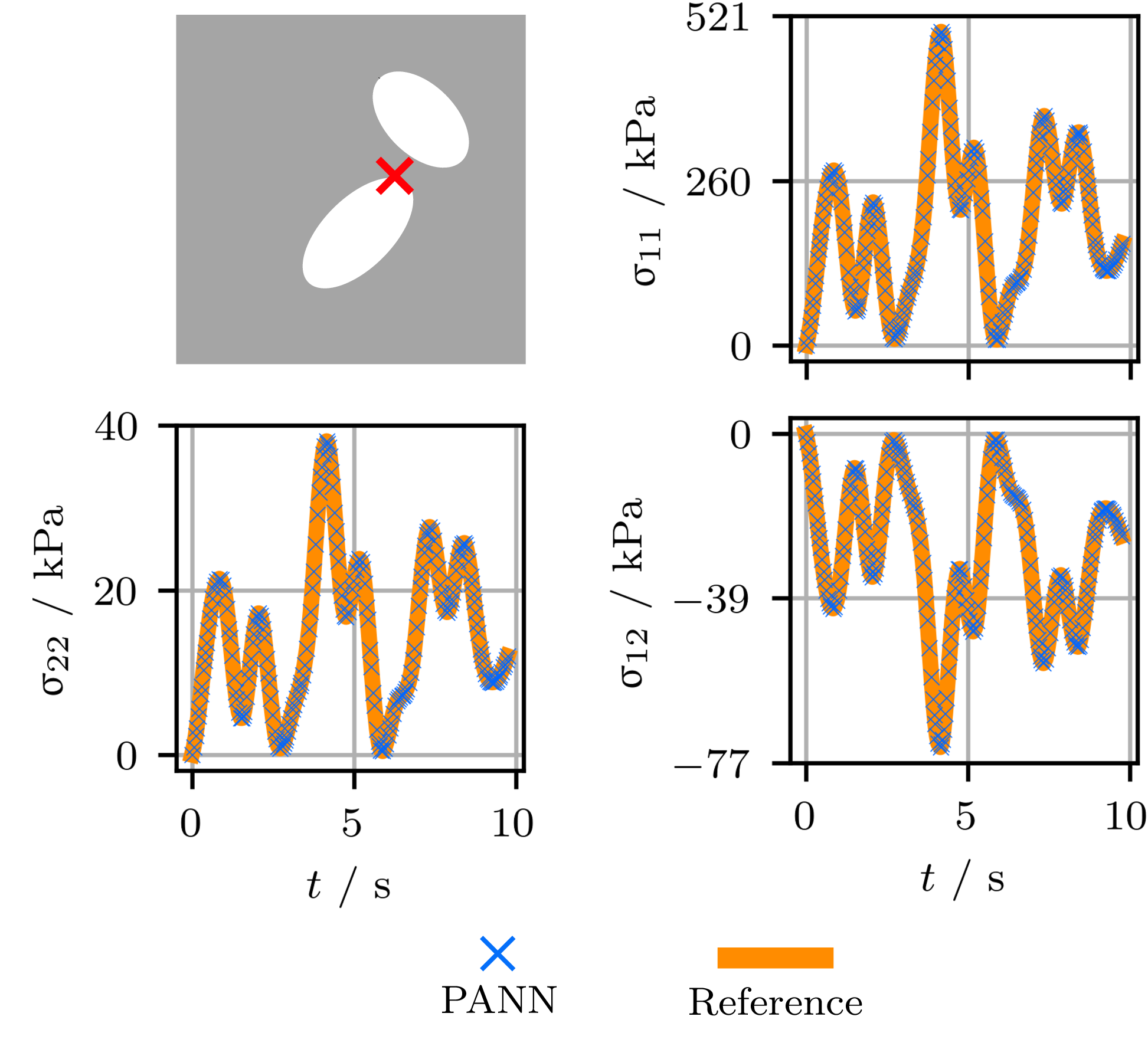}
  \vspace{-8pt}
  \caption{Prediction of stress components by the PANN and reference values of the linear viscoelastic model at the marked quadrature point. The PANN was trained using ideal data obtained from the uniaxial data, see Fig.~\ref{fig:synthetic_specimens}.}
  \label{fig:training_results_sigma}
\end{figure}
\begin{figure}[ht]
  \includegraphics[clip, trim={0.2cm 0.0cm 0cm 0cm}]{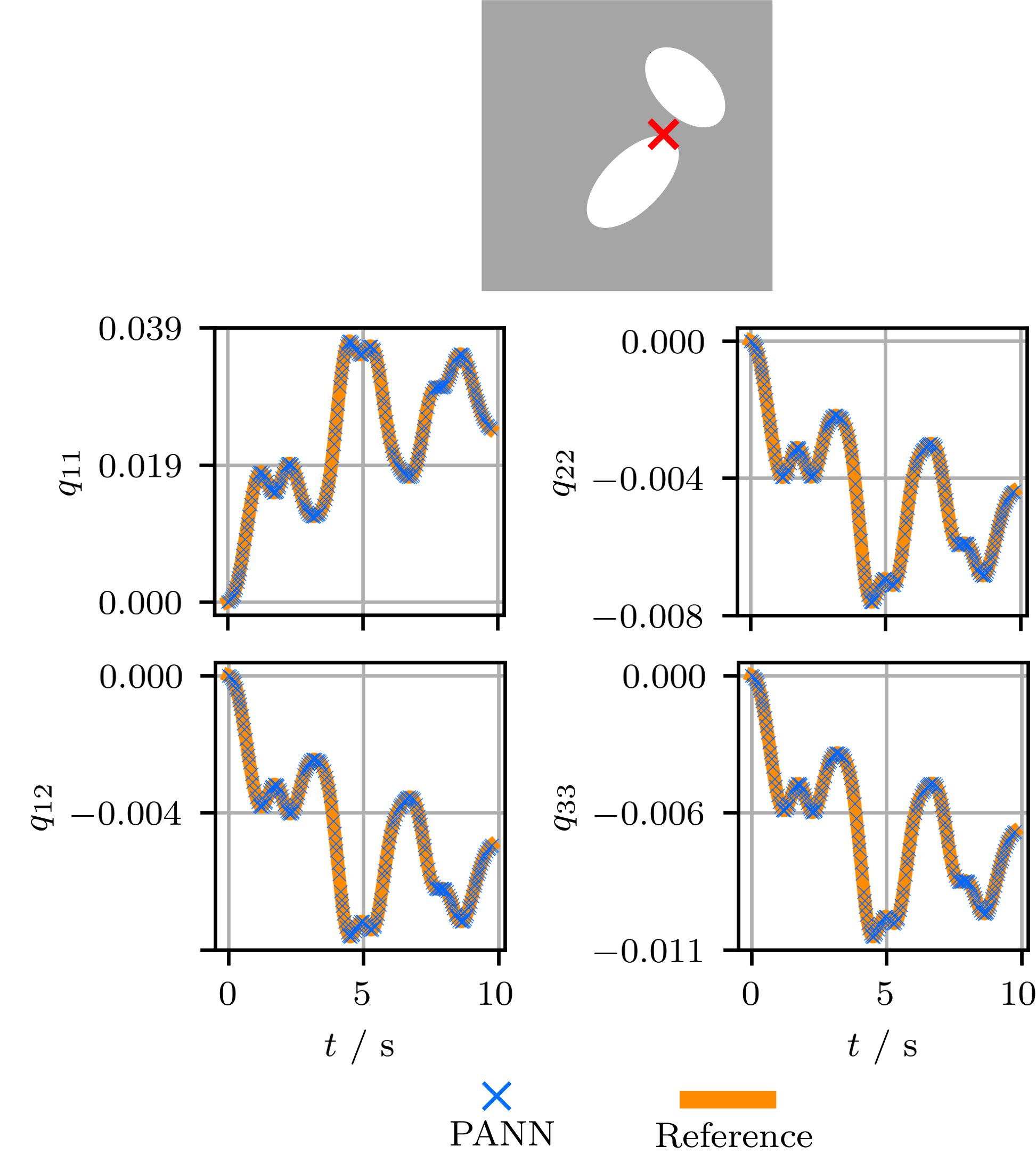}
  \caption{Prediction of internal variables by the PANN and reference values of the linear viscoelastic model at the marked quadrature point. The PANN was trained using ideal data obtained from the uniaxial data, see Fig.~\ref{fig:synthetic_specimens}.}
  \label{fig:training_results_q}
\end{figure}
The SLSQP optimization terminated without failure for all three training data sets and all ten random initializations, yielding ten trained PANNs per data set. However, the quality of the resulting calibration varied between initializations. For each data set, the PANN exhibiting the smallest final loss was therefore selected as the best-performing calibration. The corresponding losses are shown in Fig.~\ref{fig:training_results_r2_loss}.
The stresses predicted by these selected PANNs agree very well with the reference stresses calculated using the GT material model, as shown in the $R^2$ plots provided in Fig.~\ref{fig:training_results_r2_loss}.
In more detail, Fig.~\ref{fig:training_results_sigma} depicts the predicted stress components over time at a representative quadrature point of the uniaxial specimen, all of which are fitted almost perfectly.

Analogously, Fig.~\ref{fig:training_results_q} shows that the coordinates of the internal variables are likewise predicted perfectly by the PANN when compared to the GT values.
All remaining quantities, such as the thermodynamic driving forces, are predicted with comparable accuracy. This level of agreement held for all three training data sets.
\subsubsection{Validation results}
\label{sssec:validation_results}
As discussed in Sect.~\ref{ssec:generation_validation}, we validated a trained PANN by embedding it as the constitutive model in a finite element code and comparing the resulting stress response to that computed with the linear viscoelastic ground truth model. The boundary value problem considered was the validation specimen shown in Fig.~\ref{fig:synthetic_specimens}. We carried out the simulation for the three best-performing PANNs, selected according to their respective loss values, obtained from the three training data sets.

Compared to the linear viscoelastic reference model, which achieves convergence in just a single global Newton iteration for the boundary value problem under consideration, the PANN required around 6 to 10 iterations. This is an inevitable consequence of the nonlinearity introduced by activation functions within the FICNNs used to describe the potentials within the PANN. Aside from this observation, the difference in the predicted stresses between the PANN and the GT model is negligible: Fig.~\ref{fig:validation_results_paraview} illustrates this for a representative time step and the stress component $\sigma^{\textrm{PANN,ideal}}_{11}$, for all three training data sets.
All other quantities, including the remaining stress components, the internal variable, and the displacements, are likewise predicted in very good agreement with the GT model.
\begin{remark}
    Fig.~\ref{fig:validation_results_paraview} shows that the PANN can be trained from all three data sets almost perfectly. Therefore, an experiment with either a uniaxial or biaxial specimen seems sufficient. However, as elucidated by \citet{dammas_when_2025}, this does not apply to the finite strain case. Hence, both uniaxial and biaxial stress states are most likely necessary when calibrating from real experimental data at finite deformations. Note that uniaxial stress states are present in the biaxial specimen at the edges of the holes and in the arms for load application.
\end{remark}

\begin{figure*}[p]
  \centering
  \includegraphics[clip]{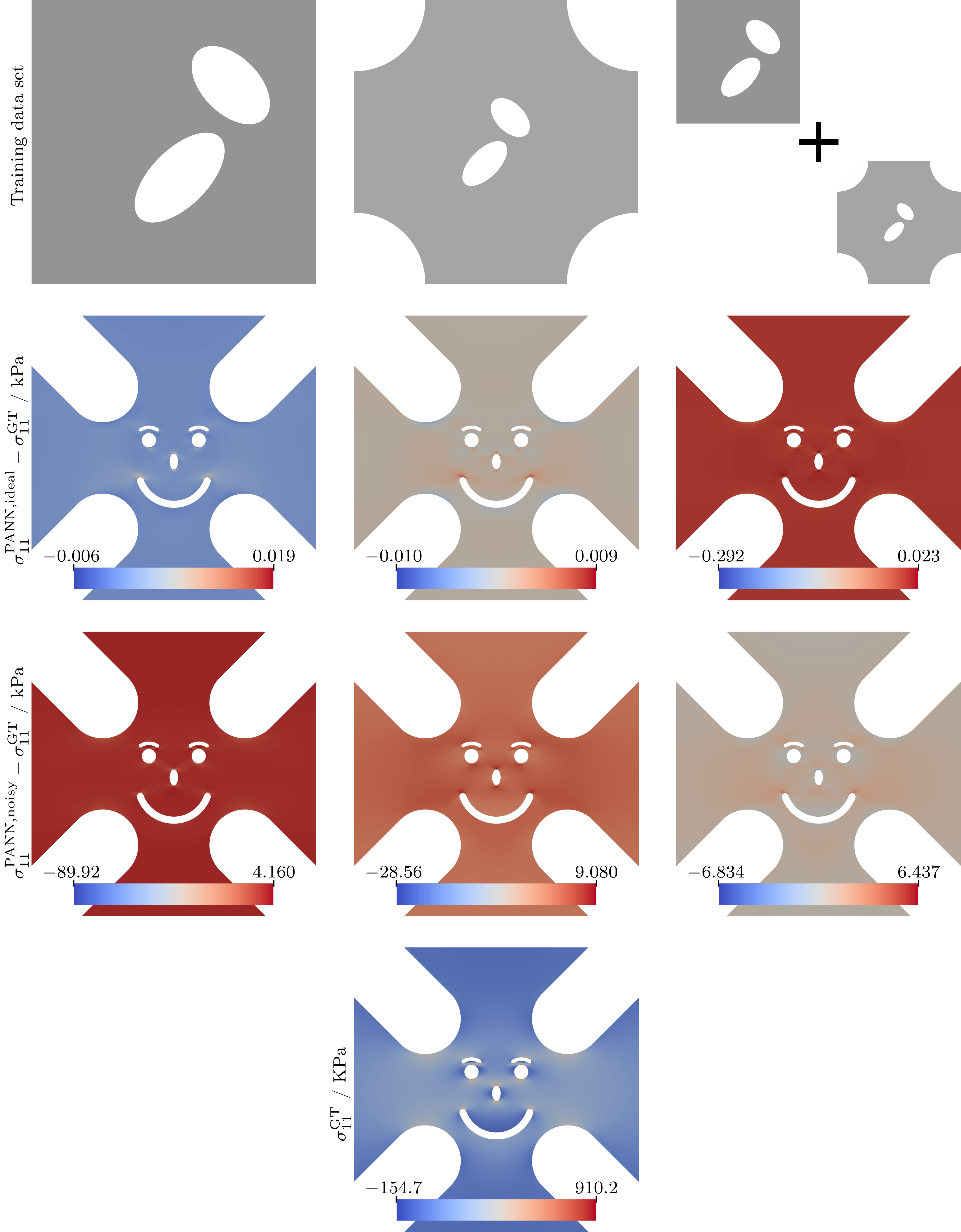}
  \caption{Stress differences $\sigma_{11}^{\textrm{PANN,ideal}} - \sigma_{11}^{\textrm{GT}}$ and $\sigma_{11}^{\textrm{PANN,noisy}} - \sigma_{11}^{\textrm{GT}}$ at time step $s=141$ of the validation boundary value problem in Fig.~\ref{fig:synthetic_specimens} between the calibrated PANN and the linear viscoelastic reference model (GT model). At the bottom the stress $\sigma_{11}^{\textrm{GT}}$ of the GT model is shown. The PANN was trained with both ideal and noisy data obtained from the uniaxial experiment, the biaxial experiment, and the combined data of both synthetic experiments, respectively.}
  \label{fig:validation_results_paraview}
\end{figure*}

\subsection{Training with noisy data}
\label{ssec:noisy_data}
Compared to synthetic experiments, real experiments are always subject to uncertainty, since, for example, measured forces and displacements are corrupted by noise. This scenario is considered in the following.

\subsubsection{Generation of artificial noise}
\label{sssec:generation_of_noise}
 To mimic the situation in a realistic experimental setup, artificial noise was added to the ideal training data, i.e., to the \emph{in-plane displacements} and \emph{global reaction forces} on the respective boundaries. We therefore adopted the method proposed by \citet{linden_dual-stage_2025} to impose artificial noise on full-field data.
\paragraph{Noisy reaction forces}
First, a noisy reaction force $\tvar{\tilde{F}}^\kappa$, $\kappa\in\{1,\dots,n_\kappa\}$ was obtained via
\begin{equation}
    \tvar{\tilde{F}}^\kappa := \tvar{F}^\kappa \tvar{r}^\kappa
    \quad\textrm{with}\quad
    \tvar{r}^\kappa \sim \mathcal{U}^\kappa[1-10^{-4};\,1+10^{-4}] \comma
    \label{eq:noisy_forces}
\end{equation}
where $\mathcal{U}^\kappa[1-10^{-4};1+10^{-4}]$ denotes a uniform distribution.
\begin{remark}
    For the uniaxial specimen, the same sample $\tvar{r}^\kappa$ was chosen for the two reaction forces on the left and right boundary, since both are identical in magnitude. Consequently, in a real experiment it would suffice to measure the reaction force on only one side. This does not hold, however, for the biaxial specimen with clamping boundary conditions, where the reaction forces differ between opposite sides of the specimen, making a measurement on each of the four arms necessary. Accordingly, four different sets of samples $\tvar{r}^\kappa$ were used to impose noise on the \emph{global reaction forces} of the biaxial specimen.
\end{remark}
\paragraph{Noisy in-plane displacements}
The \emph{in-plane nodal displacements} $\tsvar{u}^\lambda_a$ were perturbed using distinct normalized Gaussian random fields (GRFs) for every time step $s$ and coordinate $a$. These GRFs were generated on a $2048\times 2048$ square grid with correlation length $l=\frac{1}{2048}$. This grid was then interpolated onto the physical two-dimensional domain, defined by the minimum and maximum nodal positions $x^\lambda_a + \tsvar{u}^\lambda_a$ in the $x_1$- and $x_2$-direction across all time steps. Subsequently, the respective GRF $\tsvar{g}_a(x^\lambda_a + \tsvar{u}^\lambda_a)$ was evaluated at the current positions of the global nodes. The \emph{noisy in-plane nodal displacements} $\tsvar{\tilde{u}}^\lambda_a$ are then defined by
\begin{equation}
    \tsvar{\tilde{u}}^\lambda_a := \tsvar{u}^\lambda_a + \tsvar{g}_a(x^\lambda_b + \tsvar{u}^\lambda_b) \Delta \ell \eta \comma
    \label{eq:noisy_displacements}
\end{equation}
where $\Delta \ell\in\Setnum{R}_{>0}$ denotes the length scale and $\eta\in\Setnum{R}_{\geq0}$ controls the amplitude of the noise. The length scale was set to $\Delta \ell = \SI{100}{\milli\meter}$ for the uniaxial specimen and to $\Delta \ell = \SI{200}{\milli\meter}$ for the biaxial training specimen, while the amplitude was set to $\eta = 5\times10^{-5}$ in both cases. To illustrate the influence of the noise on the resulting strains, the strain component $\varepsilon_{11}$ of the uniaxial specimen is considered. Fig.~\ref{fig:ideal_and_noisy_eps11} shows this component at a representative time step for both ideal and noisy training data.
\paragraph{Influence length for assembled nodal forces}
Since in the case of non-ideal data, the loss is adopted slightly by spatial regularization of the \emph{assembled nodal forces}, we have to choose a hyper parameter for this method, see Sect.~\ref{ssec:optimization problem}. To this end, we chose the influence factor $\alpha_{\textrm{infl}}=0.8$, which has shown to yield robustness of the optimization procedure.
\begin{figure}[ht]
  \centering
  \includegraphics[clip, trim={0.0cm 0.0cm 0.0cm 0.0cm}]{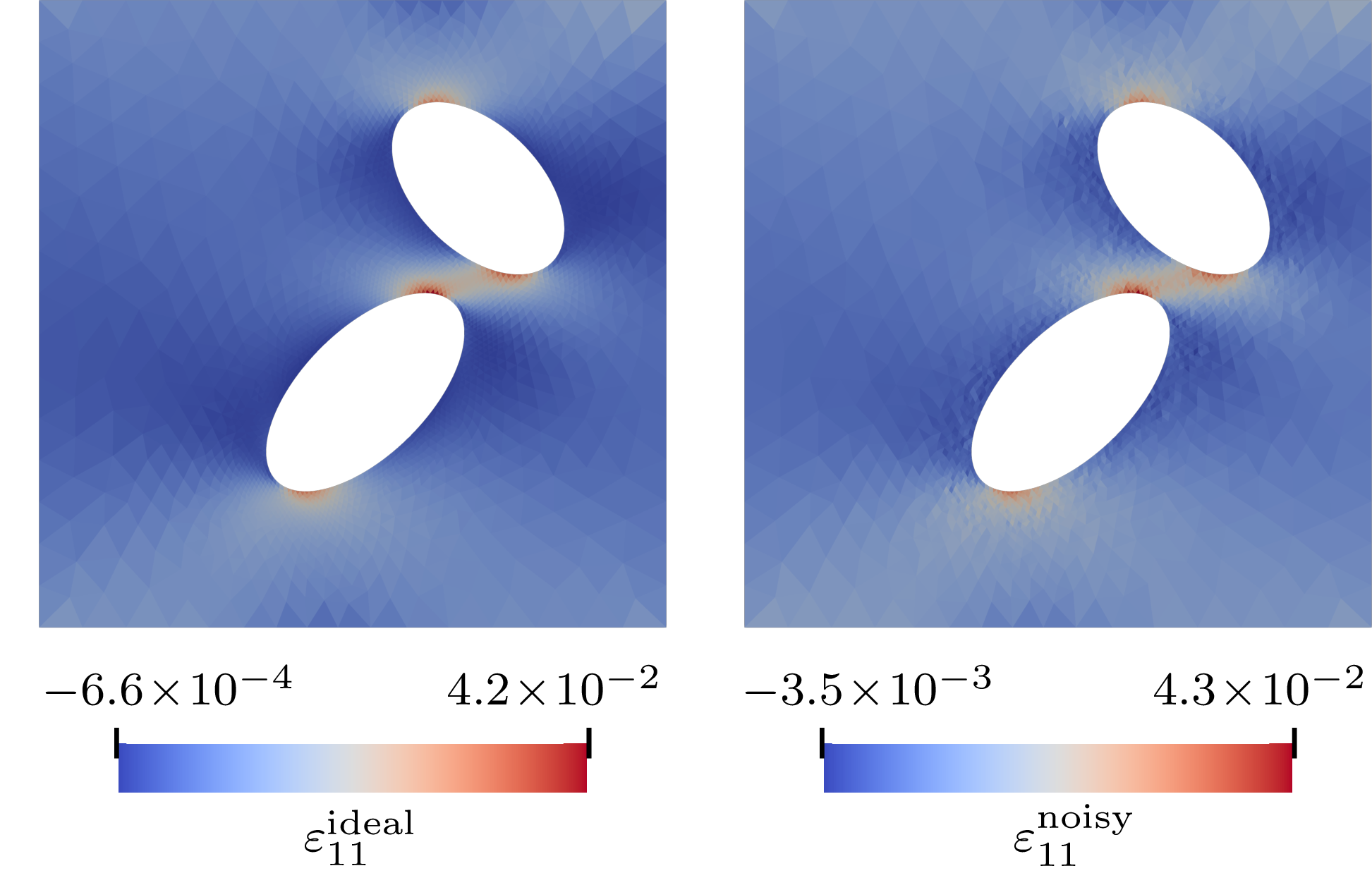}
  \caption{Ideal strain $\varepsilon_{11}^{\textrm{ideal}}$ and noisy strain $\varepsilon_{11}^{\textrm{noisy}}$ for time step $s=89$ calculated by the ideal and noisy \emph{in-plane displacements} of the uniaxial specimen training data.}
  \label{fig:ideal_and_noisy_eps11}
\end{figure}
\subsubsection{Results}
\label{sssec:results_noisy_data}
As in the case of ideal data, the PANN was trained using ten initializations for each of the three datasets: noisy uniaxial, noisy biaxial, and the combination of both. 

For the noisy data, training terminated early for some of these initializations, because the least-squares method used to compute the update in the Newton scheme in the forward pass encountered an ill-conditioned coefficient matrix.
In total, calibration completed successfully for five out of ten initializations in the case of the noisy uniaxial dataset, and for eight out of ten each in the case of the noisy biaxial and the combined noisy datasets. 
As only a limited number of calibrations can be executed, early termination could in principle also occur for the ideal data considered previously, even though it was not observed there.

Among the calibrations that completed successfully, the PANN with the smallest loss was again selected as the best-performing model for each of the three datasets. These trained PANNs were subsequently employed in a finite element simulation of the validation boundary value problem shown in Fig.~\ref{fig:synthetic_specimens}. The resulting stress component $\sigma_{11}^{\textrm{PANN,noisy}}$ is compared in Fig.~\ref{fig:validation_results_paraview} to the GT values $\sigma_{11}^{\textrm{GT}}$ obtained from the linear viscoelastic model.
Compared to the case of ideal data, the absolute stress difference is now noticeably larger for all three trained PANNs. The best agreement is obtained for the PANN trained on both noisy uniaxial and biaxial data, followed by the PANN trained solely on noisy biaxial data, and lastly by the PANN trained solely on noisy uniaxial data. We assume that this ordering is governed more by the specific initialization than by different training data. Furthermore, as expected, the largest stress differences occur near the boundaries, where spatial displacement gradients are high. Aside from this, the overall stress prediction remains reasonably accurate. The remaining stress components predicted by the three PANNs exhibit a similar level of error when compared to the reference values.
\section{Conclusion}
\label{sec:conclusion}

In this work, an unsupervised learning framework for calibrating a PANN for small-strain viscoelasticity from full-field data is proposed. The approach relies solely on quantities accessible in real experiments---\emph{global reaction forces} and \emph{surface displacements}---and formulates an optimization problem based on the EGM that is constrained under a \emph{plane stress assumption} and additionally by the evolution equations of internal variables. The underlying PANN is embedded in the GSM framework, combining \emph{invariant-based} representations of the free energy and the DDP, such that \emph{thermodynamic consistency} and \emph{material symmetry} hold by construction. To reduce the training cost, the gradient of the target loss is evaluated via the \emph{backward adjoint method}. The framework shows excellent agreement with synthetically generated data and reliably recovers the material response even in the presence of added artificial noise.

In summary, the presented framework demonstrates that a \emph{thermodynamically consistent} viscoelastic PANN can be calibrated using only data that is measurable in real experiments, thereby closing the gap between physics-augmented neural inelasticity modeling and full-field experimental characterization. Its key advantage lies in the generality of the applied framework, which can be transferred in a 
straightforward manner to different constitutive settings and thus enable a systematic, largely automated approach to material modeling.

Several extensions are planned for the future. Natural next steps include the calibration with real experimental data \citep{Abbasi2026a,Wiesheier2026}, the extension to finite strains \citep{kalina_physics-augmented_2026,Holthusen2026a}, and the transfer of the approach to further 
inelastic material classes such as elasto-plasticity \citep{Meyer2023a,Fuhg2023}. 

\section*{CRediT authorship contribution statement}
\textbf{Brain M. Riemer:} Conceptualization, Formal analysis, Investigation, Methodology, Software, Validation, Visualization, Writing -- original draft, Writing -- review \& editing.
\textbf{Markus Kästner:} Resources, Writing -- review \& editing, Funding acquisition.
\textbf{Karl A. Kalina:} Supervision, Conceptualization, Formal analysis, Investigation, Methodology, Software, Validation, Writing -- original draft, Writing -- review \& editing, Funding acquisition.

\section*{Usage of AI tools}
In preparing this work, the authors partially used Claude Opus 4.8 and Claude Sonnet 5, generative AI tools, to improve the readability and language of the manuscript. After using these tools, the authors reviewed and revised the content as necessary and take full responsibility for the content of the published article.

\section*{Acknowledgment}
The authors thank the German Research Foundation (DFG) for the support within the Research Training Group GRK 2868 D${}^3$--Project Number 493401063 and the project ANIVIS-NN with project number 571269372, grant KA 5978/2--1. A heartfelt thank also to Jörg Brummund for numerous thought-provoking comments.

\section*{Data availability}
Once the final version of the article is published, the code will be
made publicly available.

\appendix

\section{Notation}\label{app:notation}
Blackboard bold symbols are used for domains of numbers: the integers $\Setnum{Z}$ and the real numbers $\Setnum{R}$. Any domain of number with a condition as a subscript consists only of elements that satisfy that condition. For example, $\Setnum{Z}_{\geq0}:=\left\{z\in\Setnum{Z}\:|\:z \geq0\right\}$.

The space of real hyper-matrices of order $m\in\Setnum{Z}_{\geq1}$ is denoted by $\hms{n_1\times n_2 \times \dots \times n_i}$, $i\in\{1,2,\dots,m\},\:n_1,n_2,\dots,n_i\in\Setnum{Z}_{\geq1}$, while $\ts{n}{}$ is the space of $n$th order real tensors $n\in\Setnum{Z}_{\geq1}$. Fully symmetric or skew-symmetric $n$th order tensors are elements of $\tss{n}\subset\ts{n}{}$ and $\tsa{n}\subset\ts{n}{}$, respectively. Non-deviatoric real-valued 2nd order tensors are element of $\ts{2}{\tr\neq 0}\subset\ts{2}{}$ and positive definite real-valued 2nd order tensors element of $\ts{2}{\lambda>0}\subset\ts{2}{}$.
We denote real-valued 0th, 1st, 2nd and $4$th order tensors by $A$ $\ve{A}$, $\te{A}$ and $\mathbb{A}$, while vectors are represented by $\V{A}$ and matrices by $\M{A}$. 

The functional dependency on both spatial point $\ve{x}\in\ts{1}{}$ and time $t\in\Setnum{R}_{\geq0}$ of tensors and hyper-matrices is often omitted within this work. 
Nevertheless, evaluations at discrete times and iterations are marked by exclusive superscripts $s\in\{1,2,\dots,\nt\}$ and $i\in\{1,2,\dots \niter\}$ on the left of an symbol, where $\nt\in\Setnum{Z}_{\geq1}$ is the number of time steps and $\niter\in\Setnum{Z}_{\geq1}$ the maximum number of iterations. 
Furthermore, we use one more special superscript $e\in\{1,2,\dots,\nele\}$, $\nele\in\Setnum{Z}_{\geq1}$ the number of elements, to show the affiliation to a specific finite element.
For example, a 2nd order tensor evaluated at time step $s$ iteration $i$ for an element $e$ is denoted by $\tivar{\te{A}}^{e}$.

In index notation, if not stated differently, indices $a,b,c,d\in\{1,2\}$ and $k,l,\dots,r\in\{1,2,3\}$. Einstein summation convention applies to these lower case Latin letters. 
For illustration, $\tivar{\mathbb{A}}^{e}:=\tivar{A}^{e}_{klmn}\ve{e}_{k}\ve{e}_{l}\ve{e}_{m}\ve{e}_{n}$, where $\ve{e}_1, \ve{e}_2$ and $\ve{e}_3$ are the standard basis vectors of $\mathbb R^3$.
Greek indices as well as the indices $e$, $i$ and $s$ are excluded from the summation convention. 

The single and double contraction between tensors or hyper-matrices are denoted by $\cdot$ and $:$, e.g., $\te{A}\cdot\te{B}:=A_{kl}B_{lm}\ve{e}_{k}\ve{e}_{m}$ and $\te{A}:\te{B}:=A_{kl}B_{lk}$. 
In addition, $\te{A}^m$ denotes the $m$-fold single contractions between each $\te{A}$, e.g., $\te{A}^3 := \te{A}\cdot\te{A}\cdot\te{A}$. 
The trace and the symmetric part of 2nd order tensors or square matrices are denoted by $\tr(\bullet)$ and $\sym(\bullet)$, respectively, and the transposition by $(\bullet)^\top$. The Nabla operator is defined by $\nabla:=\diffp{}{x_k}\ve{e}_{k}$.
Furthermore, $\delta_{kl}$ denotes the Kronecker-Delta and $\mathbb{I}^{\textrm{sym}} := \frac{1}{2}(\delta_{km}\delta_{ln} + \delta_{kn}\delta_{lm})\ve{e}_k\ve{e}_l\ve{e}_m\ve{e}_n$, $\mathbb{I}^{\textrm{sph}} := \frac{1}{3}\delta_{kl}\delta_{mn}\ve{e}_k\ve{e}_l\ve{e}_m\ve{e}_n$ and $\mathbb{I}^{\textrm{dev}} := \mathbb{I}^{\textrm{sym}} - \mathbb{I}^{\textrm{sph}}$ the symmetric, sphere and deviatoric projector, respectively.
The special orthogonal group and the orthogonal group are represented by $\Group{O}(3):=\Group{O}(3,\Setnum{R})$ and $\Group{SO}(3):=\Group{SO}(3,\Setnum{R})$. The number of global nodes is denoted by $\nn\in\Setnum{Z}_{>0}$.
\section{Convex functional basis for a symmetric non-deviatoric 2nd order tensor and the orthogonal group}
\label{app:functional_basis_trace}
We show that a \emph{functional basis} for a symmetric non-deviatoric 2nd order tensor $\te{s}\in\tss{2} \cap \ts{2}{\tr\neq 0}\subset\tss{2}$ and group $\Group{O}(3)$ is constituted by the convex invariants 
\begin{subequations}
    \begin{align}
        \hspace{-7pt}
        I_1 &: \tss{2} \cap \ts{2}{\tr\neq 0} \to \Setnum{R}, \: \te{s}\mapsto I_1(\te{s})=\tr(\te{s}), \label{eq:s_invar1} \\
        \hspace{-7pt}
        I_2 &: \tss{2} \cap \ts{2}{\tr\neq 0} \to \Setnum{R}, \: \te{s}\mapsto I_2(\te{s})=\tr((\mathbb{I}^{\textrm{dev}}:\te{s})^2) \: \textrm{and} \label{eq:s_invar2} \\
        \hspace{-7pt}
        I_3 &: \tss{2} \cap \ts{2}{\tr\neq 0} \to \Setnum{R}, \: \te{s}\mapsto I_3(\te{s})=\tr(\te{s}^4) \point \label{eq:s_invar3}   
    \end{align}
    \label{eq:functional_basis}%
\end{subequations}
This claim requires the introduction of two more invariants
\begin{subequations}
    \begin{align}
        K &: \tss{2} \cap \ts{2}{\tr\neq 0} \to \Setnum{R}, \: \te{s}\mapsto K(\te{s})=\tr(\te{s}^2) \: \textrm{and} \label{eq:tr2_invariant} \\
        J &: \tss{2} \cap \ts{2}{\tr\neq 0} \to \Setnum{R}, \: \te{s}\mapsto J(\te{s})=\tr(\te{s}^3) \point \label{eq:tr3_invariant}
    \end{align}
    \label{eq:tr2_tr3_invariants}%
\end{subequations}
On the other hand, for the more common case of a symmetric 2nd order tensor $\te{r}\in\tss{2}$ and group $\Group{O}(3)$, the invariants in Eq.~\eqref{eq:functional_basis} do not form a \emph{functional basis}, which is why the domain is restricted for $I_1$, $I_2$ and $I_3$ in Eq.~\eqref{eq:functional_basis}. To note the difference in domain, we formally introduce the invariants
\begin{subequations}
    \begin{align}
        \Bar{I}_1 &: \tss{2} \to \Setnum{R}, \: \te{r}\mapsto \Bar{I}_1(\te{r})=\tr(\te{r}), \label{eq:r_invar1} \\
        \Bar{I}_2 &: \tss{2} \to \Setnum{R}, \: \te{r}\mapsto \Bar{I}_2(\te{r})=\tr((\mathbb{I}^{\textrm{dev}}:\te{r})^2), \label{eq:r_invar2} \\
        \Bar{I}_3 &: \tss{2} \to \Setnum{R}, \: \te{r}\mapsto \Bar{I}_3(\te{r})=\tr(\te{r}^4) \label{eq:r_invar3} \\
        \Bar{K} &: \tss{2} \to \Setnum{R}, \: \te{r}\mapsto \Bar{K}(\te{r})=\tr(\te{r}^2) \: \textrm{and} \label{eq:r_invar4} \\
        \Bar{J} &: \tss{2} \to \Setnum{R}, \: \te{r}\mapsto \Bar{J}(\te{r})=\tr(\te{r}^3) \point \label{eq:r_invar5}
    \end{align}
    \label{eq:r_invars}%
\end{subequations}
Please note that the invariants $\Bar{I}_1$, $\Bar{K}$ and $\Bar{J}$ constitute a well-known \emph{functional basis} for the case of $\te{r}\in\tss{2}$ and $\Group{O}(3)$, see \citet{smith_isotropic_1971, boehler_irreducible_1977, pennisi_irreducibility_1987}.
\begin{proposition}
    \label{prop:not_functional_basis}
    Let the invariants $\Bar{I}_1$, $\Bar{K}$ and $\Bar{J}$, defined in Eq.~\eqref{eq:r_invars}, form a functional basis for $\te{r}\in\tss{2}$ and group $\Group{O}(3)$. Then the invariants $\Bar{I}_1$, $\Bar{K}$ and $\Bar{I}_3$, defined in Eq.~\eqref{eq:r_invars}, do not form a functional basis for $\te{r}\in\tss{2}$ and group $\Group{O}(3)$.
\end{proposition}
\begin{proof}
    \label{proof:not_functional_basis}
    We use the argumentation by \citet{smith_fundamental_1970} to show that the invariants $\Bar{I}_1$, $\Bar{K}$ and $\Bar{I}_3$ are not separating the group orbits and hence do not form a functional basis~\citep{pipkin_material_1963}. For the diagonal tensors $\te{r}_1=\textrm{diag}(-2,1,1)$ and $\te{r}_2=\textrm{diag}(2,-1,-1)$, which belong to different group orbits, $\Bar{I}_1(\te{r}_1)=\Bar{I}_1(\te{r}_2)=0$, $\Bar{K}(\te{r}_1)=\Bar{K}(\te{r}_2)=6$, $\Bar{I}_3(\te{r}_1)=\Bar{I}_3(\te{r}_2)=18$, but $\Bar{J}(\te{r}_1)=-6$ and $\Bar{J}(\te{r}_2)=6$. Thus, $\Bar{I}_1$, $\Bar{K}$ and $\Bar{I}_3$ do not separate the group orbits.
\end{proof}
\begin{corollary}
    \label{corollary:not_functional_basis_deviatoric}
    The invariants $\Bar{I}_1$, $\Bar{I}_2$ and $\Bar{I}_3$, defined in Eq.~\eqref{eq:r_invars}, do not form a functional basis for $\te{r}\in\tss{2}$ and group $\Group{O}(3)$.
\end{corollary}
\begin{proof}
    \label{proof:not_functional_basis_deviatoric}
    This follows directly from the proof of Proposition~\ref{prop:not_functional_basis}, since the tensors $\te{r}_1$ and $\te{r}_2$ are deviatoric. Hence, $\Bar{K}(\te{r}_1)=\Bar{I}_2(\te{r}_1)$ and $\Bar{K}(\te{r}_2)=\Bar{I}_2(\te{r}_2)$. 
\end{proof}
\begin{lemma}
    \label{lemma:functional_bases_for_non_zero_trace}
    The functional basis constituted by $\Bar{I}_1$, $\Bar{K}$ and $\Bar{J}$, defined in Eq.~\eqref{eq:r_invars}, separate the group orbits for $\Group{O}(3)$ and $\te{r}\in\tss{2}$. Thus, the invariants $I_1$, $K$ and $J$, defined in Eq.~\eqref{eq:s_invar1} and Eq.~\eqref{eq:tr2_tr3_invariants}, separate the group orbits for $\Group{O}(3)$ and $\te{s}\in\tss{2} \cap \ts{2}{\tr\neq 0}\subset\tss{2}$.
\end{lemma}
\begin{proof}
    The fact that $\Bar{I}_1$, $\Bar{K}$ and $\Bar{J}$ separate the group orbits for $\Group{O}(3)$ and $\te{r}\in\tss{2}$ is well-known~\citep{smith_isotropic_1971, pennisi_irreducibility_1987, boehler_irreducible_1977}. Hence, they separate all group orbits where $\Bar{I}_1(\te{r})\neq0$. The elements of theses group orbits with $\Bar{I}_1(\te{r})\neq0$ build the proper subset $\tss{2} \cap \ts{2}{\tr\neq 0}\subset\tss{2}$. Thus, $I_1$, $K$ and $J$, which are defined analogous to $\Bar{I}_1$, $\Bar{K}$ and $\Bar{J}$, separate the group orbits for $\Group{O}(3)$ and $\te{s}\in\tss{2} \cap \ts{2}{\tr\neq 0}\subset\tss{2}$.
\end{proof}
\begin{proposition}
    \label{prop:functional_basis_trace_non_zero}
    The invariants $I_1$, $I_2$ and $I_3$, defined in Eq.~\eqref{eq:functional_basis}, constitute a functional basis for $\te{s}\in\tss{2} \cap \ts{2}{\tr\neq 0}\subset\tss{2}$ and group $\Group{O}(3)$.
\end{proposition}
\begin{proof}
    \label{proof:functional_basis_trace_non_zero}
    By inspection we find the syzygies
    \begin{subequations}
        \begin{align}
            0 &= 3K - I_1^2 - 3I_2 \label{eq:syz1_for_proof} \\
            0 &= 24JI_1 - 2I_1^4 - 12I_1^2I_2 + 9I_2^2 - 18I_3 \label{eq:syz2_for_proof} \point
        \end{align}
        \label{eq:syzygies_for_proof}%
    \end{subequations}
    We note that $\{I_1,K,J\}$ represents a functional basis for $\Group{O}(3)$ and $\te{s}\in\tss{2} \cap \ts{2}{\tr\neq 0}\subset\tss{2}$, see Lemma~\ref{lemma:functional_bases_for_non_zero_trace}. 
    From Eq.~\eqref{eq:syzygies_for_proof} follows that $K$ is expressible for arbitrary values of $I_1$ and $I_2$, while $J$ is only expressible for $I_1(\te{s})\neq0$.
    Since $\nexists\:\te{s}\in\tss{2} \cap \ts{2}{\tr\neq 0}:I_1(\te{s)}=0$, there exists a functional dependency between every element of $\{I_1,I_2,I_3\}$ and $\{I_1,K,J\}$. Hence, $\{I_1,I_2,I_3\}$ is a functional basis.
\end{proof}
\begin{corollary}
    \label{corollary:functional_basis_positive_definit}
    The invariants $I_1$, $I_2$ and $I_3$, defined in Eq.~\eqref{eq:functional_basis}, form a functional basis for $\Group{O}(3)$ and a positive definite symmetric 2nd order tensor $\te{s}\in\ts{2}{\lambda>0}\subset\tss{2}$.
\end{corollary}
\begin{proof}
    Since $I_1(\te{s})\neq0 \:\forall\:\te{s}\in\ts{2}{\lambda>0}$, this follows directly from the proof of Proposition~\ref{prop:functional_basis_trace_non_zero}.
\end{proof}
\section{Backward adjoint method}
\label{app:backward_adjoint}
To enable efficient computation of the gradient $\diff{L_{\textrm{EGM}}}{\V{\theta}}$, or $\diff{L_{\textrm{EGM}}}{\theta_\gamma}$ in index notation, with respect to the trainable parameters $\theta_\gamma$, $\gamma\in\{1,2,\dots,n_\theta\}$, we use the \emph{backward adjoint method}. This approach avoids back-propagating through all Newton iteration steps, cf.~\citep{Seidl2022,Kumar2025,Akerson2025}. 
Instead it relies solely on information from the last Newton iteration of each time step $s\in\{1,\dots,\nt\}$, namely the residuals $\tvar{R}^e_\omega(\tsvar{y}^e_\rho(\theta_\gamma), {}^{s-1}\!y^e_\rho(\theta_\gamma), \theta_\gamma)$, where $\omega,\rho\in\{ 1,2,\dots n_\textrm{u} \}$, $n_\textrm{u}\in\Setnum{Z}_{>0}$ the number of unknowns $\tsvar{y}^e_\rho$ that are computed iteratively in every time step. 
In this work, $n_{\textrm{u}}=5$ and $\tsvar{y}^e_\rho$ stores $\tvar{\varepsilon}_{33}^e$, $\tvar{q}_{11}^e$, $\tvar{q}_{22}^e$, $\tvar{q}_{33}^e$, and $\tvar{q}_{12}^e$.

Furthermore, the \emph{backward adjoint method} relies on the implicit function theorem together with the fact that $\diff{\tvar{R}^e_\omega}{\theta_\gamma}\approx0$.\footnote{For $\diff{\tvar{R}^e_\omega}{\theta_\gamma}\approx0$ to hold, the Newton iteration must have converged. This cannot be guaranteed for arbitrary $\theta_\gamma$ during training, however; to allow the calibration process to continue, we nevertheless treat the result of the last iteration as converged.}
An alternative is the \emph{forward sensitivity method}, which likewise avoids back-propagating through all Newton iterations. The advantage of the \emph{backward adjoint method} over the \emph{forward sensitivity method}, however, is that it introduces a new intermediate quantity, $\tvar{\lambda}^e_\omega\in\Setnum{R}$, which we refer to as adjoint parameters. As a result, the \emph{backward adjoint method} somewhat scales with the size of the vector of unknowns $\tvar{y}^e_\rho$ computed through the Newton iteration, rather than with the size of $\theta_\gamma$, as is the case for the \emph{forward sensitivity method}. The \emph{backward adjoint method} is therefore less demanding when the number of trainable parameters $\theta_\gamma$ is large, as is typically the case for NNs.

In order to calculate the gradient $\diff{L_{\textrm{EGM}}}{\theta_\gamma}$, first the linear systems of equations
\vspace{-4pt}
\begin{align}
    s=\nt &: \sum_{\rho=1}^{n_{\textrm{u}}} 
    \diffp{\:{}^{\nt}\!R^e_\omega} {\:{}^{\nt}y^e_\rho} {}^{\nt}\!\lambda^e_\rho 
    = - \diffp{L_{\textrm{EGM}}} {\:{}^{\nt}y^e_\omega} \label{eq:adjoint_method1} \\
    1 < s < \nt &: \sum_{\rho=1}^{n_{\textrm{u}}} \diffp{\:\tvar{R}^e_\omega} {\:\tsvar{y}^e_\rho} \tvar{\lambda}^e_\rho 
    = - \left[ \sum_{\rho=1}^{n_{\textrm{u}}}{}^{s+1} \!\lambda^e_\rho \diffp{\:{}^{s+1}\!R^e_\rho} {\:\tsvar{y}^e_\omega}  
    \right. \nonumber \\
    &\phantom{\hspace{3.0cm}}
    + \left.\diffp{L_{\textrm{EGM}}} {\:\tsvar{y}^e_\omega} \right] \nonumber
\end{align}
must be solved for the adjoint parameters $\tvar{\lambda}^e_\rho$. Because of robustness, we solve the systems of equations in Eq.~\eqref{eq:adjoint_method1} using a least squares method, cf. footnote~\ref{foot:least_squares}. 
Subsequently the loss gradient w.r.t. $\theta_\gamma$ is computed by
\begin{equation}
    \hspace{-0.5cm}
    \diff{L_{\textrm{EGM}}}{\theta_\gamma} =
    \sum_{s=2}^{\nt} \sum_{e=1}^{\nele} 
    \left[
    \diffp{L_{\textrm{EGM}}}{\:\tvar{\sigma}^e_{ab}} \diffp{\:\tvar{\sigma}^e_{ab}}{\theta_\gamma} + 
    \sum_{\rho=1}^{n_{\textrm{u}}}
    \tvar{\lambda}^e_\rho  \diffp{\:\tvar{R}^e_\rho}{\theta_\gamma}
    \right] \comma
    \label{eq:adjoint_method2}
\end{equation}
where $\diffp{L_{\textrm{EGM}}}{\:\tvar{\sigma}^e_{ab}}$ denotes the partial derivative of the loss w.r.t. the \emph{two-dimensional stresses} $\tvar{\sigma}^e_{ab}$ and $\diffp{L_{\textrm{EGM}}}{\:\tsvar{y}^e_\omega} = \diffp{L_{\textrm{EGM}}}{\:\tvar{\sigma}^e_{ab}} \diffp{\:\tvar{\sigma}^e_{ab}}{\:\tsvar{y}^e_\omega}$. 
The dependence on the various \emph{assembled} and \emph{local nodal forces} is thus already accounted for in these two gradients.
In summary, one \emph{forward and backward pass} consists of the following steps:
\begin{enumerate}
    \item Compute $\tsvar{y}^e_\omega$ and $L_{\textrm{EGM}}$.
    \item Calculate $\diffp{L_{\textrm{EGM}}}{\:\tvar{\sigma}^e_{ab}}$, 
    $\diffp{L_{\textrm{EGM}}}{\:\tsvar{y}^e_\omega}$,  
    $\diffp{\:\tvar{R}^e_\omega} {\:\tsvar{y}^e_\rho}$, 
    $\diffp{\:{}^{s+1}\!R^e_\rho} {\:\tsvar{y}^e_\omega}$ and 
    $\diffp{\:\tvar{\sigma}^e_{ab}}{\theta_\gamma}$ using automatic differentiation, for fixed $\tsvar{y}^e_\omega$.
    \item Determine $\tvar{\lambda}^e_\rho$ from the linear systems of equations in Eq.~\eqref{eq:adjoint_method1} sequentially, starting from the last time step $\nt$.
    \item Compute the loss gradient in Eq.~\eqref{eq:adjoint_method2}.
\end{enumerate} 
%
\begin{small}
    \bibliographystyle{apalike-ejor}
    \bibliography{ff_small_strain.bib}

@article{jailin_noise-bias_2026,
	title = {Noise-bias compensation for the unsupervised learning of constitutive laws},
	volume = {354},
	issn = {1631-0721, 1873-7234},
	url = {https://comptes-rendus.academie-sciences.fr/mecanique/articles/10.5802/crmeca.342/},
	doi = {10.5802/crmeca.342},
	language = {en},
	number = {G1},
	urldate = {2026-02-03},
	journal = {Comptes Rendus. Mécanique},
	author = {Jailin, Clément and Roux, Stéphane and Benady, Antoine and Baranger, Emmanuel},
	month = jan,
	year = {2026},
	pages = {1--24},
}

@article{pipkin_material_1963,
	title = {Material symmetry restrictions on non-polynomial constitutive equations},
	volume = {12},
	issn = {1432-0673},
	url = {https://doi.org/10.1007/BF00281238},
	doi = {10.1007/BF00281238},
	language = {en},
	number = {1},
	urldate = {2025-11-08},
	journal = {Archive for Rational Mechanics and Analysis},
	author = {Pipkin, A. C. and Wineman, A. S.},
	month = jan,
	year = {1963},
	pages = {420--426},
}

@article{linden_dual-stage_2025,
	title = {A dual-stage constitutive modeling framework based on finite strain data-driven identification and physics-augmented neural networks},
	volume = {447},
	issn = {0045-7825},
	url = {https://www.sciencedirect.com/science/article/pii/S0045782525005614},
	doi = {10.1016/j.cma.2025.118289},
	urldate = {2025-10-02},
	journal = {Computer Methods in Applied Mechanics and Engineering},
	author = {Linden, Lennart and Kalina, Karl A. and Brummund, Jörg and Riemer, Brain and Kästner, Markus},
	month = dec,
	year = {2025},
	pages = {118289},
}

@book{silhavy_mechanics_1997,
	address = {Berlin, Heidelberg},
	title = {The {Mechanics} and {Thermodynamics} of {Continuous} {Media}},
	copyright = {http://www.springer.com/tdm},
	isbn = {978-3-642-08204-7 978-3-662-03389-0},
	url = {http://link.springer.com/10.1007/978-3-662-03389-0},
	doi = {10.1007/978-3-662-03389-0},
	urldate = {2025-09-28},
	publisher = {Springer},
	author = {Šilhavý, Miroslav},
	year = {1997},
}

@book{pierron_virtual_2012,
	address = {New York, NY},
	title = {The {Virtual} {Fields} {Method}: {Extracting} {Constitutive} {Mechanical} {Parameters} from {Full}-field {Deformation} {Measurements}},
	copyright = {https://www.springernature.com/gp/researchers/text-and-data-mining},
	isbn = {978-1-4614-1823-8 978-1-4614-1824-5},
	shorttitle = {The {Virtual} {Fields} {Method}},
	url = {https://link.springer.com/10.1007/978-1-4614-1824-5},
	doi = {10.1007/978-1-4614-1824-5},
	language = {en},
	urldate = {2025-09-15},
	publisher = {Springer},
	author = {Pierron, Fabrice and Grédiac, Michel},
	year = {2012},
}

@misc{boyd_convex_2004,
	title = {Convex {Optimization}},
	url = {https://www.cambridge.org/highereducation/books/convex-optimization/17D2FAA54F641A2F62C7CCD01DFA97C4},
	doi = {10.1017/CBO9780511804441},
	language = {en},
	urldate = {2025-09-11},
	journal = {Cambridge Aspire website},
	publisher = {Cambridge University Press},
	author = {Boyd, Stephen and Vandenberghe, Lieven},
	month = mar,
	year = {2004},
	note = {ISBN: 9780511804441},
}

@misc{amos_input_2017,
	title = {Input {Convex} {Neural} {Networks}},
	url = {http://arxiv.org/abs/1609.07152},
	doi = {10.48550/arXiv.1609.07152},
	urldate = {2025-09-03},
	publisher = {arXiv},
	author = {Amos, Brandon and Xu, Lei and Kolter, J. Zico},
	month = jun,
	year = {2017},
	note = {arXiv:1609.07152 [cs]},
}

@misc{dammas_when_2025,
	title = {When invariants matter: {The} role of {I1} and {I2} in neural network models of incompressible hyperelasticity},
	shorttitle = {When invariants matter},
	url = {http://arxiv.org/abs/2503.20598},
	doi = {10.48550/arXiv.2503.20598},
	urldate = {2025-09-02},
	publisher = {arXiv},
	author = {Dammaß, Franz and Kalina, Karl A. and Kästner, Markus},
	month = mar,
	year = {2025},
	note = {arXiv:2503.20598 [cond-mat]},
}

@book{ottosen_mechanics_2005,
	address = {Amsterdam London},
	title = {The mechanics of constitutive modeling},
	isbn = {978-0-08-044606-6},
	language = {eng},
	publisher = {Elsevier},
	author = {Ottosen, Niels Saabye and Ristinmaa, Matti},
	year = {2005},
}

@book{neumann_vorlesungen_1885,
	title = {Vorlesungen über die theorie der elasticität der festen körper und des lichtäthers, gehalten an der {Universität} {Königsberg}},
	url = {http://archive.org/details/vorlesungenberd01neumgoog},
	language = {ger},
	urldate = {2025-07-17},
	publisher = {Leipzig, B. G. Teubner},
	author = {Neumann, Franz Ernst and Meyer, Oskar Emil},
	collaborator = {{unknown library}},
	year = {1885},
}

@article{smith_isotropic_1971,
	title = {On isotropic functions of symmetric tensors, skew-symmetric tensors and vectors},
	volume = {9},
	issn = {0020-7225},
	url = {https://www.sciencedirect.com/science/article/pii/0020722571900231},
	doi = {10.1016/0020-7225(71)90023-1},
	number = {10},
	urldate = {2025-06-25},
	journal = {International Journal of Engineering Science},
	author = {Smith, G. F.},
	month = oct,
	year = {1971},
	pages = {899--916},
}

@article{smith_fundamental_1970,
	title = {On a fundamental error in two papers of {C}.-{C}. {Wang} “{On} representations for isotropic functions, parts {I} and {II}”},
	volume = {36},
	issn = {1432-0673},
	url = {https://doi.org/10.1007/BF00272240},
	doi = {10.1007/BF00272240},
	language = {en},
	number = {3},
	urldate = {2025-06-25},
	journal = {Archive for Rational Mechanics and Analysis},
	author = {Smith, G. F.},
	month = jan,
	year = {1970},
	pages = {161--165},
}

@article{rosenkranz_comparative_2023,
	title = {A comparative study on different neural network architectures to model inelasticity},
	volume = {124},
	issn = {1097-0207},
	url = {https://onlinelibrary.wiley.com/doi/abs/10.1002/nme.7319},
	doi = {10.1002/nme.7319},
	language = {en},
	number = {21},
	urldate = {2025-05-15},
	journal = {International Journal for Numerical Methods in Engineering},
	author = {Rosenkranz, Max and Kalina, Karl A. and Brummund, Jörg and Kästner, Markus},
	year = {2023},
	note = {\_eprint: https://onlinelibrary.wiley.com/doi/pdf/10.1002/nme.7319},
	pages = {4802--4840},
}

@article{rosenkranz_viscoelasticty_2024,
	title = {Viscoelasticty with physics-augmented neural networks: model formulation and training methods without prescribed internal variables},
	volume = {74},
	issn = {1432-0924},
	shorttitle = {Viscoelasticty with physics-augmented neural networks},
	url = {https://doi.org/10.1007/s00466-024-02477-1},
	doi = {10.1007/s00466-024-02477-1},
	language = {en},
	number = {6},
	urldate = {2025-05-15},
	journal = {Computational Mechanics},
	author = {Rosenkranz, Max and Kalina, Karl A. and Brummund, Jörg and Sun, WaiChing and Kästner, Markus},
	month = dec,
	year = {2024},
	pages = {1279--1301},
}

@article{linden_neural_2023,
	title = {Neural networks meet hyperelasticity: {A} guide to enforcing physics},
	volume = {179},
	issn = {0022-5096},
	shorttitle = {Neural networks meet hyperelasticity},
	url = {https://www.sciencedirect.com/science/article/pii/S0022509623001679},
	doi = {10.1016/j.jmps.2023.105363},
	urldate = {2025-03-26},
	journal = {Journal of the Mechanics and Physics of Solids},
	author = {Linden, Lennart and Klein, Dominik K. and Kalina, Karl A. and Brummund, Jörg and Weeger, Oliver and Kästner, Markus},
	month = oct,
	year = {2023},
	pages = {105363},
}

@book{altenbach_kontinuumsmechanik_2018,
	address = {Berlin, Heidelberg},
	title = {Kontinuumsmechanik: {Einführung} in die materialunabhängigen und materialabhängigen {Gleichungen}},
	copyright = {http://www.springer.com/tdm},
	isbn = {978-3-662-57503-1 978-3-662-57504-8},
	shorttitle = {Kontinuumsmechanik},
	url = {http://link.springer.com/10.1007/978-3-662-57504-8},
	doi = {10.1007/978-3-662-57504-8},
	language = {de},
	urldate = {2025-03-20},
	publisher = {Springer},
	author = {Altenbach, Holm},
	year = {2018},
}

@book{holzapfel_nonlinear_2000,
	address = {Chichester ; New York},
	title = {Nonlinear solid mechanics: a continuum approach for engineering},
	isbn = {978-0-471-82304-9 978-0-471-82319-3},
	shorttitle = {Nonlinear solid mechanics},
	publisher = {Wiley},
	author = {Holzapfel, Gerhard A.},
	year = {2000},
}

@article{zheng_tensors_1993,
	title = {Tensors which characterize anisotropies},
	volume = {31},
	issn = {0020-7225},
	url = {https://www.sciencedirect.com/science/article/pii/002072259390118E},
	doi = {10.1016/0020-7225(93)90118-E},
	number = {5},
	urldate = {2025-01-04},
	journal = {International Journal of Engineering Science},
	author = {Zheng, Q. -S. and Spencer, A. J. M.},
	month = may,
	year = {1993},
	pages = {679--693},
}

@article{pennisi_irreducibility_1987,
	title = {On the irreducibility of professor {G}.{F}. {Smith}'s representations for isotropic functions},
	volume = {25},
	issn = {0020-7225},
	url = {https://www.sciencedirect.com/science/article/pii/0020722587900978},
	doi = {10.1016/0020-7225(87)90097-8},
	number = {8},
	urldate = {2025-01-04},
	journal = {International Journal of Engineering Science},
	author = {Pennisi, S. and Trovato, M.},
	month = jan,
	year = {1987},
	pages = {1059--1065},
}

@article{lokhin_nonlinear_1963,
	title = {Nonlinear tensor functions of several tensor arguments},
	volume = {27},
	issn = {0021-8928},
	url = {https://www.sciencedirect.com/science/article/pii/0021892863901497},
	doi = {10.1016/0021-8928(63)90149-7},
	number = {3},
	urldate = {2025-01-04},
	journal = {Journal of Applied Mathematics and Mechanics},
	author = {Lokhin, V. V. and Sedov, L. I.},
	month = jan,
	year = {1963},
	pages = {597--629},
}

@phdthesis{apel_approaches_2004,
	title = {Approaches to the {Description} of {Anisotropic} {Material}  {Behaviour} at {Finite} {Elastic} and {Plastic} {Deformations}  —{Theory} and {Numerics} —},
	school = {Stuttgart},
	author = {Apel, Nikolas},
	year = {2004},
}

@article{xiao_isotropic_1996,
	title = {On {Isotropic} {Extension} of {Anisotropic} {Tensor} {Functions}},
	volume = {76},
	copyright = {Copyright © 1996 WILEY-VCH Verlag GmbH \& Co. KGaA, Weinheim},
	issn = {1521-4001},
	url = {https://onlinelibrary.wiley.com/doi/abs/10.1002/zamm.19960760403},
	doi = {10.1002/zamm.19960760403},
	language = {en},
	number = {4},
	urldate = {2025-01-04},
	journal = {ZAMM - Journal of Applied Mathematics and Mechanics / Zeitschrift für Angewandte Mathematik und Mechanik},
	author = {Xiao, H.},
	year = {1996},
	note = {\_eprint: https://onlinelibrary.wiley.com/doi/pdf/10.1002/zamm.19960760403},
	pages = {205--214},
}

@article{kalina_physics-augmented_2026,
	title = {A physics-augmented neural network framework for finite strain incompressible viscoelasticity},
	volume = {455},
	issn = {0045-7825},
	url = {https://www.sciencedirect.com/science/article/pii/S0045782526001659},
	doi = {10.1016/j.cma.2026.118892},
	urldate = {2026-07-02},
	journal = {Computer Methods in Applied Mechanics and Engineering},
	author = {Kalina, Karl A. and Brummund, Jörg and Kästner, Markus},
	month = jun,
	year = {2026},
	pages = {118892},
}

@article{riemer_construction_2026,
	title = {Construction of minimal integrity bases for anisotropic hyperelasticity via structural tensors},
	volume = {216},
	issn = {0022-5096},
	url = {https://www.sciencedirect.com/science/article/pii/S0022509626002632},
	doi = {10.1016/j.jmps.2026.106763},
	urldate = {2026-07-21},
	journal = {Journal of the Mechanics and Physics of Solids},
	author = {Riemer, Brain M. and Brummund, Jörg and Kalina, Karl A. and Milor, Abel H. G. and Dammaß, Franz and Kästner, Markus},
	month = oct,
	year = {2026},
	pages = {106763},
}

@book{Haupt2000,
  title = {Continuum {{Mechanics}} and {{Theory}} of {{Materials}}},
  author = {Haupt, Peter},
  year = {2000},
  publisher = {Springer Berlin Heidelberg},
  address = {Berlin, Heidelberg},
  urldate = {2020-01-22},
  isbn = {978-3-662-04109-3},
  langid = {english}
}

@article{Dornheim2023,
  title = {Neural {{Networks}} for {{Constitutive Modeling}}: {{From Universal Function Approximators}} to {{Advanced Models}} and the {{Integration}} of {{Physics}}},
  shorttitle = {Neural {{Networks}} for {{Constitutive Modeling}}},
  author = {Dornheim, Johannes and Morand, Lukas and Nallani, Hemanth Janarthanam and Helm, Dirk},
  year = {2023},
  month = oct,
  journal = {Archives of Computational Methods in Engineering},
  issn = {1886-1784},
  doi = {10.1007/s11831-023-10009-y},
  urldate = {2023-12-13},
  langid = {english}
}

@article{Fuhg2024b,
  title = {A {{Review}} on {{Data-Driven Constitutive Laws}} for {{Solids}}},
  author = {Fuhg, Jan N. and Anantha Padmanabha, Govinda and Bouklas, Nikolaos and Bahmani, Bahador and Sun, WaiChing and Vlassis, Nikolaos N. and Flaschel, Moritz and Carrara, Pietro and De Lorenzis, Laura},
  year = {2024},
  month = nov,
  journal = {Archives of Computational Methods in Engineering},
  issn = {1886-1784},
  doi = {10.1007/s11831-024-10196-2},
  urldate = {2024-11-29},
  langid = {english}
}

@article{Ghaboussi1991,
  title = {Knowledge-{{Based Modeling}} of {{Material Behavior}} with {{Neural Networks}}},
  author = {Ghaboussi, J. and Garrett, J. H. and Wu, X.},
  year = {1991},
  journal = {Journal of Engineering Mechanics},
  volume = {117},
  number = {1},
  pages = {132--153},
  issn = {0733-9399, 1943-7889},
  doi = {10.1061/(ASCE)0733-9399(1991)117:1(132)},
  urldate = {2020-09-30},
  langid = {english}
}

@article{Wiesheier2024,
  title = {Versatile Data-Adaptive Hyperelastic Energy Functions for Soft Materials},
  author = {Wiesheier, Simon and {Moreno-Mateos}, Miguel Angel and Steinmann, Paul},
  year = {2024},
  month = oct,
  journal = {Computer Methods in Applied Mechanics and Engineering},
  volume = {430},
  pages = {117208},
  issn = {0045-7825},
  doi = {10.1016/j.cma.2024.117208},
  urldate = {2024-09-02}
}

@article{Wiesheier2026,
  title = {Data-Adaptive Spline-Based Viscoelasticity for Soft Solids},
  author = {Wiesheier, Simon and {Moreno-Mateos}, Miguel Angel and Steinmann, Paul},
  year = {2026},
  month = apr,
  journal = {Computer Methods in Applied Mechanics and Engineering},
  volume = {451},
  pages = {118705},
  publisher = {North-Holland},
  issn = {0045-7825},
  doi = {10.1016/j.cma.2025.118705},
  urldate = {2026-01-21},
  langid = {american}
}

@article{Flaschel2021,
  title = {Unsupervised Discovery of Interpretable Hyperelastic Constitutive Laws},
  author = {Flaschel, Moritz and Kumar, Siddhant and De Lorenzis, Laura},
  year = {2021},
  month = aug,
  journal = {Computer Methods in Applied Mechanics and Engineering},
  volume = {381},
  pages = {113852},
  issn = {00457825},
  doi = {10.1016/j.cma.2021.113852},
  urldate = {2021-10-08},
  langid = {english}
}

@article{Meyer2023a,
  title = {Thermodynamically Consistent Neural Network Plasticity Modeling and Discovery of Evolution Laws},
  author = {Meyer, Knut Andreas and Ekre, Fredrik},
  year = {2023},
  month = nov,
  journal = {Journal of the Mechanics and Physics of Solids},
  volume = {180},
  pages = {105416},
  issn = {0022-5096},
  doi = {10.1016/j.jmps.2023.105416},
  urldate = {2023-12-13}
}

@article{Bock2019,
  title = {A {{Review}} of the {{Application}} of {{Machine Learning}} and {{Data Mining Approaches}} in {{Continuum Materials Mechanics}}},
  author = {Bock, Frederic E. and Aydin, Roland C. and Cyron, Christian J. and Huber, Norbert and Kalidindi, Surya R. and Klusemann, Benjamin},
  year = {2019},
  month = may,
  journal = {Frontiers in Materials},
  volume = {6},
  pages = {110},
  issn = {2296-8016},
  doi = {10.3389/fmats.2019.00110},
  urldate = {2020-09-29},
  langid = {english}
}

@article{Raissi2019,
  title = {Physics-Informed Neural Networks: {{A}} Deep Learning Framework for Solving Forward and Inverse Problems Involving Nonlinear Partial Differential Equations},
  shorttitle = {Physics-Informed Neural Networks},
  author = {Raissi, M. and Perdikaris, P. and Karniadakis, G.E.},
  year = {2019},
  journal = {Journal of Computational Physics},
  volume = {378},
  pages = {686--707},
  issn = {00219991},
  doi = {10.1016/j.jcp.2018.10.045},
  urldate = {2020-10-02},
  langid = {english}
}

@article{Asad2023,
  title = {A Mechanics-Informed Neural Network Framework for Data-Driven Nonlinear Viscoelasticity},
  author = {As'ad, Faisal and Farhat, Charbel},
  year = {2023},
  journal = {AIAA SCITECH 2023 forum},
  eprint = {https://arc.aiaa.org/doi/pdf/10.2514/6.2023-0949},
  doi = {DOI: 10.2514/6.2023-0949}
}

@article{Klein2024,
  title = {Nonlinear Electro-Elastic Finite Element Analysis with Neural Network Constitutive Models},
  author = {Klein, Dominik K. and Ortigosa, Rogelio and {Mart{\'i}nez-Frutos}, Jes{\'u}s and Weeger, Oliver},
  year = {2024},
  month = may,
  journal = {Computer Methods in Applied Mechanics and Engineering},
  volume = {425},
  pages = {116910},
  issn = {0045-7825},
  doi = {10.1016/j.cma.2024.116910},
  urldate = {2024-06-17}
}

@article{Aldakheel2025,
  title = {Physics-Based Machine Learning for Computational Fracture Mechanics},
  author = {Aldakheel, Fadi and Elsayed, Elsayed S. and Heider, Yousef and Weeger, Oliver},
  year = {2025},
  month = apr,
  journal = {Machine Learning for Computational Science and Engineering},
  volume = {1},
  number = {1},
  pages = {18},
  issn = {3005-1436},
  doi = {10.1007/s44379-025-00019-x},
  urldate = {2025-08-04},
  langid = {english}
}

@article{Masi2021,
  title = {Thermodynamics-Based {{Artificial Neural Networks}} for Constitutive Modeling},
  author = {Masi, Filippo and Stefanou, Ioannis and Vannucci, Paolo and {Maffi-Berthier}, Victor},
  year = {2021},
  journal = {Journal of the Mechanics and Physics of Solids},
  volume = {147},
  pages = {104277},
  issn = {0022-5096},
  doi = {10.1016/j.jmps.2020.104277},
  urldate = {2021-07-08},
  langid = {english}
}

@article{Linka2021,
  title = {Constitutive Artificial Neural Networks: {{A}} Fast and General Approach to Predictive Data-Driven Constitutive Modeling by Deep Learning},
  shorttitle = {Constitutive Artificial Neural Networks},
  author = {Linka, Kevin and Hillg{\"a}rtner, Markus and Abdolazizi, Kian P. and Aydin, Roland C. and Itskov, Mikhail and Cyron, Christian J.},
  year = {2021},
  journal = {Journal of Computational Physics},
  volume = {429},
  pages = {110010},
  issn = {00219991},
  doi = {10.1016/j.jcp.2020.110010},
  urldate = {2021-03-25},
  langid = {english}
}

@article{Kalina2023,
  title = {{{FE}}{{{\textsuperscript{ANN}}}}: An Efficient Data-Driven Multiscale Approach Based on Physics-Constrained Neural Networks and Automated Data Mining},
  shorttitle = {{{FE}}{{{\textsuperscript{ANN}}}}},
  author = {Kalina, Karl A. and Linden, Lennart and Brummund, J{\"o}rg and K{\"a}stner, Markus},
  year = {2023},
  journal = {Computational Mechanics},
  volume = {71},
  pages = {827},
  issn = {1432-0924},
  doi = {10.1007/s00466-022-02260-0},
  urldate = {2023-02-08},
  langid = {english}
}

@article{Kalina2022a,
  title = {Automated Constitutive Modeling of Isotropic Hyperelasticity Based on Artificial Neural Networks},
  author = {Kalina, Karl A. and Linden, Lennart and Brummund, J{\"o}rg and Metsch, Philipp and K{\"a}stner, Markus},
  year = {2022},
  journal = {Computational Mechanics},
  volume = {69},
  number = {1},
  pages = {213--232},
  issn = {1432-0924},
  doi = {10.1007/s00466-021-02090-6},
  urldate = {2023-07-18},
  langid = {english}
}

@article{Fuhg2023,
  title = {Modular Machine Learning-Based Elastoplasticity: {{Generalization}} in the Context of Limited Data},
  shorttitle = {Modular Machine Learning-Based Elastoplasticity},
  author = {Fuhg, Jan Niklas and Hamel, Craig M. and Johnson, Kyle and Jones, Reese and Bouklas, Nikolaos},
  year = {2023},
  journal = {Computer Methods in Applied Mechanics and Engineering},
  volume = {407},
  pages = {115930},
  issn = {0045-7825},
  doi = {10.1016/j.cma.2023.115930},
  urldate = {2023-04-19},
  langid = {english}
}

@article{Klein2021,
  title = {Polyconvex Anisotropic Hyperelasticity with Neural Networks},
  author = {Klein, Dominik K. and Fern{\'a}ndez, Mauricio and Martin, Robert J. and Neff, Patrizio and Weeger, Oliver},
  year = {2021},
  journal = {Journal of the Mechanics and Physics of Solids},
  pages = {104703},
  issn = {00225096},
  doi = {10.1016/j.jmps.2021.104703},
  urldate = {2021-11-23},
  langid = {english}
}

@article{Thakolkaran2022,
  title = {{{NN-EUCLID}}: {{Deep-learning}} Hyperelasticity without Stress Data},
  shorttitle = {{{NN-EUCLID}}},
  author = {Thakolkaran, Prakash and Joshi, Akshay and Zheng, Yiwen and Flaschel, Moritz and De Lorenzis, Laura and Kumar, Siddhant},
  year = {2022},
  journal = {Journal of the Mechanics and Physics of Solids},
  volume = {169},
  pages = {105076},
  issn = {0022-5096},
  doi = {10.1016/j.jmps.2022.105076},
  urldate = {2022-10-15},
  langid = {english}
}

@article{Fuhg2022b,
  title = {Learning Hyperelastic Anisotropy from Data via a Tensor Basis Neural Network},
  author = {Fuhg, Jan N. and Bouklas, Nikolaos and Jones, Reese E.},
  year = {2022},
  journal = {Journal of the Mechanics and Physics of Solids},
  volume = {168},
  pages = {105022},
  issn = {00225096},
  doi = {10.1016/j.jmps.2022.105022},
  urldate = {2022-10-14}
}

@article{Tac2024a,
  title = {Benchmarking Physics-Informed Frameworks for Data-Driven Hyperelasticity},
  author = {Ta{\c c}, Vahidullah and Linka, Kevin and {Sahli-Costabal}, Francisco and Kuhl, Ellen and Tepole, Adrian Buganza},
  year = {2024},
  month = jan,
  journal = {Computational Mechanics},
  volume = {73},
  number = {1},
  pages = {49--65},
  issn = {1432-0924},
  doi = {10.1007/s00466-023-02355-2},
  urldate = {2024-09-12},
  langid = {english}
}

@article{Bahmani2024,
  title = {Physics-Constrained Symbolic Model Discovery for Polyconvex Incompressible Hyperelastic Materials},
  author = {Bahmani, Bahador and Sun, WaiChing},
  year = {2024},
  journal = {International Journal for Numerical Methods in Engineering},
  volume = {n/a},
  number = {n/a},
  pages = {e7473},
  issn = {1097-0207},
  doi = {10.1002/nme.7473},
  urldate = {2024-06-07},
  copyright = {{\copyright} 2024 John Wiley \& Sons, Ltd.},
  langid = {english}
}

@article{Benady2024,
  title = {{{NN-mCRE}}: {{A}} Modified Constitutive Relation Error Framework for Unsupervised Learning of Nonlinear State Laws with Physics-Augmented Neural Networks},
  shorttitle = {{{NN-mCRE}}},
  author = {Benady, Antoine and Baranger, Emmanuel and Chamoin, Ludovic},
  year = {2024},
  journal = {International Journal for Numerical Methods in Engineering},
  volume = {125},
  number = {8},
  pages = {e7439},
  issn = {1097-0207},
  doi = {10.1002/nme.7439},
  urldate = {2024-06-17},
  langid = {english}
}

@article{Peirlinck2024,
  title = {On Automated Model Discovery and a Universal Material Subroutine for Hyperelastic Materials},
  author = {Peirlinck, Mathias and Linka, Kevin and Hurtado, Juan A. and Kuhl, Ellen},
  year = {2024},
  month = jan,
  journal = {Computer Methods in Applied Mechanics and Engineering},
  volume = {418},
  pages = {116534},
  issn = {0045-7825},
  doi = {10.1016/j.cma.2023.116534},
  urldate = {2024-07-26}
}

@inproceedings{Czarnecki2017,
  title = {Sobolev {{Training}} for {{Neural Networks}}},
  booktitle = {Advances in {{Neural Information Processing Systems}}},
  author = {Czarnecki, Wojciech M and Osindero, Simon and Jaderberg, Max and Swirszcz, Grzegorz and Pascanu, Razvan},
  year = {2017},
  pages = {4278--4287},
  langid = {english}
}

@article{Vlassis2020,
  title = {Geometric Deep Learning for Computational Mechanics {{Part I}}: Anisotropic Hyperelasticity},
  shorttitle = {Geometric Deep Learning for Computational Mechanics {{Part I}}},
  author = {Vlassis, Nikolaos N. and Ma, Ran and Sun, WaiChing},
  year = {2020},
  journal = {Computer Methods in Applied Mechanics and Engineering},
  volume = {371},
  pages = {113299},
  issn = {0045-7825},
  doi = {10.1016/j.cma.2020.113299},
  urldate = {2021-07-08},
  langid = {english}
}

@article{Tac2022a,
  title = {Data-Driven Tissue Mechanics with Polyconvex Neural Ordinary Differential Equations},
  author = {Tac, Vahidullah and Sahli Costabal, Francisco and Tepole, Adrian B.},
  year = {2022},
  journal = {Computer Methods in Applied Mechanics and Engineering},
  volume = {398},
  pages = {115248},
  issn = {0045-7825},
  doi = {10.1016/j.cma.2022.115248},
  urldate = {2022-10-14},
  langid = {english}
}

@article{Chen2022,
  title = {Polyconvex Neural Networks for Hyperelastic Constitutive Models: {{A}} Rectification Approach},
  shorttitle = {Polyconvex Neural Networks for Hyperelastic Constitutive Models},
  author = {Chen, Peiyi and Guilleminot, Johann},
  year = {2022},
  journal = {Mechanics Research Communications},
  volume = {125},
  pages = {103993},
  issn = {00936413},
  doi = {10.1016/j.mechrescom.2022.103993},
  urldate = {2022-12-22},
  langid = {english}
}

@article{Jadoon2025a,
  title = {Inverse Design of Anisotropic Microstructures Using Physics-Augmented Neural Networks},
  author = {Jadoon, Asghar A. and Kalina, Karl A. and Rausch, Manuel K. and Jones, Reese and Fuhg, Jan Niklas},
  year = {2025},
  month = oct,
  journal = {Journal of the Mechanics and Physics of Solids},
  volume = {203},
  pages = {106161},
  issn = {0022-5096},
  doi = {10.1016/j.jmps.2025.106161},
  urldate = {2025-10-26}
}

@article{Dammass2025b,
  title = {Neural Networks Meet Phase-Field: {{A}} Hybrid Fracture Model},
  shorttitle = {Neural Networks Meet Phase-Field},
  author = {Damma{\ss}, Franz and Kalina, Karl A. and K{\"a}stner, Markus},
  year = {2025},
  month = may,
  journal = {Computer Methods in Applied Mechanics and Engineering},
  volume = {440},
  pages = {117937},
  issn = {0045-7825},
  doi = {10.1016/j.cma.2025.117937},
  urldate = {2025-04-02}
}

@article{Vijayakumaran2025,
  title = {Consistent Machine Learning for Topology Optimization with Microstructure-Dependent Neural Network Material Models},
  author = {Vijayakumaran, Harikrishnan and Russ, Jonathan B. and Paulino, Glaucio H. and Bessa, Miguel A.},
  year = {2025},
  month = mar,
  journal = {Journal of the Mechanics and Physics of Solids},
  volume = {196},
  pages = {106015},
  issn = {0022-5096},
  doi = {10.1016/j.jmps.2024.106015},
  urldate = {2025-02-03}
}

@article{Geuken2025a,
  title = {A Novel Neural Network for Isotropic Polyconvex Hyperelasticity Satisfying the Universal Approximation Theorem},
  author = {Geuken, Gian-Luca and Kurzeja, Patrick and Wiedemann, David and Mosler, J{\"o}rn},
  year = {2025},
  month = oct,
  journal = {Journal of the Mechanics and Physics of Solids},
  volume = {203},
  pages = {106209},
  issn = {0022-5096},
  doi = {10.1016/j.jmps.2025.106209},
  urldate = {2025-10-26}
}

@article{Kalina2024,
  title = {Neural Network-Based Multiscale Modeling of Finite Strain Magneto-Elasticity with Relaxed Convexity Criteria},
  author = {Kalina, Karl A. and Gebhart, Philipp and Brummund, J{\"o}rg and Linden, Lennart and Sun, WaiChing and K{\"a}stner, Markus},
  year = {2024},
  month = mar,
  journal = {Computer Methods in Applied Mechanics and Engineering},
  volume = {421},
  pages = {116739},
  issn = {00457825},
  doi = {10.1016/j.cma.2023.116739},
  urldate = {2024-02-21},
  langid = {english}
}

@book{Schroder2010,
	address = {Vienna},
	series = {{CISM} {International} {Centre} for {Mechanical} {Sciences}},
	title = {Poly-, {Quasi}- and {Rank}-{One} {Convexity} in {Applied} {Mechanics}},
	volume = {516},
	copyright = {http://www.springer.com/tdm},
	isbn = {978-3-7091-0173-5 978-3-7091-0174-2},
	url = {http://link.springer.com/10.1007/978-3-7091-0174-2},
	doi = {10.1007/978-3-7091-0174-2},
	urldate = {2025-08-03},
	publisher = {Springer},
	editor = {Schröder, Jörg and Neff, Patrizio and Maier, Giulio and Rammerstorfer, Franz G. and Salençon, Jean and Schrefler, Bernhard and Serafini, Paolo},
	year = {2010},
}

@article{Masi2024,
  title = {Neural Integration for Constitutive Equations Using Small Data},
  author = {Masi, Filippo and Einav, Itai},
  year = {2024},
  month = feb,
  journal = {Computer Methods in Applied Mechanics and Engineering},
  volume = {420},
  pages = {116698},
  issn = {0045-7825},
  doi = {10.1016/j.cma.2023.116698},
  urldate = {2024-06-06}
}

@article{He2022,
  title = {Thermodynamically Consistent Machine-Learned Internal State Variable Approach for Data-Driven Modeling of Path-Dependent Materials},
  author = {He, Xiaolong and Chen, Jiun-Shyan},
  year = {2022},
  month = jul,
  journal = {Computer Methods in Applied Mechanics and Engineering},
  pages = {115348},
  issn = {00457825},
  doi = {10.1016/j.cma.2022.115348},
  urldate = {2022-08-22},
  langid = {english}
}

@article{Settgast2020,
  title = {A Hybrid Approach to Simulate the Homogenized Irreversible Elastic-Plastic Deformations and Damage of Foams by Neural Networks},
  author = {Settgast, Christoph and H{\"u}tter, Geralf and Kuna, Meinhard and Abendroth, Martin},
  year = {2020},
  month = mar,
  journal = {International Journal of Plasticity},
  volume = {126},
  pages = {102624},
  issn = {07496419},
  doi = {10.1016/j.ijplas.2019.11.003},
  urldate = {2020-08-28},
  langid = {english}
}

@article{Malik2021,
  title = {A {{Hybrid Approach Employing Neural Networks}} to {{Simulate}} the {{Elasto-Plastic Deformation Behavior}} of {{3D-Foam Structures}}},
  author = {Malik, Alexander and Abendroth, Martin and H{\"u}tter, Geralf and Kiefer, Bjoern},
  year = {2021},
  journal = {Advanced Engineering Materials},
  volume = {n/a},
  number = {n/a},
  pages = {2100641},
  issn = {1527-2648},
  doi = {10.1002/adem.202100641},
  urldate = {2022-01-11},
  langid = {english}
}

@article{Vlassis2021,
  title = {Sobolev Training of Thermodynamic-Informed Neural Networks for Interpretable Elasto-Plasticity Models with Level Set Hardening},
  author = {Vlassis, Nikolaos N. and Sun, WaiChing},
  year = {2021},
  journal = {Computer Methods in Applied Mechanics and Engineering},
  volume = {377},
  pages = {113695},
  issn = {00457825},
  doi = {10.1016/j.cma.2021.113695},
  urldate = {2021-10-12},
  langid = {english}
}

@article{Vlassis2021b,
  title = {Component-Based Machine Learning Paradigm for Discovering Rate-Dependent and Pressure-Sensitive Level-Set Plasticity Models},
  author = {Vlassis, Nikolaos Napoleon and Sun, Waiching},
  year = {2021},
  month = oct,
  journal = {Journal of Applied Mechanics},
  pages = {1--13},
  issn = {0021-8936, 1528-9036},
  doi = {10.1115/1.4052684},
  urldate = {2021-11-18},
  langid = {english}
}

@misc{Boes2024,
  title = {Accounting for Plasticity: {{An}} Extension of Inelastic {{Constitutive Artificial Neural Networks}}},
  shorttitle = {Accounting for Plasticity},
  author = {Boes, Birte and Simon, Jaan-Willem and Holthusen, Hagen},
  year = {2024},
  month = jul,
  number = {arXiv:2407.19326},
  eprint = {2407.19326},
  primaryclass = {cs},
  publisher = {arXiv},
  doi = {10.48550/arXiv.2407.19326},
  urldate = {2024-09-12},
  archiveprefix = {arXiv}
}

@article{Jadoon2025,
  title = {Automated Model Discovery of Finite Strain Elastoplasticity from Uniaxial Experiments},
  author = {Jadoon, Asghar Arshad and Meyer, Knut Andreas and Fuhg, Jan Niklas},
  year = {2025},
  month = feb,
  journal = {Computer Methods in Applied Mechanics and Engineering},
  volume = {435},
  pages = {117653},
  issn = {0045-7825},
  doi = {10.1016/j.cma.2024.117653},
  urldate = {2025-03-28}
}

@article{vanderVelden2026,
  title = {A Note on Constitutive Artificial Neural Networks for Finite Strain Plasticity},
  author = {{van der Velden}, Tim and Boes, Birte and Brepols, Tim and Kuhl, Ellen and Holthusen, Hagen},
  year = {2026},
  month = aug,
  journal = {Mechanics Research Communications},
  volume = {155},
  pages = {104707},
  issn = {0093-6413},
  doi = {10.1016/j.mechrescom.2026.104707},
  urldate = {2026-05-24}
}

@article{Huang2022,
  title = {Variational {{Onsager Neural Networks}} ({{VONNs}}): {{A}} Thermodynamics-Based Variational Learning Strategy for Non-Equilibrium {{PDEs}}},
  shorttitle = {Variational {{Onsager Neural Networks}} ({{VONNs}})},
  author = {Huang, Shenglin and He, Zequn and Chem, Bryan and Reina, Celia},
  year = {2022},
  month = jun,
  journal = {Journal of the Mechanics and Physics of Solids},
  volume = {163},
  pages = {104856},
  issn = {0022-5096},
  doi = {10.1016/j.jmps.2022.104856},
  urldate = {2022-08-22},
  langid = {english}
}

@article{Halphen1975,
  title = {On {{Generalized Standard Materials}}.},
  shorttitle = {{{SUR LES MATERIAUX STANDARDS GENERALISES}}.},
  author = {Halphen, Bernard and Nguyen, Quoc Son},
  year = {1975},
  journal = {Journal de M{\'e}canique},
  volume = {14},
  number = {1},
  pages = {39--63},
  langid = {english}
}

@article{Coleman1963,
  title = {The Thermodynamics of Elastic Materials with Heat Conduction and Viscosity},
  author = {Coleman, Bernard D. and Noll, Walter},
  year = {1963},
  journal = {Archive for Rational Mechanics and Analysis},
  volume = {13},
  number = {1},
  pages = {167--178},
  issn = {0003-9527}
}

@article{Coleman1967,
  title = {Thermodynamics with {{Internal State Variables}}},
  author = {Coleman, Bernard D. and Gurtin, Morton E.},
  year = {1967},
  month = jul,
  journal = {The Journal of Chemical Physics},
  volume = {47},
  number = {2},
  pages = {597--613},
  issn = {0021-9606, 1089-7690},
  doi = {10.1063/1.1711937},
  urldate = {2020-01-22},
  langid = {english}
}

@book{Biot1965,
  title = {Mechanics of {{Incremental Deformations}}},
  author = {Biot, M. A.},
  year = {1965},
  publisher = {John Wiley \& Sons},
  address = {New York}
}

@article{Tac2023,
  title = {Data-Driven Anisotropic Finite Viscoelasticity Using Neural Ordinary Differential Equations},
  author = {Ta{\c c}, Vahidullah and Rausch, Manuel K. and Sahli Costabal, Francisco and Tepole, Adrian Buganza},
  year = {2023},
  month = jun,
  journal = {Computer Methods in Applied Mechanics and Engineering},
  volume = {411},
  pages = {116046},
  issn = {0045-7825},
  doi = {10.1016/j.cma.2023.116046},
  urldate = {2023-04-28},
  langid = {english}
}

@article{Asad2026,
  title = {A Staggered Training Framework for Mechanics-Informed Neural Networks in Tractable Multiscale Homogenization with Application to Woven Fabrics},
  author = {As'ad, Faisal and Farhat, Charbel},
  year = {2026},
  month = apr,
  journal = {Computer Methods in Applied Mechanics and Engineering},
  volume = {452},
  pages = {118666},
  issn = {00457825},
  doi = {10.1016/j.cma.2025.118666},
  urldate = {2026-01-23},
  langid = {english}
}

@article{Upadhyay2026,
  title = {Physics-Informed Data-Driven Discovery of Constitutive Models with Application to Strain-Rate-Sensitive Soft Materials},
  author = {Upadhyay, Kshitiz and Fuhg, Jan N. and Bouklas, Nikolaos and Ramesh, K. T.},
  year = {2026},
  month = feb,
  journal = {Computational Mechanics},
  volume = {77},
  number = {2},
  pages = {357--386},
  issn = {1432-0924},
  doi = {10.1007/s00466-024-02497-x},
  urldate = {2026-06-02},
  langid = {english}
}

@article{Abdolazizi2026,
  title = {Thermodynamically Consistent Viscoelastic Constitutive Artificial Neural Networks: {{Automating}} the Pipeline from Experimental Data to Finite Element Simulations},
  shorttitle = {Thermodynamically Consistent Viscoelastic Constitutive Artificial Neural Networks},
  author = {Abdolazizi, Kian P. and Aydin, Roland C. and Cyron, Christian J. and Linka, Kevin},
  year = {2026},
  month = oct,
  journal = {Computer Methods in Applied Mechanics and Engineering},
  volume = {460},
  pages = {119080},
  issn = {0045-7825},
  doi = {10.1016/j.cma.2026.119080},
  urldate = {2026-06-02}
}

@article{Abdolazizi2023a,
  title = {Viscoelastic {{Constitutive Artificial Neural Networks}} ({{vCANNs}}) -- a Framework for Data-Driven Anisotropic Nonlinear Finite Viscoelasticity},
  author = {Abdolazizi, Kian P. and Linka, Kevin and Cyron, Christian J.},
  year = {2023},
  month = dec,
  journal = {Journal of Computational Physics},
  pages = {112704},
  issn = {0021-9991},
  doi = {10.1016/j.jcp.2023.112704},
  urldate = {2023-12-13}
}

@article{Liu2021b,
  title = {A Continuum and Computational Framework for Viscoelastodynamics: {{I}}. {{Finite}} Deformation Linear Models},
  shorttitle = {A Continuum and Computational Framework for Viscoelastodynamics},
  author = {Liu, Ju and Latorre, Marcos and Marsden, Alison L.},
  year = {2021},
  month = nov,
  journal = {Computer Methods in Applied Mechanics and Engineering},
  volume = {385},
  pages = {114059},
  issn = {0045-7825},
  doi = {10.1016/j.cma.2021.114059},
  urldate = {2026-06-02}
}

@article{Califano2026,
  title = {Enhancing Nonlinear Viscoelastic Modeling of Elastomers through Neural Networks: {{A}} Deep Rheological Element},
  shorttitle = {Enhancing Nonlinear Viscoelastic Modeling of Elastomers through Neural Networks},
  author = {Califano, Federico and Ciambella, Jacopo},
  year = {2026},
  month = jan,
  journal = {Mechanics of Materials},
  volume = {212},
  pages = {105525},
  issn = {0167-6636},
  doi = {10.1016/j.mechmat.2025.105525},
  urldate = {2025-10-28}
}

@article{Geiger2025,
  title = {Multiscale Modeling of Viscoelastic Shell Structures with Artificial Neural Networks},
  author = {Geiger, Jeremy and Wagner, Werner and Freitag, Steffen},
  year = {2025},
  month = mar,
  journal = {Computational Mechanics},
  issn = {1432-0924},
  doi = {10.1007/s00466-025-02613-5},
  urldate = {2025-08-03},
  langid = {english}
}

@article{Friedrichs2026,
  title = {Precise, Efficient and Flexible Modeling of Crystallizing Elastomers Based on Physics-Augmented Neural Networks},
  author = {Friedrichs, Konrad and Damma{\ss}, Franz and Kalina, Karl A. and K{\"a}stner, Markus},
  year = {2026},
  month = jun,
  journal = {Computer Methods in Applied Mechanics and Engineering},
  volume = {455},
  pages = {118852},
  issn = {0045-7825},
  doi = {10.1016/j.cma.2026.118852},
  urldate = {2026-04-08}
}

@article{Pierron2023,
  title = {Material {{Testing}} 2.0: {{A}} Brief Review},
  shorttitle = {Material {{Testing}} 2.0},
  author = {Pierron, Fabrice},
  year = {2023},
  journal = {Strain},
  volume = {59},
  number = {3},
  pages = {e12434},
  issn = {1475-1305},
  doi = {10.1111/str.12434},
  urldate = {2026-07-15},
  copyright = {{\copyright} 2023 The Author. Strain published by John Wiley \& Sons Ltd.},
  langid = {english}
}

@article{Pierron2020,
  title = {Towards {{Material Testing}} 2.0. {{A}} Review of Test Design for Identification of Constitutive Parameters from Full-field Measurements},
  author = {Pierron, F. and Gr{\'e}diac, M.},
  year = {2020},
  month = sep,
  journal = {Strain},
  issn = {0039-2103, 1475-1305},
  doi = {10.1111/str.12370},
  urldate = {2020-11-27},
  langid = {english}
}

@article{Lenoir2007,
  title = {Volumetric {{Digital Image Correlation Applied}} to {{X-ray Microtomography Images}} from {{Triaxial Compression Tests}} on {{Argillaceous Rock}}},
  author = {Lenoir, N. and Bornert, M. and Desrues, J. and B{\'e}suelle, P. and Viggiani, G.},
  year = {2007},
  month = aug,
  journal = {Strain},
  volume = {43},
  number = {3},
  pages = {193--205},
  issn = {00392103, 14751305},
  doi = {10.1111/j.1475-1305.2007.00348.x},
  urldate = {2022-01-15},
  langid = {english}
}

@article{Avril2008,
  title = {Overview of {{Identification Methods}} of {{Mechanical Parameters Based}} on {{Full-field Measurements}}},
  author = {Avril, St{\'e}phane and Bonnet, Marc and Bretelle, Anne-Sophie and Gr{\'e}diac, Michel and Hild, Fran{\c c}ois and Ienny, Patrick and Latourte, F{\'e}lix and Lemosse, Didier and Pagano, St{\'e}phane and Pagnacco, Emmanuel and Pierron, Fabrice},
  year = {2008},
  month = aug,
  journal = {Experimental Mechanics},
  volume = {48},
  number = {4},
  pages = {381--402},
  issn = {0014-4851, 1741-2765},
  doi = {10.1007/s11340-008-9148-y},
  urldate = {2022-01-15},
  langid = {english}
}

@article{Roux2020,
  title = {Optimal Procedure for the Identification of Constitutive Parameters from Experimentally Measured Displacement Fields},
  author = {Roux, St{\'e}phane and Hild, Fran{\c c}ois},
  year = {2020},
  month = feb,
  journal = {International Journal of Solids and Structures},
  series = {Physics and {{Mechanics}} of {{Random Structures}}: {{From Morphology}} to {{Material Properties}}},
  volume = {184},
  pages = {14--23},
  issn = {0020-7683},
  doi = {10.1016/j.ijsolstr.2018.11.008},
  urldate = {2025-02-21}
}

@article{Romer2024,
  title = {Reduced and {{All-at-Once Approaches}} for {{Model Calibration}} and {{Discovery}} in {{Computational Solid Mechanics}}},
  author = {R{\"o}mer, Ulrich and Hartmann, Stefan and Tr{\"o}ger, Jendrik-Alexander and Anton, David and Wessels, Henning and Flaschel, Moritz and De Lorenzis, Laura},
  year = {2024},
  month = aug,
  journal = {Applied Mechanics Reviews},
  pages = {1--51},
  issn = {0003-6900},
  doi = {10.1115/1.4066118},
  urldate = {2024-12-20}
}

@article{Chen2025,
  title = {Finite {{Element Model Updating}} for {{Material Model Calibration}}: {{A Review}} and {{Guide}} to {{Practice}}},
  shorttitle = {Finite {{Element Model Updating}} for {{Material Model Calibration}}},
  author = {Chen, Bin and Starman, Bojan and Halilovi{\v c}, Miroslav and Berglund, Lars A. and Coppieters, Sam},
  year = {2025},
  month = may,
  journal = {Archives of Computational Methods in Engineering},
  volume = {32},
  number = {4},
  pages = {2035--2112},
  issn = {1886-1784},
  doi = {10.1007/s11831-024-10200-9},
  urldate = {2026-07-15},
  langid = {english}
}

@article{Mahnken1996,
  title = {A Unified Approach for Parameter Identification of Inelastic Material Models in the Frame of the Finite Element Method},
  author = {Mahnken, Rolf and Stein, Erwin},
  year = {1996},
  month = sep,
  journal = {Computer Methods in Applied Mechanics and Engineering},
  volume = {136},
  number = {3},
  pages = {225--258},
  issn = {0045-7825},
  doi = {10.1016/0045-7825(96)00991-7},
  urldate = {2024-12-20}
}

@article{Sakaridis2024,
  title = {Post-{{Necking}} Full-Field {{FEMU}} Identification of Anisotropic Plasticity from Flat Notched Tension Experiments},
  author = {Sakaridis, Emmanouil and Roth, Christian C. and Jordan, Benoit and Mohr, Dirk},
  year = {2024},
  month = dec,
  journal = {International Journal of Solids and Structures},
  volume = {305},
  pages = {113076},
  issn = {0020-7683},
  doi = {10.1016/j.ijsolstr.2024.113076},
  urldate = {2026-07-15}
}

@article{Grediac1989,
  title = {Principe Des Travaux Virtuels et Identification},
  author = {Gr{\'e}diac, Michel},
  year = {1989},
  journal = {Comptes rendus de l'Acad{\'e}mie des sciences. S{\'e}rie 2, M{\'e}canique, Physique, Chimie, Sciences de l'univers, Sciences de la Terre},
  volume = {309},
  number = {1},
  pages = {1--5}
}

@article{Claire2004,
  title = {A Finite Element Formulation to Identify Damage Fields: The Equilibrium Gap Method},
  shorttitle = {A Finite Element Formulation to Identify Damage Fields},
  author = {Claire, D. and Hild, F. and Roux, S.},
  year = {2004},
  journal = {International Journal for Numerical Methods in Engineering},
  volume = {61},
  number = {2},
  pages = {189--208},
  issn = {1097-0207},
  doi = {10.1002/nme.1057},
  urldate = {2026-04-10},
  copyright = {Copyright {\copyright} 2004 John Wiley \& Sons, Ltd.},
  langid = {english}
}

@article{Flaschel2023,
  title = {Automated Discovery of Generalized Standard Material Models with {{EUCLID}}},
  author = {Flaschel, Moritz and Kumar, Siddhant and De Lorenzis, Laura},
  year = {2023},
  month = feb,
  journal = {Computer Methods in Applied Mechanics and Engineering},
  volume = {405},
  pages = {115867},
  issn = {0045-7825},
  doi = {10.1016/j.cma.2022.115867},
  urldate = {2023-04-28},
  langid = {english}
}

@article{Abbasi2026a,
  title = {Discovery of {{Hyperelastic Constitutive Laws}} from {{Experimental Data}} with {{EUCLID}}},
  author = {Abbasi, A. and Ricci, M. and Carrara, P. and Flaschel, M. and Kumar, S. and Marfia, S. and De Lorenzis, L.},
  year = {2026},
  month = jun,
  journal = {Experimental Mechanics},
  volume = {66},
  number = {5},
  pages = {877--908},
  issn = {1741-2765},
  doi = {10.1007/s11340-026-01290-6},
  urldate = {2026-07-15},
  langid = {english}
}

@misc{Alheit2026,
  title = {{{CANN-EUCLID}}: Unsupervised Constitutive Artificial Neural Network Model Discovery from Full-Field Data},
  shorttitle = {{{CANN-EUCLID}}},
  author = {Alheit, Benjamin and Kumar, Siddhant and Peirlinck, Mathias},
  year = {2026},
  month = jun,
  number = {arXiv:2606.14565},
  eprint = {2606.14565},
  primaryclass = {cs.CE},
  publisher = {arXiv},
  doi = {10.48550/arXiv.2606.14565},
  urldate = {2026-07-08},
  archiveprefix = {arXiv}
}

@article{Linka2023,
  title = {A New Family of {{Constitutive Artificial Neural Networks}} towards Automated Model Discovery},
  author = {Linka, Kevin and Kuhl, Ellen},
  year = {2023},
  journal = {Computer Methods in Applied Mechanics and Engineering},
  volume = {403},
  pages = {115731},
  issn = {00457825},
  doi = {10.1016/j.cma.2022.115731},
  urldate = {2023-08-01},
  langid = {english}
}

@misc{Moon2026,
  title = {Physics-{{Informed Discovery}} of {{Yield Functions}} in {{Plasticity}} via {{Convex Neural Representations}}},
  author = {Moon, Hyeonbin and Cho, Donghyuk and Yu, Jecheon and Yoon, Jeong Whan and Ryu, Seunghwa},
  year = {2026},
  month = jun,
  number = {arXiv:2606.19375},
  eprint = {2606.19375},
  primaryclass = {cs.LG},
  publisher = {arXiv},
  doi = {10.48550/arXiv.2606.19375},
  urldate = {2026-06-19},
  archiveprefix = {arXiv}
}

@article{Thakolkaran2025,
  title = {Can {{KAN CANs}}? {{Input-convex Kolmogorov-Arnold Networks}} ({{KANs}}) as Hyperelastic Constitutive Artificial Neural Networks ({{CANs}})},
  shorttitle = {Can {{KAN CANs}}?},
  author = {Thakolkaran, Prakash and Guo, Yaqi and Saini, Shivam and Peirlinck, Mathias and Alheit, Benjamin and Kumar, Siddhant},
  year = {2025},
  month = aug,
  journal = {Computer Methods in Applied Mechanics and Engineering},
  volume = {443},
  pages = {118089},
  issn = {00457825},
  doi = {10.1016/j.cma.2025.118089},
  urldate = {2026-04-02},
  langid = {english}
}

@article{Meng2025,
  title = {Machine-Learning-Based Virtual Fields Method: {{Application}} to Anisotropic Hyperelasticity},
  shorttitle = {Machine-Learning-Based Virtual Fields Method},
  author = {Meng, Shuangshuang and Yousefi, Ali Akbar Karkhaneh and Avril, St{\'e}phane},
  year = {2025},
  month = feb,
  journal = {Computer Methods in Applied Mechanics and Engineering},
  volume = {434},
  pages = {117580},
  issn = {0045-7825},
  doi = {10.1016/j.cma.2024.117580},
  urldate = {2025-01-23}
}

@article{Bourdyot2026,
  title = {Learning a Hyperelastic Constitutive Model from {{3D}} Experimental Data},
  author = {Bourdyot, M. and Compans, M. and Langlois, R. and Smaniotto, B. and Baranger, E. and Jailin, C.},
  year = {2026},
  month = mar,
  journal = {Computer Methods in Applied Mechanics and Engineering},
  volume = {450},
  pages = {118592},
  issn = {0045-7825},
  doi = {10.1016/j.cma.2025.118592},
  urldate = {2026-04-08}
}

@article{Ferreira2026,
  title = {Automatically {{Differentiable Model Updating}} ({{ADiMU}}): {{Conventional}}, Hybrid, and Neural Network Material Model Discovery Including History-Dependency},
  shorttitle = {Automatically {{Differentiable Model Updating}} ({{ADiMU}})},
  author = {Ferreira, Bernardo P. and Bessa, Miguel A.},
  year = {2026},
  month = jan,
  journal = {Journal of the Mechanics and Physics of Solids},
  volume = {206},
  pages = {106408},
  issn = {0022-5096},
  doi = {10.1016/j.jmps.2025.106408},
  urldate = {2026-04-08}
}

@article{Lourenco2024,
  title = {An Indirect Training Approach for Implicit Constitutive Modelling Using Recurrent Neural Networks and the Virtual Fields Method},
  author = {Louren{\c c}o, R{\'u}ben and Georgieva, Petia and Cueto, Elias and {Andrade-Campos}, A.},
  year = {2024},
  month = may,
  journal = {Computer Methods in Applied Mechanics and Engineering},
  volume = {425},
  pages = {116961},
  issn = {0045-7825},
  doi = {10.1016/j.cma.2024.116961},
  urldate = {2024-04-04}
}

@article{Li2026,
  title = {Learning Anisotropic Hyperelasticity with an Unsupervised Symmetry-Aware Equilibrium-Based Neural Network},
  author = {Li, Shun and Li, Lingfeng and Chen, Chang Qing},
  year = {2026},
  month = jun,
  journal = {Computer Methods in Applied Mechanics and Engineering},
  volume = {455},
  pages = {118911},
  issn = {0045-7825},
  doi = {10.1016/j.cma.2026.118911},
  urldate = {2026-04-11}
}

@article{Shi2025a,
  title = {Deep Learning without Stress Data on the Discovery of Multi-Regional Hyperelastic Properties},
  author = {Shi, Ruike and Yang, Haitian and Chen, Jianxu and Hackl, Klaus and Avril, St{\'e}phane and He, Yiqian},
  year = {2025},
  month = jul,
  journal = {Computational Mechanics},
  volume = {76},
  number = {1},
  pages = {117--146},
  issn = {1432-0924},
  doi = {10.1007/s00466-024-02591-0},
  urldate = {2026-05-27},
  langid = {english}
}

@article{Tac2026,
  title = {Fully Data-Driven Inverse Characterization of Heterogeneous Materials with Hyper-Network Neural {{ODEs}}},
  author = {Ta{\c c}, Vahidullah and {Amiri-Hezaveh}, Amirhossein and Bechtel, Grace N. and Loftin, Titus and Rausch, Manuel K. and Sahli Costabal, Francisco and Tepole, Adrian Buganza},
  year = {2026},
  month = mar,
  journal = {npj Computational Materials},
  publisher = {Nature Publishing Group},
  issn = {2057-3960},
  doi = {10.1038/s41524-026-02027-8},
  urldate = {2026-04-11},
  copyright = {2026 The Author(s)},
  langid = {english}
}

@article{Chaurasiya2026,
  title = {Hetero-{{EUCLID}}: {{Interpretable}} Model Discovery for Heterogeneous Hyperelastic Materials Using Stress-Unsupervised Learning},
  shorttitle = {Hetero-{{EUCLID}}},
  author = {Chaurasiya, Kanhaiya Lal and Dutta, Saurav and Kumar, Siddhant and Joshi, Akshay},
  year = {2026},
  month = apr,
  journal = {Computer Methods in Applied Mechanics and Engineering},
  volume = {452},
  pages = {118729},
  issn = {0045-7825},
  doi = {10.1016/j.cma.2026.118729},
  urldate = {2026-05-05}
}

@article{Benady2024a,
  title = {Unsupervised Learning of History-Dependent Constitutive Material Laws with Thermodynamically-Consistent Neural Networks in the Modified {{Constitutive Relation Error}} Framework},
  author = {Benady, Antoine and Baranger, Emmanuel and Chamoin, Ludovic},
  year = {2024},
  month = may,
  journal = {Computer Methods in Applied Mechanics and Engineering},
  volume = {425},
  pages = {116967},
  issn = {0045-7825},
  doi = {10.1016/j.cma.2024.116967},
  urldate = {2024-12-20}
}

@article{Wu2025,
  title = {Learning the Physics-Consistent Material Behavior from Measurable Data via {{PDE-constrained}} Optimization},
  author = {Wu, Xinxin and Zhang, Yin and Mao, Sheng},
  year = {2025},
  month = mar,
  journal = {Computer Methods in Applied Mechanics and Engineering},
  volume = {437},
  pages = {117748},
  issn = {0045-7825},
  doi = {10.1016/j.cma.2025.117748},
  urldate = {2025-01-23}
}

@article{Gavris2025,
  title = {Discovering Neural Elastoplasticity from Kinematic Observations},
  author = {Gavris, Georgios Barkoulis and Sun, WaiChing},
  year = {2025},
  month = sep,
  journal = {Proceedings of the National Academy of Sciences},
  volume = {122},
  number = {38},
  pages = {e2508732122},
  publisher = {Proceedings of the National Academy of Sciences},
  doi = {10.1073/pnas.2508732122},
  urldate = {2026-04-08}
}

@article{Gavris2026,
  title = {Discovering Neural Cohesive Zone Laws from Displacement Fields},
  author = {Gavris, Georgios Barkoulis and Sun, WaiChing},
  year = {2026},
  month = apr,
  journal = {Computer Methods in Applied Mechanics and Engineering},
  volume = {452},
  pages = {118733},
  issn = {0045-7825},
  doi = {10.1016/j.cma.2026.118733},
  urldate = {2026-04-11}
}

@article{Seidl2022,
  title = {Calibration of Elastoplastic Constitutive Model Parameters from Full-Field Data with Automatic Differentiation-Based Sensitivities},
  author = {Seidl, D. Thomas and Granzow, Brian N.},
  year = {2022},
  journal = {International Journal for Numerical Methods in Engineering},
  volume = {123},
  number = {1},
  pages = {69--100},
  issn = {1097-0207},
  doi = {10.1002/nme.6843},
  urldate = {2026-04-08},
  copyright = {{\copyright} 2021 John Wiley \& Sons Ltd.},
  langid = {english}
}

@article{Kumar2025,
  title = {A Comparative Study of Calibration Techniques for Finite Strain Elastoplasticity: {{Numerically-exact}} Sensitivities for {{FEMU}} and {{VFM}}},
  shorttitle = {A Comparative Study of Calibration Techniques for Finite Strain Elastoplasticity},
  author = {Kumar, Sanjeev and Seidl, D. Thomas and Granzow, Brian N. and Yang, Jin and Fuhg, Jan Niklas},
  year = {2025},
  month = sep,
  journal = {Computer Methods in Applied Mechanics and Engineering},
  volume = {444},
  pages = {118159},
  issn = {0045-7825},
  doi = {10.1016/j.cma.2025.118159},
  urldate = {2026-04-08}
}

@article{Akerson2025,
  title = {Learning Constitutive Relations from Experiments: 1. {{PDE}} Constrained Optimization},
  shorttitle = {Learning Constitutive Relations from Experiments},
  author = {Akerson, Andrew and Rajan, Aakila and Bhattacharya, Kaushik},
  year = {2025},
  month = aug,
  journal = {Journal of the Mechanics and Physics of Solids},
  volume = {201},
  pages = {106128},
  issn = {0022-5096},
  doi = {10.1016/j.jmps.2025.106128},
  urldate = {2026-06-02}
}

@article{Flaschel2025,
  title = {Convex Neural Networks Learn Generalized Standard Material Models},
  author = {Flaschel, Moritz and Steinmann, Paul and De Lorenzis, Laura and Kuhl, Ellen},
  year = {2025},
  month = jul,
  journal = {Journal of the Mechanics and Physics of Solids},
  volume = {200},
  pages = {106103},
  issn = {0022-5096},
  doi = {10.1016/j.jmps.2025.106103},
  urldate = {2025-03-28}
}

@article{Shi2026,
  title = {A Thermodynamically Consistent Multi-Internal-Variable Neural Network Framework for Viscoelastic Constitutive Modeling},
  author = {Shi, Ziqiang and Luo, Xue},
  year = {2026},
  month = may,
  journal = {Computational Mechanics},
  issn = {1432-0924},
  doi = {10.1007/s00466-026-02793-8},
  urldate = {2026-05-16},
  langid = {english}
}

@article{Holthusen2024,
  title = {Theory and Implementation of Inelastic {{Constitutive Artificial Neural Networks}}},
  author = {Holthusen, Hagen and Lamm, Lukas and Brepols, Tim and Reese, Stefanie and Kuhl, Ellen},
  year = {2024},
  month = aug,
  journal = {Computer Methods in Applied Mechanics and Engineering},
  volume = {428},
  pages = {117063},
  issn = {0045-7825},
  doi = {10.1016/j.cma.2024.117063},
  urldate = {2024-06-17}
}

@article{Holthusen2024a,
  title = {Polyconvex Inelastic Constitutive Artificial Neural Networks},
  author = {Holthusen, Hagen and Lamm, Lukas and Brepols, Tim and Reese, Stefanie and Kuhl, Ellen},
  year = {2024},
  journal = {PAMM},
  volume = {24},
  number = {3},
  pages = {e202400032},
  issn = {1617-7061},
  doi = {10.1002/pamm.202400032},
  urldate = {2025-04-15},
  copyright = {{\copyright} 2024 The Author(s). Proceedings in Applied Mathematics \& Mechanics published by Wiley-VCH GmbH.},
  langid = {english}
}

@article{Holthusen2026,
  title = {A Generalized Dual Potential for Inelastic {{Constitutive Artificial Neural Networks}}: {{A JAX}} Implementation at Finite Strains},
  shorttitle = {A Generalized Dual Potential for Inelastic {{Constitutive Artificial Neural Networks}}},
  author = {Holthusen, Hagen and Linka, Kevin and Kuhl, Ellen and Brepols, Tim},
  year = {2026},
  month = jan,
  journal = {Journal of the Mechanics and Physics of Solids},
  volume = {206},
  pages = {106337},
  issn = {0022-5096},
  doi = {10.1016/j.jmps.2025.106337},
  urldate = {2025-10-26}
}

@article{Holthusen2026a,
  title = {A Complement to Neural Networks for Anisotropic Inelasticity at Finite Strains},
  author = {Holthusen, Hagen and Kuhl, Ellen},
  year = {2026},
  journal = {Computer Methods in Applied Mechanics and Engineering},
  volume = {450},
  pages = {118612},
  doi = {10.1016/j.cma.2025.118612},
  langid = {english}
}

@misc{Knipper2026,
  title = {Finite {{Element-Based Material Learning}} via {{Automatic Differentiation}}: {{Learning}} Constitutive Neural Network Models from Full-Field Deformation Data},
  shorttitle = {Finite {{Element-Based Material Learning}} via {{Automatic Differentiation}}},
  author = {Knipper, Matthias and Ji, Chenyi and Brand, Malte and Linka, Kevin},
  year = {2026},
  month = may,
  number = {arXiv:2606.05199},
  eprint = {2606.05199},
  primaryclass = {physics.comp-ph},
  publisher = {arXiv},
  doi = {10.48550/arXiv.2606.05199},
  urldate = {2026-07-08},
  archiveprefix = {arXiv}
}

@article{collange2015reproducibility,
  title   = {Numerical reproducibility for the parallel reduction on multi- and many-core architectures},
  author  = {Collange, Sylvain and Defour, David and Graillat, Stef and Iakymchuk, Roman},
  journal = {Parallel Computing},
  volume  = {49},
  pages   = {83--97},
  year    = {2015},
  publisher = {Elsevier}
}

@article{boehler_irreducible_1977,
	title = {On {Irreducible} {Representations} for {Isotropic} {Scalar} {Functions}},
	volume = {57},
	copyright = {Copyright © 1977 WILEY-VCH Verlag GmbH \& Co. KGaA, Weinheim},
	issn = {1521-4001},
	url = {https://onlinelibrary.wiley.com/doi/abs/10.1002/zamm.19770570608},
	doi = {10.1002/zamm.19770570608},
	language = {en},
	number = {6},
	urldate = {2025-01-04},
	journal = {ZAMM - Journal of Applied Mathematics and Mechanics / Zeitschrift für Angewandte Mathematik und Mechanik},
	author = {Boehler, J. P.},
	year = {1977},
	note = {\_eprint: https://onlinelibrary.wiley.com/doi/pdf/10.1002/zamm.19770570608},
	pages = {323--327},
}

@article{jailin_experimental_2024,
	title = {Experimental {Learning} of a {Hyperelastic} {Behavior} with a {Physics}-{Augmented} {Neural} {Network}},
	volume = {64},
	issn = {1741-2765},
	url = {https://doi.org/10.1007/s11340-024-01106-5},
	doi = {10.1007/s11340-024-01106-5},
	language = {en},
	number = {9},
	urldate = {2026-08-28},
	journal = {Experimental Mechanics},
	author = {Jailin, C. and Benady, A. and Legroux, R. and Baranger, E.},
	month = nov,
	year = {2024},
	pages = {1465--1481},
}

@article{knauf_narouie_unsupervised_2026,
	title = {Unsupervised {Constitutive} {Model} {Discovery} from {Sparse} and {Noisy} {Data}},
	volume = {452},
	issn = {0045-7825},
	url = {https://www.sciencedirect.com/science/article/pii/S0045782525009946},
	doi = {10.1016/j.cma.2025.118722},
	urldate = {2026-08-28},
	journal = {Computer Methods in Applied Mechanics and Engineering},
	author = {Knauf Narouie, Vahab and Urrea-Quintero, Jorge-Humberto and Cirak, Fehmi and Wessels, Henning},
	month = apr,
	year = {2026},
	pages = {118722},
}
\end{small}
\end{document}